\documentclass{amsart}

\usepackage{array}
\usepackage{mathrsfs}
\usepackage{euscript}
\usepackage{cancel}
\usepackage{commandes_GMM_preprint}

\begin{document}

\markboth{Graille Magin Massot}{Kinetic Theory of  Plasmas}

\title{KINETIC THEORY OF PLASMAS: TRANSLATIONAL ENERGY}
\date{\today}

\author{BENJAMIN GRAILLE}
\address{Laboratoire de Math\'ematiques d'Orsay -
UMR 8628 CNRS, Universit\'e Paris-Sud\\
91405, Orsay Cedex, France}
\email{benjamin.graille@math.u-psud.fr}

\author{THIERRY E. MAGIN$^{\star}$}
\address{
Center for Turbulence Research, Stanford University\\
488 Escondido Mall,  Stanford, California 94305, USA\\
also at Reacting Flow Environments Branch, NASA Ames Research Center, Moffett Field, California 94035, USA}
\email{magin@stanford.edu}

\author{MARC MASSOT}
\address{Laboratoire EM2C - UPR CNRS 288, Ecole Centrale Paris\\
Grande Voie des Vignes, 92295 Ch\^{a}tenay-Malabry Cedex, France}
\email{marc.massot@em2c.ecp.fr}

\begin{abstract}
In the present contribution, we derive from kinetic theory a unified fluid model for multicomponent plasmas by accounting for the electromagnetic field influence. We deal with a possible  thermal nonequilibrium of the translational energy of the particles, neglecting their internal energy and the reactive collisions. 
Given the strong disparity of mass between the electrons and heavy particles,  such as molecules, atoms, and ions, we conduct a dimensional analysis of the Boltzmann equation and introduce a scaling based on 
the square root of the ratio of the electron mass to a characteristic heavy-particle mass.
We then generalize the Chapman-Enskog method,  emphasizing the  role of a multiscale perturbation parameter on the collisional operator, the streaming operator, and the collisional invariants  of the Boltzmann equation. 
The system is examined at successive orders of approximation, each of which corresponding to a physical time scale. 
The multicomponent Navier-Stokes regime is reached for the heavy particles, which 
follow a hyperbolic scaling, and is coupled to first order drift-diffusion equations 
for the electrons, which follow a parabolic scaling.
The transport coefficients are then calculated in terms of bracket operators whose mathematical structure allows for positivity properties to be determined. They exhibit an anisotropic behavior when the magnetic field is strong enough.  We also give a complete description of the Kolesnikov effect, $i.e.$, the crossed contributions to the mass and energy transport fluxes coupling the electrons and heavy particles. Finally, the first and second principles of thermodynamics are proved to be satisfied  by deriving a total energy equation and an entropy equation. 
Moreover, the system of equations is shown to be conservative and the purely convective system hyperbolic, thus leading to a well defined structure.
\end{abstract}

\keywords{Kinetic theory; plasmas in thermal nonequilibrium; conservation equations; multicomponent transport properties.}

\subjclass[2000]{82C40, 76X05, 41A60}

\maketitle

\renewcommand{\thefootnote}{\fnsymbol{footnote}}
\footnotetext[1]{Corresponding author}

\section{Introduction}

Plasmas are ionized gas mixtures, either magnetized or not, that have many practical applications.
For instance, lightning is a well-known natural plasma and has been 
studied for many years~\cite{Bazelyan}.
A second application is encountered in hypersonic flows; when a spacecraft enters into a planetary atmosphere at hypervelocity, the gas temperature and pressure strongly rise through a shock wave, consequently, dissociation and ionization of the gas particles occur in the shock layer. Hypersonic flow conditions are reproduced in dedicated wind-tunnels such as plasmatrons, arc-jet facilities, and shock-tubes~\cite{park,gabitta,tirsky}. A third example was found about two decades ago, when large-scale electrical discharges were discovered in the mesosphere and lower ionosphere above large thunderstorms; these plasmas are now commonly referred to as sprites~\cite{anne,pasko}. Fourth, discharges at atmospheric pressure have received renewed attention in recent years due to their ability to enhance the reactivity of a variety of gas flows for applications ranging from surface treatment to flame stabilization and ignition (see~\cite{sergey,raizer,starikovskaia,vanveldhuizen} and references cited therein). Fifth, Hall thrusters are being developed to replace chemical systems for many on-orbit propulsion tasks on communications and exploration spacecraft~\cite{gerjan,iain}. Finally, two important applications of magnetized plasmas are the laboratory thermonuclear fusion~\cite{bobrova,schnack} and the magnetic reconnection phenomenom in astrophysics~\cite{yamada}.

Depending on the magnitude of the ratio of the reference particle mean free path to the system characteristic length (Knudsen number), two different approaches are generally followed to describe the transport of mass, momentum, and energy in a plasma~\cite{bird}: either a particle approach at high values of the Knudsen number  (solution to the Boltzmann equation using Monte Carlo methods), or a fluid approach at  low values (solution to macroscopic conservation equations by means of computational fluid dynamics methods).  In this work, we study plasmas that can be described by a fluid approach, thus encompassing  most of the above-mentioned applications. In this case, kinetic theory can be used to obtain the governing conservation equations and  transport fluxes of plasmas. Hence, closure of the problem is realized at the microscopic level by determining from experimental measurements either the potentials of interaction between the gas particles, or the cross-sections. 

A complete model of plasmas shall allow for the following physical phenomena to be described
 \begin{itemize}
\item Thermal non equilibrium of the translational energy,
\item Influence of the electromagnetic field,
\item Occurrence of reactive collisions,
\item Excitation of internal degrees of freedom.
\end{itemize}
So far, no such unified model has been derived  by means of kinetic theory. Besides,  a derivation of the mathematical structure of the conservation equations also appears to be crucial in the design of the  associated numerical methods.  Based on our previous work, we investigate in the present study the thermal nonequilibrium of the translational energy~\cite{magin1} and the  influence of the magnetic field~\cite{graille1}. We generalize the Chapman-Enskog method within the context of a dimensional analysis of the Boltzmann equation,  emphasizing the  role of a multiscale perturbation parameter on the collisional operator, the streaming operator, and the collisional invariants  of the Boltzmann equation. Then, we obtain macroscopic equations eventually leading to a sound entropy structure. Moreover, the system of equations is shown to be conservative and the purely convective system hyperbolic. Let us now describe in more detail how these issues are currently addressed  in the literature.

First, a multiscale analysis  is essential to resolve the Boltzmann equation governing the velocity distribution functions. We recall that a fluid can be described in the continuum limit provided that the Knudsen number is small. Besides, in the case of plasmas,  a thermal nonequilibrium may occur between the velocity distribution functions of the electrons and heavy particles (atoms, molecules, and ions), given the strong disparity of mass between both types of species. Ergo, the square root of the ratio of the electron mass to a characteristic heavy-particle mass represents an additional small parameter to be accounted for in the derivation of an asymptotic solution to the Boltzmann equation.  Literature abounds with expressions of the scaling for the perturbative solution method. For instance, significant results  are given in references~\cite{chmieleski,daybelge,devoto,anatoliy2,zhdanov}. Yet, Petit and Darrozes~\cite{petit} have suggested that the only sound scaling is obtained by means of a dimensional analysis of the Boltzmann equation. Subsequently, Degond and Lucquin~\cite{degond1,degond2} have established a formal theory of epochal relaxation based on such a scaling. 
In their study, the mean velocity of the electrons is shown to vanish in an inertial frame. Moreover,  the heavy-particle diffusive fluxes were scarcely dealt with since their work is restricted to a single type of heavy particles, and thus 
no multicomponent diffusion was to be found; in such a simplified context, the 
details of the interaction between the heavy particles and electrons 
degenerate
and  the  positivity of the entropy production is straightforward. 
We will establish a theory based on a multiscale analysis for a multicomponent plasma (which includes the single heavy-particle case)
where the mean electron velocity is the mean heavy-particle velocity in absence of external forces. As an alternative, Magin and Degrez~\cite{magin1} have also proposed a model
for a multicomponent plasma  based on a hydrodynamic referential. They have applied a multiscale analysis to the derivation of the multicomponent 
transport fluxes and coefficients. However, the proposed treatment of the collision operators is heuristic. Moreover, since the hydrodynamic velocity is used to define the referential instead of the mean heavy-particle velocity, the Chapman-Enskog method requires a transfer of  lower order terms in the integral equation for the electron perturbation function to ensure mass conservation. Finally, we also emphasize that the development of models for plasmas in thermal equilibrium shall always be obtained as a particular case of the general theory.

Second, the magnetic field  induces anisotropic transport fluxes when the electron collision frequency  is lower than the electron cyclotron frequency of gyration around the magnetic lines. Braginskii~\cite{braginskii} has studied the case of fully ionized plasmas composed of one single ion species. Recently, Bobrova $et~al.$ have generalized the previous work to multicomponent plasmas. However, the scaling used in both contributions does not comply with a dimensional analysis of the Boltzmann equation. Lucquin~\cite{lucquin1,lucquin2} has investigated magnetized plasmas in the latter framework.  Nevertheless, the same limitation is found for  the diffusive fluxes as in reference~\cite{degond1,degond2}. Finally, Giovangigli and Graille~\cite{graille1} have studied the Enskog expansion of magnetized plasmas and obtained macroscopic equations together with expressions of transport fluxes and coefficients. Unfortunately, the difference of mass between the  electrons and heavy particles is not accounted for in their work. 

Third, plasmas are strongly reactive gas mixtures. The kinetic mechanism comprises numerous reactions~\cite{capitelli}: dissociation of molecules by electron and heavy-particle impact, three body recombination, ionization by  electron and heavy-particle impact, associative ionization, dissociative recombination, radical reactions, charge exchange$\ldots$ Giovangigli and Massot~\cite{massot} have  derived a formal theory of chemically reacting flows for the case of neutral gases. Subsequently, Giovangigli and Graille~\cite{graille1} have studied the case of ionized gases. We recall that their scaling does not take into account the mass disparity between electrons and  heavy particles.  Besides, in chemical equilibrium situations, a long-standing theoretical debate in the literature deals with nonuniqueness of the two-temperature Saha equation. Recently, Giordano and Capitelli~\cite{mimmo} have emphasized the importance of the physical constraints imposed on the system by using a thermodynamic approach. A derivation based on kinetic theory should further improve the understanding of the problem.  Choquet and Lucquin~\cite{choquet} have already studied the case of ionization reactions by electron impact.

Fourth, molecules rotate and vibrate, and moreover, the electronic energy levels of atoms and molecules are excited. Generally, the rotational energy mode is considered to be fully excited above a few Kelvins.  In a plasma environment, the vibrational and electronic energy modes are also significantly excited. The detailed treatment of the internal degrees of freedom is however beyond the scope of the present contribution where we will only tackle the translational energy in the context of thermal nonequilibrium. The reader is thus referred to the specialized literature~\cite{brun,maccourt1,nagnibeda}. 

Fifth, the development of numerical methods to solve conservation equations relies on the identification of their intrinsic mathematical structure. For instance, the system of conservation equations of mass, momentum, and energy is found to be nonconservative for thermal nonequilibrium ionized gases. Therefore, this formulation is not suitable for numerical approximations of discontinuous solutions. Coquel and Marmignon~\cite{coquel} have derived  a well-posed conservative formulation based on a phenomenological approach. Nevertheless, their derivation is not consistent with a scaling issued from a dimensional analysis. We will show that kinetic theory, based on first principles, naturally allows for an adequate mathematical structure to be obtained, as opposed to the phenomenological approach.

In this work, we propose to derive the multicomponent plasma conservation equations of mass, momentum, and energy,  as well as the expressions of the associated multicomponent transport fluxes and coefficients. 
The multicomponent Navier-Stokes regime is reached for the heavy particles, which 
follow a hyperbolic scaling, and is coupled to first order drift-diffusion equations 
for the electrons, which follow a parabolic scaling. We deal here with first-order equations, thus one order beyond the expansion commonly investigated for the electrons in the literature. The derivation relies on kinetic theory and is based on the ansatz that the particles constitutive of the plasma are inert and only possess translational degrees of freedom.  The electromagnetic field influence is accounted for.  
In Section~\ref{sec:secbol}, we express the Boltzmann equation in a noninertial reference frame. 
We show that the mean heavy-particle velocity is a suitable choice for the referential velocity. This step is essential to establish a formalism where the electrons follow the bulk movement of the plasma. Then, we define the reference quantities of the system in order to derive the scaling of the Boltzmann equation from a dimensional analysis.  The multiscale aspect occurs in both the streaming operator and collision operator of the Boltzmann equation. Consequently, Section~\ref{sec:secpremi} is devoted to the scaling of the partial collision operators between unlike particles. Besides, we determine the space of collisional invariants associated with respectively the electrons and the heavy particles. In Section~\ref{sec:sec4}, we resort to a Chapman-Enskog method to derive macroscopic conservation equations. The system is examined at successive orders of approximation, each of which corresponding to a physical time scale. For that purpose, scalar products and linearized collision operators are introduced. 
The global expressions of the macroscopic fluid equations are then provided up to Navier-Stokes equations for the heavy particles and first-order drift-diffusion equations for the electrons.
We also prove that our choice of referential is essential in order to reach this expansion level. In Section~\ref{sec:sectransport},  we establish the formal existence and uniqueness of a solution to the Boltzmann equation. The multicomponent transport coefficients are then calculated in terms of bracket operators whose mathematical structure allows for the sign of the transport coefficients to be determined; in particular, the Kolesnikov effect, or the crossed contributions to the mass and energy transport fluxes coupling the electrons and heavy particles. The explicit expressions of the transport coefficients can be obtained by means of a Galerkin spectral method~\cite{chapman}  disregarded in the present contribution. Finally in Section~\ref{sec:secentropy}, the first and second principles of thermodynamics are proved to be satisfied  by deriving a total energy equation and an entropy equation. Then, we establish, from a fluid standpoint, a conservative formulation of the system of macroscopic equations. We also identify  the mathematical structure of the purely convective system. Hence, we demonstrate that kinetic theory shall be used as a guideline in the derivation of the macroscopic conservation equations as well as in the design of the associated numerical methods. 

Beyond the obvious interest of such a study from the point of view of the applications and design of numerical schemes, the present contribution 
also yields a formal kinetic theory of mixtures of separate masses, where the light species obey a parabolic scaling whereas the heavy species obey a hyperbolic scaling. The 
original treatment of the two different scalings for fluid flows was first provided by Bardos $et~al.$~\cite{bgl}. In their study, the purely hyperbolic scaling yields the compressible gas dynamics equations, whereas the purely 
parabolic scaling leads to the low Mach number limit. These scalings are quite classical and both of them can be used for various asymptotics such as the Vlasov-Navier-Stokes equations 
in different regimes investigated by Goudon $et~al.$~\cite{goudon1,goudon2}. Yet, a rigourous derivation of a set of macroscopic
equations in the situation where the hyperbolic and parabolic scalings are entangled 
in the same problem is an original result obtained in the present work.
 
\section{Boltzmann equation}\label{sec:secbol}

\subsection{Assumptions}

\begin{enumerate}
\item {\label{hyp1}}Our plasma is described by the kinetic theory of gases based on classical mechanics, provided that: a) The mean distance between the gas particles $1/(n^0)^{1/3}$ is larger than the thermal de Broglie wavelength, where $n^0$ is a reference number density, b) The square of the ratio of the electron thermal speed  $V_\elec^0$ to the speed of light is small.
\item {\label{hyp1bis}}The reactive collisions and particle internal energy are not accounted for.  
\item {\label{hyp1ter}}The particle interactions are modeled as binary encounters by means of a Boltzmann collision operator, provided that: a) The gas is sufficiently dilute, $i.e.$, the mean distance between the gas particles $1/(n^0)^{1/3}$ is larger than the particle interaction distance $(\sigma^0)^{1/2}$, where $\sigma^0$ is a reference differential cross-section common to all species, b) The plasma parameter, quantity proportional to the number of electrons in a sphere of radius equal to the Debye length, is supposed to be large. Hence, multiple charged particle interactions are treated as equivalent binary collisions by means of a Coulomb potential screened at the Debye length~\cite{balescu,delcroix}. 
\item {\label{hyp2}} A plasma is composed of electrons and a multicomponent mixture of 
heavy particles (atoms, molecules, and ions). The ratio of the electron mass $m_{\elec}^0$  to a characteristic heavy-particle mass $m_h^0$ is such that the nondimensional number $\varepsilon=\sqrt{\me^0/\mh^0}$ is small.  
\item {\label{hyp3}} The number of Mach, defined as a reference hydrodynamic velocity divided by the heavy-particle thermal speed $\Mh=v^0/V_h^0$, is supposed to be of the order of one.
\item {\label{hyp4}} The macroscopic time scale $t^0$ is assumed to be comparable with the heavy-particle kinetic time scale $t_h^0$ divided by $\epsilon$. The macroscopic length scale is based on a reference convective length $L^0=v^0t^0$. 
\item {\label{hyp6}} The reference electrical and thermal energies of the system are of the same order of magnitude.
\end{enumerate}
The mean free path $l^0$ and macroscopic length scale $L^0$ allow for the Knudsen number to be defined $Kn=l^0/L^0$.  It will be shown that this quantity  is small, provided that assumptions~\eqref{hyp2}-\eqref{hyp4} are satisfied. Therefore, a continuous description of the system is deemed to be possible.

\subsection{Inertial frame}
The choice of an adequate referential will prove to be essential in the following 
multiscale analysis. Two referentials are commonly used in the literature.
Degond and Lucquin~\cite{degond1,degond2} work in the inertial frame, as Ferziger and Kaper~\cite{ferziger}. The second referential is presented in the following section.
Considering assumptions~\eqref{hyp1}-\eqref{hyp1ter}, the temporal evolution of the velocity distribution function $\fHi^\star$ of the plasma particles $i$ is governed  in the phase space $(\x^\star,\ci^\star)$ by the Boltzmann equation~\cite{cercignani1,ferziger}
\begin{equation}\label{eqbolbol}
\Di[i]^\star(\fHi^\star)=\Ji^\star, \quad i\in\espece,
\end{equation}
where symbol $\espece$ is the set of indices of the gas species. Dimensional quantities are denoted by the superscript $^\star$. The streaming operator reads
\begin{equation}\label{eqinertialstream}
\Di[i] ^\star(\fHi^\star)=\dtstar \fHi^\star +\ci^\star\dscal\dxstar\fHi^\star+\frac{\qi^\star}{\mi^\star}\left(\E^\star+\ci^\star\pvect\B^\star\right)\!\dscal\dcistar\fHi^\star, \quad i\in\espece,
\end{equation}
in an inertial frame. Symbol $t^\star$ stands for time, $\E^\star$, the electric field, $\B^\star$, the magnetic field, $\mi^\star$, the mass of the particle $i$, and $\qi^\star$, its charge. The collision operator is given by
\begin{align}\label{eqcolop}
\Ji^\star&=\sume[j]\Jij^\star\left(\fHi[i]^\star,\fHi[j]^\star\right),\quad i\in\espece,
\end{align}
with the partial collision operator of particle $j$ impacting on particle $i$
\begin{equation}
\Jij^\star\left(\fHi[i]^\star,\fHi[j]^\star\right)=\int \left(\fHi[i]^{\star\prime}\fHi[j]^{\star\prime}-\fHi[i]^\star\fHi[j]^\star\right)|\ci^\star-\ci[j]^\star|\seceff^\star \d\ovec \d\ci[j]^\star, \quad i,j\in\espece.\label{eqpartcolop}
\end{equation}
After collision, quantities are denoted  by the superscript $^\prime$. 
The differential cross-section $\seceff^\star=$$\seceff^\star\left[\mred^\star|\ci^\star-\ci[j]^\star|^2/(\boltz T^0),\ovec\dscal\evec\right]$   depends on the relative kinetic energy of the colliding particles and the cosine of the angle between the unit vectors  of relative velocities $\ovec=(\ci^{\star\prime}-\ci[j]^{\star\prime})/|\ci^{\star\prime}-\ci[j]^{\star\prime}|$ and $\evec=(\ci^\star-\ci[j]^\star)/|\ci^\star-\ci[j]^\star|$. Quantity $\mred^\star=\mi^\star\mi[j]^\star/(\mi^\star+\mi[j]^\star)$ is  the reduced mass of the particle pair, $T^0$, a reference temperature, and $\boltz$, Boltzmann's constant. Therefore, the differential cross-sections are symmetric with respect to their indices $i,j\in\espece$, $i.e.$, $\seceff^\star=\seceff[ji]^\star$.

\subsection{Noninertial frame}
Sutton and Sherman~\cite{sutton}, as Chapman and Cowling~\cite{chapman}, 
have proposed a noninertial frame based on the hydrodynamic velocity 
\begin{equation}\label{eqhydro1}
\rho^\star\speed^\star=\sume[j]\int\mi[j]^\star\ci[j]^\star\fHi[j]^\star \d\ci[j]^\star,
\end{equation}
where the mixture mass density is defined as $\rho^\star=\sum_{j\in\espece}\rhoi[j]^\star$. Symbol $\rhoi^\star=\ni^\star\mi^\star$ stands for the partial mass density, and $\ni^\star=\int \fHi^\star \d\ci^\star$, the partial number density.
It is a convenient choice 
since it is the referential associated with the definition 
of the peculiar velocities 
\begin{equation}\label{eqpecuhydro}
\Ci^{\speed\star}=\ci^\star-\speed^\star , \quad i\in\espece,
\end{equation}
induced from the usual momentum constraint. We infer from definition \eqref{eqhydro1} that the global diffusion flux vanishes
\begin{equation}\label{eqmassconst}
\sume[j]\int\mi[j]^\star\Ci[j]^{\speed\star}\fHi[j]^\star \d\ci[j]^\star=0,
\end{equation}
that is, the standard momentum constraint. 

Given the strong disparity of mass between the electrons and heavy particles, a frame linked with the heavy particles appears to be a rather natural choice for plasmas, as fully justified in the following detailed multiscale analysis. 
Thus, we define the mean electron velocity and mean heavy-particle velocity
\begin{equation}\label{eqhydro}
\rhoe^\star\speed_\elec^\star=\int\me^\star\ce^\star\fe^\star \d\ce^\star, \quad\rhoi[\heavy]^\star\vitesse^\star=\sumi[j]\int\mi[j]^\star\ci[j]^\star\fHi[j]^\star \d\ci[j]^\star,
\end{equation}
where the heavy-particle mass density reads $\rhoi[\heavy]^\star=\sum_{j\in\lourd}\rhoi[j]^\star$. In the $\vitesse^\star$  referential,
the free electrons interact with heavy particles whose global movement is frozen in space. A similar viewpoint is commonly adopted in the quantum theory of molecules when the Born-Oppenheimer approximation is used to study the motion of the bound electrons about the nuclei~\cite{born}. 
Based on the following definition of peculiar velocities
\begin{equation}\label{eqpecu}
\Ci^{\star}=\ci^\star-\vitesse^\star , \quad i\in\espece,
\end{equation}
 the heavy-particle diffusion flux is shown to vanish
\begin{equation}\label{eqmassconst2}
\sumi[j]\int\mi[j]^\star\Ci[j]^{\star}\fHi[j]^\star \d\ci[j]^\star=0.
\end{equation}

For now, we defer the choice of the referential velocity. Therefore, we use the symbol $\velocity^\star$ to define the peculiar velocities $\Ci^{\velocity\star}=\ci^\star-\velocity^\star$, $i\in\espece$. Then, the Boltzmann equation is expressed in a frame moving at $\velocity^\star$ velocity by means  of the latter change of variables. Hence, the streaming operator \eqref{eqinertialstream} is transformed into the expression
\begin{multline}
\Di[i]^\star(\fHi^\star)=\dtstar \fHi^\star +\left(\Ci^{\velocity\star}+\velocity^\star\right)\dscal\dxstar\fHi^\star+\frac{\qi^\star}{\mi^\star}\left[\E^\star+\left(\Ci^{\velocity\star}+\velocity^\star\right)\pvect\B^\star\right]\dscal\dCistar\fHi^\star\\
-\frac{\D \velocity^\star}{\D\temps^\star}\dscal\dCistar\fHi^\star
-\dCistar^\star\fHi^\star\ptens\Ci^{\velocity\star}\pmat\dxstar\velocity^\star,
\end{multline}
where ${\D}/{\D\temps^\star}=\dtstar+\velocity^\star\dscal\dxstar$. 
The partial collision operator~\eqref{eqpartcolop} is found to be
\begin{equation}
\label{eqpartcolopC}
\Jij^\star\left(\fHi[i]^\star,\fHi[j]^\star\right)=\int \left(\fHi[i]^{\prime\star}\fHi[j]^{\prime\star}-\fHi[i]^\star\fHi[j]^\star\right)|\Ci^{\velocity\star}-\Ci[j]^{\velocity\star}|\seceff^\star \d\ovec \d\Ci[j]^{\velocity\star},\quad i,j\in\espece.
\end{equation}
In a noninertial frame, the velocity distribution function $f_i^\star$, the differential cross-section $\seceff^\star=\seceff^\star\left[\mred^\star|\Ci^{\velocity\star}-\Ci[j]^{\velocity\star}|^2/(\boltz T^0),~\ovec\dscal\evec\right]$, as well as both the unit vectors $\ovec=(\Ci^{\velocity\prime\star}-\Ci[j]^{\velocity\prime\star})/|\Ci^{\velocity\prime\star}-\Ci[j]^{\velocity\prime\star}|$ and $\evec=(\Ci^{\velocity\star}-\Ci[j]^{\velocity\star})/|\Ci^{\velocity\star}-\Ci[j]^{\velocity\star}|$ depend on the peculiar velocities. Nevertheless, for simplicity reasons, notations are unchanged with respect to the inertial frame, where the previous quantities depend on the absolute velocities.  Moreover, we define collisional invariants.
\begin{definition}
The space of scalar collisional invariants $\invspace^{\velocity\star}$ is spanned by the following elements
\begin{equation}
\label{eqdimcolinv}
\left\{
\vcenter{\halign{
$#\hfil$&$\;\,=\; \bigl(#\bigr)_{i\in\espece},\hfil$&$\qquad#,\hfil$
\cr
\invstar[\velocity,j]&\mi^\star\delta_{ij}&j\in\espece
\cr\noalign{\vskip4pt}
\invstar[\velocity,\ns+\nu]&\mi^\star\Cinorme[i\nu]^{\velocity\star}&\nu\in\{1,2,3\}
\cr\noalign{\vskip4pt}
\invstar[\velocity,\ns+4]&\tfrac{1}{2}\mi^\star\Ci^{\velocity\star}\dscal\Ci^{\velocity\star}
\cr}} \right.
\end{equation}
where symbol $\ns$ denotes the cardinality of the set of species $\espece$.
\end{definition} 
Besides, a scalar product  is introduced
\begin{equation}
\label{eqppscal}
\ppscal{\xi^\star,\zeta^\star}^{\velocity\star} =\sume[j]\int\xi_j^\star\pcont\bar{\zeta_j}^\star\;\d\Ci[j]^{\velocity\star},
\end{equation}
for families $\xi^\star={(\xi_i^\star)}_{i\in\espece}$ and $\zeta^\star={(\zeta_i^\star)}_{i\in\espece}$. Symbol $\pcont$ stands for the maximum contracted product in space and symbol $~\bar{}~$ for the conjugate transpose operation. Ergo, the collision operator 
$\Ji[]^\star={(\Ji^\star)}_{i\in\espece}$ defined in eq.~\eqref{eqcolop} obeys the following property.

\begin{property}
The collision operator $\Ji[]^\star$ is orthogonal to the space of collisional invariants $\invspace^{\velocity\star}$, $i.e.$, $\ppscal{\invstar[\velocity,l],\Ji[]^\star}^{\velocity\star}=0$,  for all $l\in\{1,\ldots,\ns{+}4\}$.\label{propo2.1}
\end{property}

\begin{proof}
The projection of the collision operator $\Ji[]^\star$ onto $\invstar[\velocity,l]$,  $l\in\{1,\ldots,\ns{+}4\}$,
is shown to be
\begin{equation*}
\tfrac{1}{4}\sum_{i,j\in\espece}\int 
(\fHi[i]^{\star\prime}\fHi[j]^{\star\prime}-\fHi[i]^\star\fHi[j]^\star)
(\invstari[\velocity,l]+\invstarj[\velocity,l]-\inv[\velocity,l'\star]_i-\inv[\velocity,l'\star]_j)
|\Ci^{\velocity\star}-\Ci[j]^{\velocity\star}|\sigma_{ij}^\star 
\d\ovec \d\Ci^{\velocity\star} \d\Ci[j]^{\velocity\star},
\end{equation*}
see for instance Chapman and Cowling~\cite{chapman}. 
The latter expression vanishes for all $l\in\{1,\ldots,\ns{+}4\}$.
\end{proof}

Finally, the macroscopic properties can be expressed by means of the scalar product of the distribution functions and the collisional invariants
\begin{equation*}
\left\{
\vcenter{\halign{
$\ppscal{\fHi[]^\star,#}{}^\star\hfil$&$\;\,=\; #,\hfil$&$\qquad#,\hfil$
\cr
\invstar[i]&\rhoi^\star&i\in\espece
\cr\noalign{\vskip4pt}
\invstar[\ns+\nu]&\rho^\star(v_\nu^\star-u_\nu^\star)
&\nu\in\{1,2,3\}
\cr\noalign{\vskip4pt}
\invstar[\ns+4]&\rho^\star e^\star+\tfrac{1}{2}\rho^\star(\speed^\star-\velocity^\star)\dscal(\speed^\star-\velocity^\star)
\cr}} \right.
\end{equation*}
where quantity $e^\star$ stands for  the gas thermal energy per unit mass.

\subsection{Dimensional analysis}\label{sec:secdiman}

A sound scaling of the Boltzmann equation is deduced from a dimensional analysis inspired from Petit and Darrozes~\cite{petit}. First, reference quantities are introduced in Table~\ref{tab1}.
\begin{table}[ht] 
\caption{Reference quantitities.\label{tab1}}
{\small\begin{tabular}{@{}ccc@{}} \toprule
&{\bf Common to all species}&\\
Temperature&$T^0$&\\
Number density&$n^0$&\\
Differential cross-section&$\sigma^0$&\\
Mean free path&$l^0$&\\
Macroscopic time scale &$t^0$&\\
Hydrodynamic velocity &$v^0$&\\
Macroscopic length &$L^0$&\\
Electric field&$E^0$&\\
Magnetic field&$B^0$&\\
&&\\
&{\bf Electrons}&{\bf Heavy particles}\\
Mass&$\me^0$&$m_h^0$\\
Thermal speed&$V_{\elec}^0$&$V_h^0$\\
Kinetic time scale&$t_{\elec}^0$&$t_h^0$\\
 \botrule
\end{tabular}}
\end{table} 
The characteristic temperature, number density, differential cross-section, mean free path, macroscopic time scale, hydrodynamic velocity, macroscopic length, and electric and magnetic fields are assumed to be common to all species. The nondimensional number
\begin{equation}
\epsilon=\sqrt{\frac{\me^0}{\mh^0}}
\end{equation}
quantifies the ratio of the electron mass to a reference heavy-particle mass. 
According to assumption~\eqref{hyp2},  the value of $\epsilon$ is small. Consequently, electrons exhibit a larger thermal speed than that of  heavy particles
\begin{equation}\label{eqthermalspeed}
V_{\elec}^0=\sqrt{\frac{\boltz \temp^0}{\me^0}}, \qquad V_\heavy^0=\sqrt{\frac{\boltz \temp^0}{\mh^0}}=\epsilon V_\elec^0.
\end{equation}
Moreover, the electron and heavy particle temperatures may be distinct, provided that 
eq.~\eqref{eqthermalspeed} does not fail to describe the reference thermal speeds.  The differential cross-sections  are of the same order of magnitude $\sigma^0$. Hence, the characteristic mean free path 
$l^0={1}/{(n^0\sigma^0)}$ is found to be identical for all species. As a result, the kinetic time scale, or relaxation time of a distribution function towards its respective quasi-equilibrium state,  is lower for electrons than for heavy particles
\begin{equation}
t_{\elec}^0=\frac{l^0}{V_{\elec}^0},\qquad
t_\heavy^0=\frac{l^0}{V_\heavy^0}=\frac{t_{\elec}^0}{\epsilon}.
\end{equation}
Assumption~\eqref{hyp4} states that the macroscopic time scale reads
\begin{equation}
t^0=\frac{t_\heavy^0}{\epsilon}.
\end{equation}
It is shown in Section~\ref{sec:sec4} that this quantity corresponds to the averaged time during which electrons and heavy particles exchange their energy  through encounters. In addition, the macroscopic temporal and spatial scales are linked by the expression
\begin{equation}
L^0=v^0t^0,
\end{equation}
where the hydrodynamic velocity  is determined by the Mach number $\Mh={v^0}/{V_\heavy^0}$. Given assumption~\eqref{hyp3}, the Mach number is of the order of one.  Hence, the Knudsen number
\begin{equation}
Kn=\frac{l^0}{L^0}=\frac{\epsilon}{\Mh},
\end{equation}
is small. Finally, following assumption~\eqref{hyp6}, the reference electric field is such that
\begin{equation}
q^0E^0L^0=\boltz T^0.
\end{equation}
The intensity of the  magnetic field is governed by the  Hall numbers of the electrons and heavy particles
\begin{equation}
\beta_{\elec}=\frac{q^0B^0}{\me^0}t_{\elec}^0=\epsilon^{1-b},\qquad
\beta_{\heavy}=\frac{q^0B^0}{\mh^0}t_{\heavy}^0=\epsilon\beta_{\elec},
\end{equation}
defined as the Larmor frequencies, ${q^0B^0}/{\me}$ for the electrons and ${q^0B^0}/{\mh^0}$ for the heavy particles, multiplied by their corresponding kinetic time scale. The magnetic field is assumed to be proportional to a power of $\epsilon$ by means of an integer $b\le1$. The physical interpretation of the $b$ parameter appears at the end of Section~\ref{sec:secentropy}. 

The dimensional analysis can be summarized as follows: a) Two spatial scales were introduced, one spatial scale at the microscopic level and one spatial scale at the macroscopic level; b) Whereas three temporal scales were defined, two time scales at the microscopic level, respectively for the electrons and for the heavy particles, and one time scale at the macroscopic level, common to all species. 
  
Nondimensional variables are based on the reference quantities. They are denoted by dropping the superscript $\null^\star$. In particular, one has the following expressions for  the particle velocities 
\begin{equation}
\ce^\star={V_{\elec}^0}\ce,\qquad
\ci^\star = {V_{\heavy}^0}\ci,\quad i\in\lourd,
\end{equation}
where symbol $\lourd$ stands for the set of indices of heavy particles. Both the reference hydrodynamic velocity and  mean heavy-particle velocity are equal to $v^0$. Indeed, the hydrodynamic velocity defined in eq.~\eqref{eqhydro1} is found to be
\begin{equation} \label{pfff}
(\rhoi[\heavy]+\varepsilon^2\rhoe)\Mh\speed=\rhoi[\heavy]\Mh\vitesse+\varepsilon^2\rhoe\speed_\elec,
\end{equation}
whereas the mean electron and heavy-particle velocities given in eq.~\eqref{eqhydro} read
\begin{equation}
\rhoe\speed_\elec=\frac{1}{\varepsilon}\int\ce\fe\d\ce, \quad \rhoi[\heavy]\Mh\vitesse=\sumi[j]\int\mi[j]\ci[j]\fHi[j]\d\ci[j].
\end{equation}
The peculiar velocities are given by the relations
\begin{equation}\label{eqpaculiar}
\Ce^\velocity = \ce - \epsilon\Mh\velocity,\qquad
\Ci^\velocity = \ci - \Mh\velocity,\quad i\in\lourd.
\end{equation}
Usually, they are associated with the momentum constraints of
the mixture, so that the natural choice is $\velocity=\speed$. In such a case, we 
get the following relation
\begin{equation*}
\sumi[j]\int\mi[j]\Ci[j]^\speed\fHi[j]\d\Ci[j]^\speed + \varepsilon\int\Ce^\speed\fe\d\Ce^\speed =0.
\end{equation*}
However, the hydrodynamic velocity of the mixture,  electrons included,  
can also be expanded in the $\varepsilon$ parameter and thus receives contributions at various $\varepsilon$ orders in the Chapman-Enskog method. 
Since the change of referential should not differ 
depending on the expansion order, we could mimic the approach of Lucquin and Degond \cite{degond1,degond2,lucquin1} and take $\velocity=0$, which means working in the inertial framework.
However, we follow a different path not only by choosing the mean heavy-particle velocity as referential velocity, $\velocity=\vitesse$, but also by defining the peculiar velocities based on this quantity, as opposed to Petit and Darrozes \cite{petit}. 
The rationale for such a choice is threefold: a) This 
quantity does not depend on $\varepsilon$
while still being a perturbation 
of the hydrodynamic velocity of the complete mixture up to second order in $\varepsilon$
\begin{equation}\label{eqrelation}
(\rhoi[\heavy]+\varepsilon^2\rhoe)\Mh(\speed-\vitesse)
=\varepsilon\int\Ce^{\vitesse}\fe\d\Ce^{\vitesse},
\end{equation}
since quantity $\int\Ce^{\vitesse}\fe\d\Ce^{\vitesse}$ taken on a perturbation of a Maxwell-Boltzmann distribution will be of $\ordre(\varepsilon)$ in the framework of the Chapman-Enskog expansion presented in Section~\ref{sec:sec4}; 
b) It will prove to be the natural referential in which the heavy particles thermalize in the context 
of the proposed multiscale analysis;  c) It will also prove to be the only available 
choice for electron thermalization and successive order resolubility, thus 
making the proposed change of referential optimal and 
leading to a rigourous framework as well as a simplified algebra.
In the sequel, since there is no ambiguity, we will drop the $\vitesse$ superscript in the use of the peculiar velocities $\Ce^{\vitesse}$ and $\Ci[j]^{\vitesse}$, $i\in\lourd$.

Consequently, the heavy-particle diffusion flux vanishes, as shown in eq.~\eqref{eqmassconst2} 
\begin{equation}
\sumi[j]\int\mi[j]\Ci[j]\fHi[j]\d\Ci[j]=0.
\end{equation}

Thus, the Boltzmann equation~\eqref{eqbolbol} can be expressed in nondimensional form, respectively for the electrons and heavy particles, as
\begin{multline}\label{eqbol1}
\dt \fe +\tfrac{1}{\epsilon \Mh}(\Ce+\epsilon\Mh\vitesse)\dscal\dx\fe
+\epsilon^{-(1+b)}\qe\bigl[(\Ce+\epsilon\Mh\vitesse)\pvect\B\bigr]\dscal\dCe\fe \\
+\left(\tfrac{1}{\epsilon \Mh}\qe\E-\epsilon\Mh\tfrac{\D\vitesse}{\D\temps}\right)\dscal\dCe\fe
-\dCe\fe\ptens\Ce\pmat\dx\vitesse=\tfrac{1}{\epsilon^2}\Je,
\end{multline}
\begin{multline}\label{eqbol2}
\dt \fHi +\tfrac{1}{ \Mh}(\Ci+\Mh\vitesse)\dscal\dx\fHi
+\epsilon^{1-b}\tfrac{\qi}{\mi}\bigl[(\Ci+\Mh\vitesse)\pvect\B\bigr]\dscal\dCi\fHi \\
+\left(\tfrac{1}{ \Mh}\tfrac{\qi}{\mi}\E-\Mh\tfrac{\D\vitesse}{\D\temps}\right)\dscal\dCi\fHi
-\dCi\fHi\ptens\Ci\pmat\dx\vitesse=\tfrac{1}{\epsilon}\Ji,
\quad i\in\lourd,
\end{multline}
where the collision operators read
\begin{align}
\Je&=\Jee\left(\fe,\fe\right)+\sumi[j]\Jej\left(\fe,\fHi[j]\right),\\
\Ji&= \tfrac{1}{\epsilon}\Jie(\fHi,\fe) +\sumi[j]\Jij(\fHi,\fHi[j] ), &&i\in\lourd.
\end{align}
The collisional invariants \eqref{eqdimcolinv} depend on the mass ratio as well, as shown in their  nondimensional form.

\begin{definition}
The space of scalar collisional invariants $\invspace$ is spanned by the following elements
$\inveps = (\invepse,\invepsh)$, $l\in\{1,\ldots,\ns{+}4\}$, with
\begin{equation}
\label{eqcolinveps}
\left\{
\vcenter{\halign{
$#\hfil$&$\;\,=\;\,#,\hfil$&$\qquad#\hfil$&$\;\,=\;\,#$&$\bigl(#$&$#\hfil$
\cr
\invepse[j]&\varepsilon^2\delta_{\elec j}&
\invepsh[j]&&\mi\delta_{ij}\bigr)_{i\in\lourd},&\qquad j\in\espece,
\cr\noalign{\vskip4pt}
\invepse[\ns+\nu]&\varepsilon\Cinorme[\elec \nu]&
\invepsh[\ns+\nu]&&\mi\Cinorme[i\nu]\bigr)_{i\in\lourd},&\qquad\nu\in\{1,2,3\},
\cr\noalign{\vskip4pt}
\invepse[\ns+4]&\tfrac{1}{2}\Ce\dscal\Ce&
\invepsh[\ns+4]&&\tfrac{1}{2}\mi\Ci\dscal\Ci\bigr)_{i\in\lourd}.&
\cr}} \right.
\end{equation}
\end{definition} 
It is worth noticing the influence of the hierarchy of scales; 
whereas the scaling does not
introduce any structural change in the mass and energy collisional invariants, the electron contribution disappears from the momentum collisional invariant vector in the limit of $\varepsilon$ tends to zero. 
A similar behavior can be observed for the total mass; however, the single species collisional invariants are not affected. 

Let us underline that eq. \eqref{eqbol1} for the light species is typical of a parabolic scaling, which corresponds to the low Mach number limit for the electron gas, whereas eq. \eqref{eqbol2} for the heavy species
is typical of a hyperbolic scaling, which corresponds to the compressible gas dynamics for the heavy-species gas mixture \cite{bgl}. The present scaling is thus intermediate between the usual cases and  the mathematical structure of the resulting system of macroscopic equations has to be identified.

For a family $\xi={(\xi_i)}_{i\in\espece}$, we introduce two separate contributions: $\xi_\elec$, concerning the electrons, and $\xi_\heavy={(\xi_i)}_{i\in\lourd}$, concerning the heavy particles.
Hence, the scalar product  
between the families $\xi=(\xi_i)_{i\in\espece}$ and $\zeta=(\zeta_i)_{i\in\espece}$ defined in eq.~\eqref{eqppscal} is decomposed into a sum of partial scalar products with different scales
\begin{equation}
\ppscal{\xi,\zeta} =\ppscale{\xi_\elec,\zeta_\elec}+\varepsilon^3\ppscalh{\xi_\heavy,\zeta_\heavy},
\end{equation}
given by the expressions
\begin{equation}
\ppscale{\xi_\elec,\zeta_\elec}=\int\xi_\elec\pcont\bar{\zeta_\elec}\;\d\Ce,\qquad
\ppscalh{\xi_\heavy,\zeta_\heavy}=\sumi[j]\int\xi_j\pcont\bar{\zeta_j}\;\d\Ci[j].
\end{equation}
Finally, we introduce the collision  operator $\Jepsi[]=(\varepsilon \Je, \tfrac{1}{\epsilon} \Jh)$, where 
eq.~\eqref{eqbol1} has been multiplied by a factor
$\varepsilon^3$ corresponding to a coherent scaling of the two Boltzmann equations. Then, we derive the following property.

\begin{property}
\label{th:corcolinv}
The collision operator $\Jepsi[]$ is orthogonal to the space of collisional invariants $\invspace$, $i.e.$, $\ppscal{\inveps,\Jepsi[]}=0$,  for all $l\in\{1,\ldots,\ns{+}4\}$. 
Furthermore, the terms of $\ppscal{\inveps,\Jepsi[]}$ cancel by pair of interaction, $i.e.$,
\begin{gather}
\ppscale{\invepse,\Jee}=0,\label{eqcolinvee}\\
\smash{\sumi[j]}\ppscale{\invepse,\Jej}+\ppscalh{\invepsh,\Jhe}=0,\label{eqcolinveh}\\
\sumi[j]\ppscalh{\invepsh,\Jhj}=0,
\end{gather} 
respectively for the electron, electron heavy-particle, and heavy-particle interactions.
\end{property}

\begin{proof}
The projection of the collision operator $\Jepsi[]$ onto $\inveps$,  $l\in\{1,\ldots,\ns{+}4\}$, is given by the expression
\begin{align*}
\ppscal{\inveps,\Jepsi[]}&=
\varepsilon\ppscale{\invepse,\Jee}+
\varepsilon\sumi[j]\ppscale{\invepse,\Jej} +
\varepsilon\ppscalh{\invepsh,\Jhe} +
\varepsilon^2\sumi[j]\ppscalh{\invepsh,\Jhj}.
\end{align*}
The terms of this sum are examined by interaction pairs
\begin{equation*}
\ppscale{\invepse,\Jee}=\frac{1}{4}\int
(\fe^{\prime}\feu^{\prime}-\fe\feu)
(\invepse+\invepseu-\invepse[l \prime]-\invepseu[l \prime])
|\Ce-\Ci[\elec 1]|\sigma_{\elec \elec 1} \d\ovec \d\Ce\d\Ci[\elec 1],
\end{equation*}
\begin{multline*}
\smash{\sumi[j]}\;\ppscale{\invepse,\Jej}+\ppscalh{\invepsh,\Jhe}\\
=\frac{1}{2}\sum_{j\in\lourd}\int
(\fe^{\prime}\fHi[j]^{\prime}-\fe\fHi[j])
(\invepse+\invepshj-\invepse[l \prime]-\invepshj[l\prime])
|\Ce-\varepsilon\Ci[j]|\sigma_{\elec j} \d\ovec \d\Ce\d\Ci[j],
\end{multline*}
\begin{multline*}
\smash{\sumi[j]}\;\ppscalh{\invepsh,\Jhj}\\
=\frac{1}{4}\sum_{i,j\in\lourd}\int \left(\fHi[i]^{\prime}\fHi[j]^{\prime}-\fHi[i]\fHi[j]\right)\left( \invepshi[l]+\invepshj[l]-\invepshi[l\prime]-\invepshj[l\prime]\right)|\Ci-\Ci[j]|\sigma_{ij} \d\ovec \d\Ci\d\Ci[j].
\end{multline*}
These expressions vanish and thus the sum $\ppscal{\inveps,\Jepsi[]}=0$.
\end{proof}

The multiscale analysis occurs at three levels: a) In the kinetic equations~\eqref{eqbol1} and 
\eqref{eqbol2}; b) In the collisional invariants \eqref{eqcolinveps} and thus in the conservation of the associated  macroscopic quantities; c) In the collision operators. Actually, encounters between particles of the same type are dealt with as usual, whereas the collision operators between unlike particles depend on the $\varepsilon$ parameter via their relative kinetic energy and velocity, and  the vectors $\ovec$ and $\evec$. The scaling of these operators is investigated in the following section.

\section{Preliminary results}\label{sec:secpremi}

\subsection{Electron heavy-particle collision dynamics}

The study of the electron heavy-particle collision dynamics yields the dependence of the peculiar velocities on the 
$\epsilon$ parameter. First, we express the momentum conservation in terms of the peculiar velocities.
Considering a collision of a heavy species, $i\in\lourd$, against an electron, the peculiar velocities after collision $\Ci'$ and $\Ce'$ are related to their counterpart before collision $\Ci$ and $\Ce$
\begin{equation} \label{vitaprescollision1}
\left| \;\,
\vcenter{\halign{
$#\hfil=$&$\,\displaystyle\frac{#}{\mi+\epsilon^2}\,\Ce\hfil\,+$&
$\,\displaystyle\frac{#\mi}{\mi+\epsilon^2}\;\Ci\hfil\,$
&$#$&$\,\displaystyle\frac{#}{\mi+\epsilon^2}\, |\epsilon\,\Ci-\Ce| \,\ovec ,\hfil$&
$\qquad#\in\lourd,\hfil$
\cr
\Ci'&\epsilon&&+s&\epsilon&i
\cr\noalign{\vskip4pt}
\Ce'&\epsilon^2&\epsilon&-s&\mi
\cr}} \right.
\end{equation}
where the direction of the relative velocities after collision  is defined in their center of mass by 
$$\ovec=s\frac{\epsilon\Ci'-\Ce'}{|\epsilon\Ci'-\Ce'|}.$$
Symbol $s$ stands for an integer either equal to $+1$ for the collision operator $\Jie$, $i\in\lourd$, or $-1$ for $\Jej[i]$, $i\in\lourd$. This notation allows for the interaction considered in eq.~\eqref{eqpartcolopC} to be distinguished. We are now able to expand the crossed-collision operators.

\subsection{Expansion of the collision operator $\Jie$}

The dimensional analysis yields the following expression of the nondimensional collision operator
$\Jie$, $i\in\lourd$,
\begin{multline}
\label{Jie}
\Jie(\fHi,\fe)(\Ci) = \int \seceff[ie]\!\left( |\adimg|^2,\ovec\dscal \tfrac{\adimg}{|\adimg |} \right)
|\epsilon\Ci-\Ce| \\
\Bigl[ \fHi(\Ci')\fe(\Ce')-\fHi(\Ci)\fe(\Ce) \Bigr] \;\d\ovec \;\d\Ce, 
\end{multline}
where the relative kinetic energy and the vector $\evec$ are expressed by means of the vector $\adimg=s(\varepsilon\Ci-\Ce)/(1+\varepsilon^2/\mi)^{1/2}$. 

We then introduce the generalized momentum cross-section~\cite{chapman} in a thermal nonequilibrium context
\begin{equation}
\Qijl i\elec l(|\adimg|^2) = 
2\pi \int_0^\pi \seceff[i\elec](|\adimg|^2,\cos\theta)(1-\cos^l\theta)\sin\theta\;\d\theta, 
\quad i\in\lourd,
\end{equation}
where symbol $\theta$ stands for the angle between the vectors $\ovec$ and $\evec$. For $l=1$, this cross-section represents the average momentum transfered in encounters between $i$ heavy particles and electrons for a given value of the relative kinetic energy.

\begin{theorem}\label{th:thmJie}
The collision operator $\Jie$, $i\in\lourd$, can be expanded in the form
\begin{equation}
\Jie(\fHi,\fe)(\Ci)=\epsilon\Jieu(\fHi,\fe)(\Ci)+\epsilon^2\Jied(\fHi,\fe)(\Ci)
+\epsilon^3\Jiet(\fHi,\fe)(\Ci)+\ordre(\epsilon^4).
\end{equation}
The zero-order collision operator $\Jieo(\fHi,\fe)(\Ci)$, $i\in\lourd$, vanishes.  The first-order term $\Jieu$, $i\in\lourd$, reads
\begin{equation}
\Jieu(\fHi,\fe)(\Ci) = -
\frac{1}{\mi}\partial_{\Ci}\fHi(\Ci) \dscal \int  \Qijl i\elec 1(|\adimg|^2) 
\;|\adimg| \adimg \;
\fe(\adimg) \;\d\adimg,\quad i\in\lourd.\label{Jieu}
\end{equation}
The second order term $\Jied$, $i\in\lourd$, is found to be
\begin{multline}
\Jied(\fHi,\fe)(\Ci) = - \frac{1}{\mi} \partial_{\Ci}(\fHi(\Ci)\Ci) \pmat
\int \Qijl i\elec 1(|\adimg|^2) 
|\adimg| \partial_{\Ce}\fe(\adimg )\ptens \adimg \;\d\adimg\\
+ \frac{1}{4\mi^2}\partial_{\Ci\Ci}^2\fHi(\Ci)\pmat
\int \Qijl i\elec 2 (|\adimg |^2)
|\adimg| (|\adimg|^2\identite-3\adimg\ptens\adimg) \fe(\adimg)\;\d\adimg\\
+ \frac{1}{\mi^2}\partial_{\Ci\Ci}^2\fHi(\Ci)\pmat
\int \Qijl i\elec 1(|\adimg|^2)
|\adimg| \adimg\ptens\adimg \fe(\adimg)\;\d\adimg.
\label{Jied}
\end{multline}
Finally, the third-order term $\Jiet$, $i\in\lourd$, is given by
\begin{multline}
\Jiet(\fHi,\fe)(\Ci) =
 \frac{1}{\mi}\partial_{\Ci}(\tfrac{1}{2}\fHi\Ci\ptens\Ci) \pcont \int  
 \Qijl{i}{\elec}{1}(|\adimg|^2) 
\;|\adimg| \partial_{\Ce\Ce}^2 \fe(\adimg)\ptens\adimg 
\;\d\adimg \\
+ \frac{1}{2\mi^2}\partial_{\Ci}\fHi(\Ci)\dscal\int  \Qijl{i}{\elec}{1}(|\adimg|^2) 
|\adimg| \adimg\ptens\adimg\dscal\partial_{\Ce}\fe(\adimg)\;\d\adimg\\
\displaybreak[0]
+ \frac{1}{\mi^2}\partial_{\Ci}(\partial_{\Ci}\fHi(\Ci)\ptens\Ci)
\pcont \int  \Qijl{i}{\elec}{1}(|\adimg|^2) 
|\adimg| \adimg\ptens\adimg\ptens\partial_{\Ce}\fe(\adimg)\;\d\adimg\\
\displaybreak[0]
+ \frac{1}{4\mi^2}\partial_{\Ci}(\partial_{\Ci}\fHi(\Ci)\ptens\Ci)
\pcont \int  \Qijl{i}{\elec}{2}(|\adimg|^2) 
|\adimg| (|\adimg|^2\identite-3\adimg\ptens\adimg)\ptens
\partial_{\Ce}\fe(\adimg)\;\d\adimg\\
- \frac{1}{4\mi^2}\partial^3_{\Ci\Ci\Ci}\fHi(\Ci)
\pcont \int  \Qijl{i}{\elec}{1}(|\adimg|^2) 
|\adimg| (|\adimg|^2\identite+\adimg\ptens\adimg)\ptens\adimg
\fe(\adimg)\;\d\adimg\\
\displaybreak[0]
- \frac{1}{4\mi^2}\partial^3_{\Ci\Ci\Ci}\fHi(\Ci)
\pcont \int  \Qijl{i}{\elec}{2}(|\adimg|^2) 
|\adimg| (|\adimg|^2\identite-3\adimg\ptens\adimg)\ptens\adimg
\fe(\adimg)\;\d\adimg\\
+ \frac{1}{12\mi^2}\partial^3_{\Ci\Ci\Ci}\fHi(\Ci)
\pcont \int  \Qijl{i}{\elec}{3}(|\adimg|^2) 
|\adimg| (3|\adimg|^2\identite-5\adimg\ptens\adimg)\ptens\adimg
\fe(\adimg)\;\d\adimg\\
+ \frac{3}{2\mi} \Jieu(\fHi,\fe)(\Ci).
\label{Jiet}
\end{multline}
\end{theorem}

\begin{proof}
The change of variable $\d \Ce=-(1+\varepsilon^2/\mi)^{3/2}\d\adimg$ allows for the  differential cross-section  dependance on $\varepsilon$ to be eliminated
\begin{multline*}
\Jie(\fHi,\fe)(\Ci) = 
\int \seceff[ie]\!\left(|\adimg |^2,\ovec\dscal\tfrac{\adimg}{|\adimg |}\right)
|\adimg|{(1+\epsilon^2/\mi)}^{2}\\
\Bigl[ \fHi(\Ci')\fe(\Ce')-\fHi(\Ci)\fe(\Ce) \Bigr] \;\d\ovec \;\d\adimg, \quad i \in \lourd.
\end{multline*}
Then, the peculiar velocities are expanded in a power series of $\varepsilon$
\begin{align*}
\Ci'&=\Ci+\varepsilon\tfrac{1}{\mi}\avec-\varepsilon^3\tfrac{1}{2\mi^2}\avec+\ordre(\epsilon^4), &\avec &=-\adimg+|\adimg |\ovec, \quad i\in\lourd,\\
\Ce'&=-|\adimg | \ovec +  \epsilon\Ci+\varepsilon^2\tfrac{1}{\mi}\bvec
+\ordre(\epsilon^4),&\bvec&=-\adimg+\tfrac{1}{2}|\adimg |\ovec,\\
\Ce&=-\adimg+ \epsilon\Ci -\varepsilon^2\tfrac{1}{2\mi}\adimg
+\ordre(\epsilon^4).
\end{align*}
Hence, the distribution functions are found to be
\begin{multline*}
\fHi(\Ci')=\fHi(\Ci)+\varepsilon\tfrac{1}{\mi}\partial_{\Ci}\fHi(\Ci)\dscal\avec  +\varepsilon^2\tfrac{1}{2\mi^2}\partial_{\Ci\Ci}^2\fHi(\Ci)\pmat \avec\ptens\avec\\
\null+\varepsilon^3 \tfrac{1}{6\mi^3}\partial_{\Ci\Ci\Ci}^3\fHi(\Ci)\pcont \avec\ptens\avec\ptens\avec
-\varepsilon^3 \tfrac{1}{2\mi^2}\partial_{\Ci}\fHi(\Ci)\dscal\avec
+\ordre(\epsilon^4), \quad i\in\lourd,
\end{multline*}
\begin{multline*}
\fe(\Ce')=\fe(-|\adimg | \ovec )+\varepsilon\partial_{\Ce}  \fe(-|\adimg | \ovec )\dscal\Ci
+\varepsilon^2 \tfrac{1}{2}\partial_{\Ce\Ce}^2\fe(-|\adimg | \ovec )\pmat\Ci\ptens\Ci \\
\null+\varepsilon^2 \tfrac{1}{\mi}\partial_{\Ce}  \fe(-|\adimg | \ovec )\dscal\bvec 
+\varepsilon^3\tfrac{1}{6}\partial_{\Ce\Ce\Ce}^3 \fe(-|\adimg | \ovec )\pcont \Ci\ptens\Ci\ptens\Ci\\
\null+\varepsilon^3 \tfrac{1}{\mi}\partial_{\Ce\Ce}^2 \fe(-|\adimg | \ovec )\pmat
\Ci\ptens\bvec+\ordre(\epsilon^4),
\end{multline*}
\begin{multline*}
\fe(\Ce)=\fe(-\adimg )+\varepsilon\partial_{\Ce}  \fe(-\adimg)\dscal\Ci+
\varepsilon^2\tfrac{1}{2}\partial_{\Ce\Ce}^2\fe(-\adimg)\pmat\Ci\ptens\Ci\\
\null-\varepsilon^2\tfrac{1}{2\mi}\partial_{\Ce}\fe(-\adimg)\dscal\adimg
+\varepsilon^3\tfrac{1}{6}\partial_{\Ce\Ce\Ce}^3 \fe(-\adimg)\pcont \Ci\ptens\Ci\ptens\Ci\\
\null-\varepsilon^3\tfrac{1}{2\mi}\partial_{\Ce\Ce}^2 \fe(-\adimg)\pmat\Ci\ptens\adimg
+\ordre(\epsilon^4).
\end{multline*}
Combining these equations, the zero-order term $\Jieo$, $i\in\lourd$, is thus given by 
\begin{multline*}
\Jieo(\fHi,\fe)(\Ci) = \fHi(\Ci)\int \seceff[ie]\!\left(|\adimg |^2,\ovec\dscal\evec\right)\\
|\adimg|^3 \Bigl[ \fe(-|\adimg | \ovec )-\fe(-|\adimg|\evec) \Bigr] \;\d\ovec \;\d\evec\;\d|\adimg|.
\end{multline*}
Intertwining $\evec$ with $\ovec$, the integral is shown to vanish. Then, eqs.~\eqref{Jieu}-\eqref{Jiet} are obtained after some lengthy calculation.
\end{proof}

Theorem~\ref{th:thmJie} admits three corollaries.

\begin{corollary}\label{th:Jieuiso}
The first-order collision operator $\Jieu(\fHi,\fe)$, $i\in\lourd$, vanishes when $\fe$ is an isotropic function of the velocity $\Ce$.
\end{corollary}

\begin{proof}
Expression \eqref{Jieu} immediately yields that the integrand is an odd function of $\adimg$ if $\fe$ is isotropic in the mean heavy-particle frame, so that the first-order collision operator vanishes.
\end{proof}

\begin{remark}\label{th:rem1}
So far, we note that this property is strongly related to our choice of referential. For example, such a property is not satisfied when $\velocity=0$.
Thus, the structure of the expansion of collisional integrals depends on the initial choice of  referential.
We will come back to this point in Section \ref{sec:secsubjustif}.
\end{remark}

A collision frequency is defined as a Maxwell-Boltzmann averaged momentum cross-section
\begin{equation*}
\colfreqiez = \frac{1}{\tempe}\int \Qijl i\elec 1(|\adimg|^2) 
|\adimg |^3 \feo(\adimg) \;\d\adimg, \quad i\in\lourd,
\end{equation*}
where $\feo(\adimg)=\ne\exp\left[-\adimg\dscal\adimg/(2\tempe)\right]/ \left(2\pi\tempe\right)^{3/2}$.

\begin{corollary}\label{th:Jiediso}
If $\feo=\ne\exp\left[-\Ce\dscal\Ce/(2\tempe)\right]/ \left(2\pi\tempe\right)^{3/2}$, the second-order collision operator reads
\begin{equation}
\Jied(\fHi,\feo)(\Ci) = \frac{\colfreqiez}{3\mi}\Bigl(\dCi \dscal (\fHi\Ci)+\frac{\tempe}{\mi}\laplaceCi\fHi\Bigr), \quad i\in\lourd.
\label{Jiedfeo}
\end{equation}
\end{corollary}

\begin{proof}
A direct calculation of $\Jied(\fHi,\feo)(\Ci)$ given in \eqref{Jied} immediately yields expression \eqref{Jiedfeo} if 
$\feo=\ne\exp\left[-\Ce\dscal\Ce/(2\tempe)\right]/ \left(2\pi\tempe\right)^{3/2}$.
\end{proof}

\begin{corollary}\label{th:Jietiso}
The third-order collision operator $\Jiet(\fHi,\fe)$, $i\in\lourd$, vanishes when $\fe$ is an isotropic function of the velocity $\Ce$.
\end{corollary}

\begin{proof}
Expression \eqref{Jiet} immediately yields that the integrand is an odd function of $\adimg$ if $\fe$ is isotropic in the mean heavy-particle frame, so that the third-order collision operator vanishes.
\end{proof}

\subsection{Expansion of the collision operator $\Jej[i]$}

The dimensional analysis yields the following expression of the nondimensional collision operator $\Jej[i]$, $i\in\lourd$,
\begin{multline}
\label{Jei}
\Jej[i](\fe,\fHi)(\Ce) = \int \seceff[ei]\Bigl( \tfrac{\mi|\Ce-\epsilon\Ci|^2}{\mi+\epsilon^2},\ovec\dscal\evec \Bigr)
|\Ce-\epsilon\Ci| \\
\Bigl[ \fe(\Ce')\fHi(\Ci')-\fe(\Ce)\fHi(\Ci) \Bigr] \;\d\ovec \;\d\Ci.
\end{multline}
The original set of variables $\{\Ce, \Ci, \ovec\}$ is retained. We introduce the momentum cross-section
\begin{equation}
\Qijl \elec i 1 (|\Ce|^2) = 2\pi \int_0^{\pi} \seceff[\elec i](|\Ce|^2,\cos\theta)
(1-\cos\theta)\sin\theta \;\d\theta,\quad i\in\lourd,
\end{equation}
representing the average momentum transfered in encounters between electrons and heavy particles  $i\in\lourd$.  It is equal to the cross-section  $\Qijl i\elec 1$.

\begin{theorem}\label{th:theoJei}
The collision operator $\Jej[i]$, $i\in\lourd$, can be expanded in the form
\begin{multline}
\label{Jeiexpanded}
\Jej[i](\fe,\fHi)(\Ce)=\Jejo[i](\fe,\fHi)(\Ce)+\epsilon\Jeju[i](\fe,\fHi)(\Ce)+\epsilon^2\Jejd[i](\fe,\fHi)(\Ce)\\+\epsilon^3\Jejt[i](\fe,\fHi)(\Ce)+\ordre(\epsilon^4).
\end{multline}
The zero-order term $\Jejo[i]$, $i\in\lourd$, is given by the expression
\begin{equation}\label{Jeio}
\Jejo[i](\fe,\fHi)(\Ce)=\int \fHi(\Ci)\;\d\Ci
\int\seceff[\elec i] \Bigl(|\Ce|^2,\ovec\dscal\tfrac{\Ce}{|\Ce|}\Bigr) |\Ce|
\Bigl[ \fe(|\Ce|\ovec) - \fe(\Ce) \Bigr] \d\ovec.
\end{equation}
The first-order term $\Jieu$, $i\in\lourd$, reads
\begin{multline}\label{Jeiu}
\Jeju[i](\fe,\fHi)(\Ce)= \Bigl(\int \fHi(\Ci)\Ci \;\d\Ci\Bigr) \dscal \\
\bigg\{ \partial_{\Ce} 
\int 
\seceff[ei](|\Ce|^2,\tfrac{\Ce}{|\Ce|}\dscal\ovec)
\bigl[\fe(\Ce) - \fe(|\Ce|\ovec)\bigr] |\Ce| \;\d\ovec \\
+ \int
\seceff[\elec i](|\Ce|^2,\tfrac{\Ce}{|\Ce|}\dscal\ovec) |\Ce|
\bigl[\partial_{\Ce} \fe(|\Ce|\ovec) - \partial_{\Ce}\fe(\Ce)\bigr]
\;\d\ovec \bigg\}.
\end{multline}
The second order term $\Jied$, $i\in\lourd$,  is found to be
\begin{equation}\label{Jeid}
\Jejd[i](\fe,\fHi)(\Ce) =
\tfrac{1}{\mi} K_{\elec i}^{2,1}(\Ce) \int\fHi(\Ci)\;\d\Ci 
+ \tfrac{1}{2} {\boldsymbol K}_{\elec i}^{2,2}(\Ce) \;
\pmat \int\fHi(\Ci)\Ci\ptens\Ci\;\d\Ci ,
\end{equation}
with
\begin{multline*}
K_{\elec i}^{2,1} (\Ce) = 
\partial_{\Ce} \dscal\int 
\seceff[\elec i](|\Ce|^2,\tfrac{\Ce}{|\Ce|}\dscal\ovec)
(\Ce - |\Ce|\ovec) |\Ce| \fe(|\Ce|\ovec) \;\d\ovec \\
- |\Ce|\Ce \dscal \int
\partial_{\Ce} \seceff[\elec i](|\Ce|^2,\tfrac{\Ce}{|\Ce|}\dscal\ovec)
\bigl[ \fe(|\Ce|\ovec)-\fe(\Ce)\bigr] \;\d\ovec,
\end{multline*}
and
\begin{multline*}
{\boldsymbol K}_{\elec i}^{2,2} (\Ce) = 
\partial_{\Ce\Ce}^2 \int 
\seceff[\elec i](|\Ce|^2,\tfrac{\Ce}{|\Ce|}\dscal\ovec)
|\Ce| \bigl[\fe(\Ce|\ovec)-\fe(\Ce)\bigr] \;\d\ovec \\
+ 2 \int
\partial_{\Ce} \bigl(\seceff[\elec i](|\Ce|^2,\tfrac{\Ce}{|\Ce|}\dscal\ovec)|\Ce|\bigr)
\ptens
\bigl[ \partial_{\Ce}\fe(\Ce) - \partial_{\Ce}\fe(|\Ce|\ovec)\bigr] \;\d\ovec \\
+ |\Ce| \int 
\seceff[\elec i](|\Ce|^2,\tfrac{\Ce}{|\Ce|}\dscal\ovec)
\bigl[\partial_{\Ce\Ce}^2\fe(\Ce) - \partial_{\Ce\Ce}^2\fe(|\Ce|\ovec)\bigr] \;\d\ovec \\
+ 2 |\Ce| \int 
\seceff[\elec i](|\Ce|^2,\tfrac{\Ce}{|\Ce|}\dscal\ovec)
\tfrac{\Ce}{|\Ce|}\ptens\ovec \partial_{\Ce\Ce}^2\fe(|\Ce|\ovec) \;\d\ovec.
\end{multline*}
\end{theorem}

\begin{proof}
The relative velocity and  peculiar velocities after collision are expanded in a power series of $\varepsilon$. For $i\in\lourd$, we have
\begin{gather*}
|\Ce-\epsilon\Ci|=|\Ce|-\varepsilon\tfrac{\Ce}{|\Ce|}\dscal\Ci
+\varepsilon^2\bscal+\ordre(\epsilon^3),\\
\Ci'=\Ci+\varepsilon\tfrac{1}{\mi}\cvec-\varepsilon^2\tfrac{1}{\mi}\dvec
+\ordre(\epsilon^3),\\
\Ce'=|\Ce| \ovec +  \epsilon\dvec+\varepsilon^2
\bigl(\tfrac{1}{\mi}\cvec+\aavec\bigr)+\ordre(\epsilon^3),
\end{gather*}
with
$\bscal=\tfrac{1}{2|\Ce|}\Bigl[|\Ci|^2-\bigl(\tfrac{\Ce}{|\Ce|}\dscal\Ci\bigr)^2\Bigr]$,
$\cvec =\Ce-|\Ce|\ovec$,
$\dvec=\Ci-\tfrac{\Ce}{|\Ce|}\dscal\Ci\ovec$,
$\aavec=\bscal\ovec$.
Hence, the distribution functions are found to be
\begin{multline*}
\fHi(\Ci')=\fHi(\Ci)+\varepsilon\tfrac{1}{\mi}\partial_{\Ci}\fHi(\Ci)\dscal\cvec  +\varepsilon^2
\tfrac{1}{2\mi^2}\partial_{\Ci\Ci}^2\fHi(\Ci)\pmat \cvec\ptens\cvec\\
-\varepsilon^2\tfrac{1}{2\mi}\partial_{\Ci}\fHi(\Ci)\dscal\dvec +\ordre(\epsilon^3), \quad i\in\lourd,
\end{multline*}
\begin{multline*}
\fe(\Ce')=\fe(|\Ce| \ovec )+\varepsilon\partial_{\Ce}  \fe(|\Ce| \ovec )\dscal\dvec
+\frac{1}{2}\varepsilon^2\partial_{\Ce\Ce}^2\fe(|\Ce|\ovec)\pmat\dvec\ptens\dvec\\
+\varepsilon^2 \partial_{\Ce} \fe(|\Ce|\ovec)\dscal\left(\tfrac{1}{\mi}\cvec+\aavec\right)
+\ordre(\epsilon^3).
\end{multline*}
Combining these equations, we obtain eqs.~\eqref{Jeio}-\eqref{Jeid} after some lengthy calculation.
\end{proof}

Theorem~\ref{th:theoJei} admits three corollaries. First, we define the entropy produced  at order $\varepsilon^0$ in collisions between electrons and $i$ heavy particles 
\begin{equation*}
\prodenteio = -\int \Jejo[i](\fe,\fHi)(\Ce) \;\ln\left[\frac{(2\pi)^{3/2}n^0}{\partitione}\fe(\Ce)\right]\;\d\Ce, \quad i\in\lourd,
\end{equation*}
where $\partitione=({2\pi \me^0\boltz T^0}/{\planck^2})^{3/2}$ is the translational partition function of electrons. Symbol $\planck$ stands for Planck's constant. 
The zero-order operator describes the relaxation of the electron populations towards an isotropic distribution function in the mean heavy-particle frame. 

\begin{corollary}\label{th:prodenteio}
The zero-order collision operator $\Jejo[i](\fe,\fHi)$, $i\in\lourd$, vanishes when $\fe$ is an isotropic function of the velocity $\Ce$.
Moreover, the zero-order entropy is nonnegative, that writes $\prodenteio  \geq 0$, $i\in\lourd$, and the inequality is an equality if and only if $\fe$ is an isotropic function of the velocity $\Ce$.
\end{corollary}

\begin{proof}
If $\fe$ is an isotropic function of $\Ce$, we have
$ \fe\left(|\Ce|\ovec\right) = \fe\left(\Ce\right) $ for any $\ovec$ in the unit sphere,
so that expression~\eqref{Jeio} implies that $\Jejo[i](\fe,\fHi)=0$.
The zero-order production entropy production reads
\begin{multline*}
\prodenteio  =
- \ni\int \seceff[\elec i] \bigl(|\Ce|^2,\ovec\dscal\tfrac{\Ce}{|\Ce|}\bigr) \; |\Ce|^3
\Bigl[ \fe(|\Ce|\ovec) - \fe(\Ce) \Bigr] \\
\ln\left[\frac{(2\pi)^{3/2}n^0}{\partitione} \fe(\Ce) \right]\d|\Ce|\d\tfrac{\Ce}{|\Ce|}\d\ovec,
\end{multline*}
and intertwining $\tfrac{\Ce}{|\Ce|}$ with $\ovec$,
\begin{equation*}
\prodenteio= 
\frac{\ni}{2} \int \seceff[\elec i] \bigl(|\Ce|^2,\ovec\dscal\tfrac{\Ce}{|\Ce|}\bigr) \; |\Ce|^3
\Omega\Bigl( \fe(|\Ce|\ovec),\fe(\Ce) \Bigr) \d|\Ce|\d\tfrac{\Ce}{|\Ce|}\d\ovec.
\end{equation*}
where $\Omega(x,y)=(x-y)\ln(x/y)$ is a nonnegative function.
We then obtain that $\prodenteio$, $i\in\lourd$, is nonnegative and equal to 0 if and only if $\fe$ is isotropic in the mean heavy-particle frame.
\end{proof}

\begin{corollary}\label{th:Jeju}
The first-order collision operator $\Jeju[i](\fe,\fHi)$, $i\in\lourd$, vanishes when $\fHi$ is an isotropic function of the velocity $\Ci$. 
\end{corollary}

\begin{proof}
Expression \eqref{Jeiu} immediately yields that the integrand is an odd function of $\Ci$, $i \in\lourd$, if $\fHi$ is isotropic in the mean heavy-particle frame, so that the first-order collision operator vanishes.
\end{proof}

\begin{corollary}\label{th:Jeidiso}
Considering $\feo=\ne\exp\left[-\Ce\dscal\Ce/(2\tempe)\right]/ \left(2\pi\tempe\right)^{3/2}$ and
$\fHio=\ni\mi^{3/2} \exp\left[-\mi\Ci\dscal\Ci/(2\temph)\right]/ \left(2\pi\temph\right)^{3/2}$, $ i\in\lourd$, 
 the second-order collision operator  $\Jejd[i](\feo,\fHio)(\Ce)$, $i\in\lourd$, reads
\begin{multline}
\label{Jeidfeofio}
\Jejd[i](\feo,\fHio)(\Ce) = 
(\temph-\tempe)\frac{\ni}{\mi}\frac{1}{\tempe} \feo(\Ce)|\Ce| \\
\Bigl[ \partial_{\Ce}\dscal 
\bigl( \Qijl \elec i 1 (\Ce) \Ce \bigr)
 + \bigl(1- \tfrac{|\Ce|^2}{\tempe} \bigr)  \Qijl \elec i 1 (\Ce) \Bigr].
\end{multline}
\end{corollary}

\begin{proof}
A direct calculation of $\Jejd[i](\feo,\fHi)(\Ce)$ given in \eqref{Jeid} immediately yields expression \eqref{Jeidfeofio} if 
$\feo$ and $\fHio$ are the Maxwell-Boltzmann distribution functions given in the assumptions of corollary~\ref{th:Jeidiso}.
\end{proof}

\subsection{Electron and heavy-particle collisional invariants}

In the asymptotic limit $\varepsilon\rightarrow 0$, the space of collisional
invariants $\invspace$ defined in eq.~\eqref{eqcolinveps} splits into 
two subspaces naturally associated with our choice of scaling.
\begin{definition}
The space of scalar electron collisional invariants $\invspace_\elec$ is spanned by the following elements
\begin{equation}
\left\{
\vcenter{\halign{
$#\hfil$&$\;\,=\; #\hfil$
\cr
\inve[1]&1,
\cr\noalign{\vskip4pt}
\inve[2]&\tfrac{1}{2}\Ce\dscal\Ce.
\cr}} \right.
\end{equation}
\end{definition}

\begin{definition}
The space of scalar heavy-particle collisional invariants $\invspace_\heavy$ is spanned by the following elements
\begin{equation}
\left\{
\vcenter{\halign{
$#\hfil$&$\;\,=\; \bigl(#\bigr)_{i\in\lourd},\hfil$&$\qquad#,\hfil$
\cr
\invH[j]&\mi\delta_{ij}&j\in\lourd
\cr\noalign{\vskip4pt}
\invH[\nh+\nu]&\mi\Cinorme[i\nu]&\nu\in\{1,2,3\}
\cr\noalign{\vskip4pt}
\invH[\nh+4]&\tfrac{1}{2}\mi\Ci\dscal\Ci
\cr}} \right.
\end{equation}
where symbol $\nh$ denotes the cardinality of the set $\lourd$. 
\end{definition} 
The decoupling of the collision invariants is clearly identified in the proposed scaling. 
More precisely,  the definition of the electron linearized collision operator (given in Section~\ref{sec:sec4}) will involve the electron partial collision operator $\Jee$ and the mixed partial collision operators $\Jejo[i],~i\in\lourd$, satisfying the following important property.

\begin{property}
\label{th:corJeo}
The partial collision operators $\Jejo[i],~i\in\lourd$, are orthogonal to the space of collisional invariants $\invspace_\elec$, $i.e.$, $\ppscale{\inve[l],\Jejo[i]}=0$  for all $l\in\{1,2\}$. 
\end{property}

\begin{proof}
The projection of the collision operator $\Jejo[i]$, $i\in\lourd$, onto $\inve[l]$, $l\in\{1,2\}$  reads
\begin{equation*}
\ppscale{\inve[l],\Jejo[i]}=\ni\int\seceff[\elec i] \Bigl(|\Ce |^2,\ovec\dscal\tfrac{\Ce}{|\Ce|}\Bigr) \; |\Ce|^3
\Bigl[ \fe\left(|\Ce|\ovec\right) - \fe\left(\Ce\right) \Bigr]\inve[l]\;\d|\Ce|\d\ovec\d\tfrac{\Ce}{|\Ce|}.
\end{equation*}
Intertwining $\ovec$ with $\tfrac{\Ce}{|\Ce|}$, the projection $\ppscale{\inve[l],\Jejo[i]}$ is shown to vanish for all  $l\in\{1,2\}$. 
\end{proof}
 We underline that the partial collision operators $\Jejo[i],~i\in\lourd$, are not orthogonal for the scalar product $\ppscale{\cdot,\cdot}$ to the space spanned by the electron momentum.  It is the reason why 
the  vector $\Ce$ does not belong to $\invspace_\elec$. 
In contrast,   the definition of the heavy-particle linearized collision operator (given in Section~\ref{sec:sec4}) will only involve the heavy-particle partial collision operators $\Jij$, $i,j\in\lourd$. 

Subsequently, using the newly defined collisional invariants, the orthogonality property~\ref{th:corcolinv} of the cross-collision operators can be rewritten
\begin{equation}\label{eqcrossedcol1}
\sumi[j]\ppscale{\inve[1],\Jej}=0,\quad\ppscalh{\invH[i],\Jhe}=0,\quad i \in\lourd,
\end{equation}
for mass conservation,
\begin{equation}\label{eqcrossedcol2}
\varepsilon{\sumi[j]}\ppscale{C_{\elec\nu},\Jej}+\ppscalh{\invH[\nh+\nu],\Jhe}=0,\quad\nu\in\{1,2,3\},
\end{equation}
for momentum conservation, and
\begin{equation}\label{eqcrossedcol3}
\sumi[j]\ppscale{\inve[2],\Jej}+\ppscalh{\invH[\nh+4],\Jhe}=0,
\end{equation}
for energy conservation.
This set of relations is essential since it corresponds to the conservation of mass, momentum, and energy  in the
electron heavy-particle interactions through the various orders in $\varepsilon$ 
of the Chapman-Enskog expansion.

Then, the macroscopic properties are expressed  as partial scalar products of the distribution functions and the new collisional invariants
\begin{equation}\label{eqcontelec}
\left\{
\vcenter{\halign{
$\ppscale{\fe,#}{}\hfil$&$\;\,=\; #,\hfil$
\cr
\inve[1]&\rhoe
\cr\noalign{\vskip4pt}
\inve[2]&\rhoe \energiee
\cr}} \right.
\end{equation}
 and 
\begin{equation}\label{eqconsvelocity}
\left\{
\vcenter{\halign{
$\ppscalh{\fH,#}{}\hfil$&$\;\,=\; #\hfil$&$\qquad#,\hfil$
\cr
\invH[i]&\rhoi,&i\in\lourd
\cr\noalign{\vskip4pt}
\invH[\nh+\nu]&0,&\nu\in\{1,2,3\}
\cr\noalign{\vskip4pt}
\invH[\nh+4]&\rhoi[\heavy]\energiei[\heavy].
\cr}} \right.
\end{equation}
Symbol $\energiee$ stands for the electron thermal energy per unit mass  and $\energiei[\heavy]$, the heavy-particle thermal energy per unit mass.  It is important to mention that the later quantities are defined in the mean heavy-particle frame. Furthermore, the decoupling of the collisional invariants is also consistent with the expression of the macroscopic properties. In particular, the electron momentum is not a collision invariant in the proposed asymptotic limit,  the electron mass flux is not constrained in the mean heavy-particle velocity referential.

Moreover, translational temperatures are introduced as averaged thermal energies in the mean heavy-particle frame.
\begin{definition}
The electron and heavy-particle translational temperatures are given by
\begin{align}
\tempe&=\frac{2}{3\ne}\ppscale{\fe,\inve[2]},\\
\temph&=\frac{2}{3\nH}\ppscalh{\fH,\invH[\nh+4]},
\end{align}
where the heavy-particle number density reads $\nH=\sum_{j\in\lourd}\ni[j]$. 
\end{definition} 
Consequently, the energy can be rewritten
\begin{equation*}
\ppscale{\fe,\inve[2]}
=\frac{3}{2} \ne\tempe,
\end{equation*}
for the electrons, and
\begin{equation*}
\ppscalh{\fH,\invH[\nh+4]}=\frac{3}{2}\nH\temph
\end{equation*}
for the heavy particles. Moreover, it will be shown in Section~\ref{sec:sec4} that both temperatures are generally different.

\section{Chapman-Enskog method}\label{sec:sec4}

We resort to an Enskog expansion to derive an approximate solution to the Boltzmann equations~\eqref{eqbol1}-\eqref{eqbol2} by expanding the species distribution functions as
\begin{align}
\label{chapenske}
\fe&=\feo(1+\epsilon\phie+\epsilon^2\phide+\epsilon^3\phite) + \ordre(\epsilon^4),\\
\label{chapenski}
\fHi&=\fHio(1+\epsilon\phii+\epsilon^2\phidi) + \ordre(\epsilon^3),\qquad i\in\lourd,
\end{align}
and by imposing that the zero-order contributions $\feo$ and $\fHo$  yield the local macroscopic properties
\begin{align} 
\label{conste}
\ppscale{\feo,\inve[l]} &= \ppscale{\fe,\inve[l]}, &&l\in\{1,2\}, \\
\label{constH}
\ppscalh{\fHo,\invH[l]} &= \ppscalh{\fH,\invH[l]}, &&l\in\{1,\ldots,\nh{+}4\}.
\end{align}

Hence, based upon the dimensional analysis of section~\ref{sec:secdiman}, the electron Boltzmann equation~\eqref{eqbol1} appears to be
\begin{multline}
\label{eqbolbole}
\epsilon^{-2}\Demd(\feo) + \epsilon^{-1}\Demu(\feo,\phie) + \Dez(\feo,\phie,\phide)
+ \epsilon \Depu(\feo,\phie,\phide,\phite) \\
= \epsilon^{-2} \Jemd + \epsilon^{-1} \Jemu + \Jez + \epsilon \Jepu + \ordre(\epsilon^2),
\end{multline}
where the electron streaming operators read at successive orders 
\begin{align*}
\displaybreak[0]
\Demd(\feo) &= \delta_{b1}\qe (\Ce\pvect\B)\dscal\dCe\feo,
\displaybreak[0] \\
\Demu(\feo,\phie) &= \Demuhat(\feo)
+ \qe (\delta_{b1}\Ce\pvect\B)\dscal\dCe(\feo\phie),
\displaybreak[0] \\
\Demuhat(\feo) &= 
\tfrac{1}{\Mh}\Ce\dscal\dx\feo
+ \qe \bigl( \tfrac{1}{\Mh}\Ep + \delta_{b0}\Ce\pvect\B \bigr) \dscal\dCe\feo,
\displaybreak[0]  \\
\Dez(\feo,\phie,\phide) &=\Dezhat(\feo,\phie) + \qe (\delta_{b1}\Ce\pvect\B)\dscal\dCe(\feo\phide),
\displaybreak[0]  \\
\Dezhat(\feo,\phie) & =
\dt\feo + \tfrac{1}{\Mh}\Ce\dscal\dx(\feo\phie) + \vitesse\dscal\dx\feo  - \bigl(\dCe\feo\ptens\Ce\bigr)\pmat\dx\vitesse \\
& \phantom{ = } + \qe\bigl(\delta_{b0}\Mh\vitesse\pvect\B+\delta_{b(-1)}\Ce\pvect\B\bigr)\dscal\dCe\feo \\
& \phantom{ = } + \qe \bigl( \tfrac{1}{\Mh}\Ep + \delta_{b0}\Ce\pvect\B \bigr) \dscal\dCe(\feo\phie),
\displaybreak[0] \\
\Depu(\feo,\phie,\phide,\phite) &= 
\Depuhat(\feo,\phie,\phide) + \qe (\delta_{b1}\Ce\pvect\B)\dscal\dCe(\feo\phite), 
\displaybreak[0] \\
\Depuhat(\feo,\phie,\phide) &= 
\dt(\feo\phie) + \tfrac{1}{\Mh}\Ce\dscal\dx(\feo\phide) + \vitesse\dscal\dx(\feo\phie) \\
& \phantom{ = }
- \Mh\tfrac{\D\vitesse}{\D\temps}\dscal\dCe\feo
- \bigl(\dCe(\feo\phie)\ptens\Ce\bigr)\pmat\dx\vitesse\\
&\phantom{ = }
+\qe\bigl(\delta_{b(-1)}\Mh\vitesse\pvect\B+\delta_{b(-2)}\Ce\pvect\B\bigr)\dscal\dCe\feo\\
& \phantom{ = } + \qe\bigl(\delta_{b0}\Mh\vitesse\pvect\B+\delta_{b(-1)}\Ce\pvect\B\bigr)\dscal\dCe(\feo\phie)\\
& \phantom{ = }+ \qe \bigl( \tfrac{1}{\Mh}\Ep + \delta_{b0}\Ce\pvect\B \bigr) \dscal\dCe(\feo\phide),
\end{align*}
with the electric field expressed into the mean heavy-particle frame 
$\Ep=\E+\delta_{b1}\Mh^2\vitesse\pvect\B$. 
The electron collision operators are given   by 
\begin{align*}
\Jemd &= \cancel{\Jee(\feo,\feo)}+ \sumi[j] \cancel{\Jejo(\feo,\fHio[j])},
\displaybreak[0]
\intertext{for ease of readability in Sections~\ref{sec:secsub1}-\ref{sec:secsub2}, we strike through the collision operators that vanish when $\feo$ and $\fHio$, $i\in\lourd$, are isotropic functions,}
\Jemu &= \Jee(\feo\phie,\feo) + \Jee(\feo,\feo\phie) \\
& \phantom{=} 
+ {\sumi[j]}\; \Jejo(\feo\phie,\fHio[j]) + \cancel{\Jejo(\feo,\fHio[j]\phii[j])} + \cancel{\Jeju(\feo,\fHio[j])} \displaybreak[0] \\
\Jez &= \Jee(\feo\phide,\feo) + \Jee(\feo\phie,\feo\phie) + \Jee(\feo,\feo\phide)\\
& \phantom{ = }+ {\sumi[j]}\; \Jejo(\feo\phide,\fHio[j]) + \cancel{\Jejo(\feo,\fHio[j]\phidi[j])}+\Jezhat, \displaybreak[0] \\
\Jezhat&= {\sumi[j]}\; \Jejo(\feo\phie,\fHio[j]\phii[j])  + \cancel{\Jeju(\feo\phie,\fHio[j])} + \Jeju(\feo,\fHio[j]\phii[j]) + \Jejd(\feo,\fHio[j]), \displaybreak[0] \\
\Jepu &= \Jee(\feo\phite,\feo) 
+ \Jee(\feo\phide,\feo\phie) + \Jee(\feo\phie,\feo\phide) + \Jee(\feo,\feo\phite) \\
&\phantom{ = }+ {\sumi[j]}\; \Jejo(\feo\phite,\fHio[j]) + \cancel{\Jejo(\feo,\fHio[j]\phiti[j])}+\Jepuhat, \displaybreak[0] \\
\Jepuhat&=\smash{\sumi[j]}\; \Big\{
 \Jejo(\feo\phide,\fHio[j]\phii[j]) + \Jejo(\feo\phie,\fHio[j]\phidi[j]) 
+\cancel{\Jeju(\feo\phide,\fHio[j])} + \Jeju(\feo\phie,\fHio[j]\phii[j]) \\
& \qquad\qquad
\Jeju(\feo,\fHio[j]\phidi[j]) +\Jejd(\feo\phie,\fHio[j]) + \Jejd(\feo,\fHio[j]\phii[j]) + {\Jejt(\feo,\fHio[j])} \Big\}.
\end{align*}
Likewise, the heavy-particle Boltzmann equation~\eqref{eqbol2} is found to be 
\begin{equation}
\label{eqbolbolh}
\Diz(\fHio) + \epsilon\Diu(\fHio\!,\phii) = 
\epsilon^{-1} \Jimu + \Jiz + \epsilon \Jipu + \ordre(\epsilon^2), \qquad i\in\lourd,
\end{equation}
where the heavy-particle streaming operators read at successive orders 
\begin{multline*}
\Diz(\fHio) =
\dt\fHio + \bigl(\tfrac{1}{\Mh}\Ci{+}\vitesse\bigr)\dscal\dx\fHio 
+ \tfrac{\qi}{\mi} \bigl( \tfrac{1}{\Mh}\Ep + \delta_{b1} \Ci\pvect\B \bigr) \dscal\dCi\fHio \\
- \Mh \tfrac{\D\vitesse}{\D\temps} \dscal\dCi\fHio
- \bigl(\dCi\fHio\ptens\Ci\bigr)\pmat\dx\vitesse,
\end{multline*}
\begin{multline*}
\Diu(\fHio\!,\phii) = 
\dt(\fHio\!\phii) + \bigl(\tfrac{1}{\Mh}\Ci{+}\vitesse\bigr)\dscal\dx(\fHio\!\phii) \\
+ \tfrac{\qi}{\mi} \delta_{b0} \bigl[(\Ci{+}\Mh\vitesse)\pvect\B\bigr]\dscal\dCi\fHio 
+ \tfrac{\qi}{\mi} \bigl( \tfrac{1}{\Mh}\Ep
+ \delta_{b1} \Ci\pvect\B \bigr) \dscal\dCi(\fHio\!\phii) \\
- \Mh \tfrac{\D\vitesse}{\D\temps} \dscal\dCi(\fHio\!\phii) 
- \bigl(\dCi(\fHio\!\phii)\ptens\Ci\bigr)\pmat\dx\vitesse.
\end{multline*}
The heavy-particle collision operators are given by
\begin{align*}
\displaybreak[0]
\Jimu &= {\sumi[j]} \; \cancel{\Jij(\fHio[i]\!,\fHio[j])}+\cancel{\Jieu(\fHio[i]\!,\feo)},\\
\displaybreak[0]
\Jiz &= {\sumi[j]}\; \Jij(\fHio\phii,\fHio[j]) + \Jij(\fHio\!,\fHio[j]\phii[j])+\cancel{ \Jieu(\fHio[i]\phii,\feo)} +\Jizhat,\\
\displaybreak[0]
\Jizhat &= \Jieu(\fHio[i]\!,\feo\phie) + \Jied(\fHio[i]\!,\feo),\\
\displaybreak[0]
\Jipu &= {\sumi[j]}\; \Jij(\fHio\phidi,\fHio[j]) + \Jij(\fHio\phii,\fHio[j]\phii[j]) + \Jij(\fHio\!,\fHio[j]\phidi[j])
+\cancel{ \Jieu(\fHio[i]\phidi,\feo)} +\Jipuhat,\\
\Jipuhat&=\Jieu(\fHio[i]\phii,\feo\phie) 
+ \Jieu(\fHio[i]\!,\feo\phide)
+ \Jied(\fHio[i]\phii,\feo) + \Jied(\fHio[i]\!,\feo\phie) +\cancel{\Jiet(\fHio[i]\!,\feo)}.
\end{align*}
In the Chapman-Enskog method, the plasma is observed at successive orders of the $\varepsilon$ parameter equivalent to as many time scales. 
The micro- and macroscopic equations derived at each order are reviewed in Table~\ref{tab12}. 
\begin{table}[th]
\caption{Chapman-Enskog steps. \label{tab12}}
{\small \begin{tabular}{@{}llll@{}} \toprule 
{\bf Order} &{\bf Time}& {\bf Heavy particles}&{\bf Electrons}\\
&&&\\
$\epsilon^{-2}$ &$t_\elec^0$&&Expression of $\feo$\\ 
&&&Thermalization ($\tempe$)\\
&&&\\
$\epsilon^{-1}$ &$t_\heavy^0$&Expression of $\fHio$, $i\in\lourd$&Equation for $\phie$\\
&&Thermalization ($\temph$) &Zero-order momentum relation \\
&&&\\
$\epsilon^{0}$&$t^0$ &Equation for $\phii$, $i\in\lourd$&Equation for $\phi_\elec^2$\\
&&Euler equations&Zero-order drift-diffusion equations\\
&&&First-order momentum relation\\
&&&\\
$\epsilon$&$\frac{t^0}{\varepsilon}$ &Navier-Stokes equations&First-order drift-diffusion equations\\
 \botrule
\end{tabular}}
\end{table}

\subsection{Order $\epsilon^{-2}$: electron thermalization}

We resolve the electron Boltzmann equation~\eqref{eqbolbole} at order $\epsilon^{-2}$ corresponding to the kinetic time scale $t_\elec^0$. The electron population is shown to thermalize  in the mean heavy-particle frame to a quasi-equilibrium state described by a Maxwell-Boltzmann distribution function  at temperature $\tempe$.  In contrast, heavy particles do not exhibit any ensemble property at this order.

\begin{proposition}
\label{th:propoelec}
The zero-order electron distribution function $\feo$, 
solution to eq.~\eqref{eqbolbole} at order $\epsilon^{-2}$, 
$i.e.$, $\Demd(\feo)=\Jemd$,
that satisfies the scalar constraints \eqref{conste}
is a Maxwell-Boltzmann distribution function at the electron temperature
\begin{equation}\label{eqbolte}
\feo = \ne \left(\frac{1}{2\pi\tempe}\right)^{3/2} 
\exp\left( - \frac{1}{2\tempe}\Ce\dscal\Ce\right).
\end{equation}
\end{proposition}

\begin{proof}
Multiplying the equation $\Demd(\feo)=\Jemd$ by $\ln\left[{(2\pi)^{3/2}n^0}\feo/{\partitione}\right]$ and integrating over $\d\Ce$ yields the entropy production 
\begin{equation*}
\prodentee^0+\sumi[j]\prodenteio[j] +\delta_{b1}\qe\!\!\int\!(\Ce\pvect\B)\dscal\dCe\feo \ln\left[{(2\pi)^{3/2}n^0}\feo/{\partitione}\right] \d\Ce=0,
\end{equation*}
with  $\prodentee^0=-\int \Jee(\feo,\feo)(\Ce) \;\ln\left[{(2\pi)^{3/2}n^0}\feo/{\partitione}\right]\;\d\Ce$. Using the equality $\dCe\feo \ln\left[{(2\pi)^{3/2}n^0}\feo/{\partitione}\right]=\dCe\left\{\feo \ln\left[{(2\pi)^{3/2}n^0}\feo/{\partitione}\right]-\feo\right\} $
and integrating by parts, the entropy production is found to be $\prodentee^0+\sum_{j\in\lourd}\prodenteio[j] =0$.
Moreover, a well-established derivation yields
\begin{equation*}
\prodentee^0= \frac{1}{4}\int \Omega(\feo\feuo,\feop\feuop)
|\Ce-\Ceu|\sigma_{\elec\elec 1} \d\ovec \d\Ce \d\Ceu \geq 0.
\end{equation*}
Using corollary \ref{th:prodenteio}, we first obtain that $\prodenteio\ge 0$, $i\in\lourd$, so that both terms $\prodentee^0=0$ and $\prodenteio=0$, $i\in\lourd$. Then,  corollary \ref{th:prodenteio} implies that $\feo$ is isotropic in the mean heavy-particle frame. Seeing that $\prodentee^0=0$, $\ln\feo$ is thus a collisional invariant, $i.e.$ is in the space $\invspace_\elec$. By using the macroscopic constraints, expression~\eqref{eqbolte} is readily obtained.
\end{proof}

The choice of the referential velocity in which electrons thermalize  
will turns out to be crucial for the rest of the resolution.  
In the $\velocity=\vitesse$ frame, the quasi-equilibrium electron velocity distribution function is isotropic and the electrons follow the bulk movement associated with the hydrodynamic velocity of the mixture, following a physically relevant scenario.
As already mentioned, the mean heavy-particle velocity $\vitesse$
does not depend on the
small $\varepsilon$ parameter while still being close to the actual hydrodynamic velocity $\speed$ of the entire mixture; this property is essential in order to conduct a 
rigourous multiscale analysis in the framework of the present Chapman-Enskog expansion.
The relevance of such a choice of referential will be thoroughly 
investigated in section \ref{sec:secsubjustif}.

\subsection{Order $\varepsilon^{-1}$: heavy-particle thermalization}
\label{sec:secsub1}

We resolve the heavy-particle Boltzmann equation~\eqref{eqbolbolh} at order $\epsilon^{-1}$ corresponding to the kinetic time scale $t_\heavy^0$. The heavy-particle population is shown to thermalize  in the mean heavy-particle frame to a quasi-equilibrium state described by a Maxwell-Boltzmann distribution function  at temperature $\temph$.  

\begin{proposition}
\label{th:propoheavy}
Considering $\feo$ given by eq.~\eqref{eqbolte},
the zero-order family of heavy-particle distribution functions $\fHo$
solution to eq.~\eqref{eqbolbolh} at order~$\epsilon^{-1}$, 
$i.e.$, $\Jimu=0$, $i\in\lourd$,
that satisfies the scalar constraints~\eqref{constH}
is a family of Maxwell-Boltzmann distribution functions at the heavy-particle temperature
\begin{equation}\label{eqbolti}
\fHio= \ni \left(\frac{\mi}{2\pi\temph}\right)^{3/2} 
\exp\left( - \frac{\mi}{2\temph}\Ci\dscal\Ci\right),\quad i\in\lourd.
\end{equation}
\end{proposition}

\begin{proof}
As the zero-order electron distribution function $\feo$ is isotropic in the mean heavy-particle frame, corollary~\ref{th:Jieuiso} yields that the heavy-particle Boltzmann equation~\eqref{eqbolbolh} reads at order $\epsilon^{-1}$
\begin{equation*}
\sumi[j] \Jij(\fHio[i]\!,\fHio[j]) = 0, \quad i\in\lourd.
\end{equation*}
After some classical algebra~\cite{chapman}, we obtain expression~\eqref{eqbolti} for the zero-order heavy-particle distribution functions. 
\end{proof}

Thus, propositions~\ref{th:propoelec} and \ref{th:propoheavy} allow for electron and heavy particles  quasi-equilibrium states to be obtained at different temperatures.

\subsection{Order $\varepsilon^{-1}$: electron momentum relation}
\label{sec:secsub3}

We conduct the resolution and 
derive a momentum relation based on the electron Boltzmann equation~\eqref{eqbolbole} at order $\epsilon^{-1}$ corresponding to the kinetic time scale $t_\heavy^0$. 
We then emphasize an original property of the Chapman-Enskog expansion at this order associated with both the absence of momentum constraint in eq. \ref{eqcontelec} and the multiscale analysis. 

With the previously obtained Maxwell-Boltzmann electron distribution function,  we first define
the electron linearized collision operator in the case.
\begin{definition}
The electron linearized collision operator $\boltFe$ reads
\begin{equation*}
\boltFe(\phie)=-\frac{1}{\feo}\Bigl[\Jee(\feo\phie,\feo)+\Jee(\feo,\feo\phie)+\sumi[j]\Jejo(\feo\phie,\fHio[j])\Bigr],
\end{equation*}
where $\feo$ is given by eq.~\eqref{eqbolte} and $\fHio$ by eq.~\eqref{eqbolti}.
\end{definition}

The kernel of this operator is given in the following property.

\begin{property}\label{th:kerIe}
The kernel of the linearized collision operator $\boltFe$ is the space of scalar electron collisional invariants 
$\invspace_\elec$.
\end{property}

\begin{proof}
The linearized collision operator $\boltFe$ is rewritten in the form
\begin{multline*}
\boltFe(\phie) = - \int \feuo
\bigl( \phiep+\phieup-\phie-\phieu \bigr) |\Ce-\Ceu| \seceff[\elec\elec 1] \;\d\ovec \d\Ceu \\
- \sumi[j]
\ni[j]\int\seceff[\elec j]\Bigl(|\Ce|^2,\ovec\dscal\tfrac{\Ce}{|\Ce|}\Bigr) |\Ce| \bigl( \phie(|\Ce|\ovec) - \phie(\Ce) \bigr) \d\ovec.
\end{multline*}
We then obtain that the space $\invspace_\elec$ is in the kernel of $\boltFe$. Reciprocally, if $\boltFe(\phie)=0$, multiplying $\boltFe(\phie)$ by $\feo\phie$ and integrating over $\d\Ce$ yields 
\begin{multline*}
 \frac{1}{4} \int \feo\feuo
\bigl( \phiep+\phieup-\phie-\phieu \bigr)^2 |\Ce-\Ceu| \seceff[\elec\elec 1] \;\d\ovec \d\Ce\d\Ceu \\
+ \frac{1}{2} \sumi[j]
\ni[j]\int\seceff[\elec j] \Bigl(|\Ce|^2,\ovec\dscal\tfrac{\Ce}{|\Ce|}\Bigr) |\Ce| \feo\bigl( \phie(|\Ce|\ovec) - \phie(\Ce) \bigr)^2 \d\ovec \d\Ce= 0,
\end{multline*}
so that $\phie$ is in the space $\invspace_\elec$.
\end{proof}

Based on corollaries \ref{th:prodenteio} and \ref{th:Jeju}, the electron Boltzmann equation~\eqref{eqbolbole} is found to be   at order $\epsilon^{-1}$
\begin{equation}\label{eqepsilon-1}
\feo\boltFe(\phie)+\delta_{b1}\qe\dCe(\feo\phie)\dscal\Ce\pvect\B=-\Demuhat(\feo),
\end{equation}
with the constraints
\begin{equation}\label{eqconsepsilon-1}
\ppscale{\feo\phie,\inve[l]} = 0, \qquad l\in\{1,2\}.
\end{equation}
The terms $\dCe(\feo\phie)\dscal\Ce\pvect\B$ and $\Demuhat(\feo)$ 
are orthogonal to the kernel of $\boltFe$ for the scalar product $\ppscale{\cdot,\cdot}$. Consequently, no macroscopic conservation equations of mass and energy are derived at this order.

In fact, for \emph{any} value of $\decall$, defining the shifted Maxwell-Boltzmann distribution
\begin{equation}
\label{distMBtranslate}
\fe^{\decall 0} = \ne \left(\frac{1}{2\pi\tempe}\right)^{3/2}
\exp\left( -\frac{1}{2\tempe}(\Ce-\varepsilon\Mh\decall)^2\right),
\end{equation}
we can expand it as a function of $\varepsilon$
\begin{equation}
\label{distMBtranslate_dev}
\fe^{\decall0} 
= 
\feo
\left(
1+\varepsilon \frac{\Mh}{\tempe}\Ce\dscal\decall
+ \varepsilon^2 \frac{\Mh^2}{2\tempe}\left[-\decall\dscal\decall + \frac{(\Ce\dscal\decall)^2}{\tempe}\right]   
\right)+
\ordre(\varepsilon^3),
\end{equation}
which still yields, at leading order,  the same distribution as defined in eq. \eqref{eqbolte}. 
We then realize that 
the Chapman-Enskog expansion can be rewritten in a different way at this order~:
\begin{equation}
\label{CEexp_decal}
\feo(1+\epsilon\phie+\epsilon^2\phide)
= 
\fe^{\decall 0}(1+\epsilon\phie^\decall+\epsilon^2\phie^{\decall 0})+\ordre(\epsilon^3), 
\end{equation}
with 
\begin{equation}
\label{phi_decal}
\phie = \phie^\decall + \frac{\Mh}{\tempe}\Ce\dscal\decall, \quad
\phide =\phie^{\decall 2} + \frac{\Mh}{\tempe}(\Ce\dscal\decall)\phie^\decall  + \frac{\Mh^2}{2\tempe}\left[-\decall\dscal\decall + \frac{(\Ce\dscal\decall)^2}{\tempe}\right].
\end{equation}

It is interesting to notice that, whatever  the choice of $\decall$, 
the part of the hydrodynamic velocity of the full mixture
\begin{equation*}
(\rhoi[\heavy]+\varepsilon^2\rhoe)\Mh\speed= \rhoi[\heavy]\Mh\vitesse + \varepsilon^2\rhoe\ve,
\end{equation*}
associated with the electrons $\rhoe\ve$ will be splitted into two parts 
at the same order of the multiscale expansion 
\begin{equation*}
\ve = \Mh\vitesse + \Ve+\ordre(\varepsilon) = \Mh(\vitesse + \decall) + {\Ve}^{\decall}+\ordre(\varepsilon), \quad 
{\Ve}^{\decall} = \Ve - \Mh\decall,
\end{equation*}
with $\rhoe{\Ve}^{\decall} = 
\int\Ce\fe^{\decall 0}\phie^{\decall}\d\Ce$. 
Thus,
as opposed to the classical expansion, since no momentum constraint is to be found for the electrons,
the definition of the mixture hydrodynamic velocity does not allow to uniquely define the 
electron diffusion velocities. 
In any case, the hydrodynamic velocity of the mixture is  $\vitesse$ at order 
$\varepsilon^{-1}$.
It is then necessary to properly delineate the possible choices for the $\decall$  velocity, which should not be confused with a change of referential, since it only influences the electron Chapman-Enskog expansion.

\begin{lemma}
\label{th:uniquefe}
In the chosen frame of reference, any velocity $\decall$ leads to a new definition of $\phie^\decall$
for which 
property \ref{th:kerIe} 
is preserved and thus leads to an equivalent resolubility condition for $\phie^\decall$
as for $\phie$.
Moreover, the resolution of $\phie^\decall$ is completely equivalent to the resolution of $\phie$.
\end{lemma}

\begin{proof}
It is sufficient to note that the difference $\delta\phie^\decall = \phie^\decall-\phie=-\Mh\Ce\dscal\decall/\tempe$
  is orthogonal to the collisional invariants $\ppscale{\feo\delta\phie^\decall,\inve[l]}=0$, $l\in\{1,2\}$.
\end{proof}

It is interesting to note that for our choice of moving frame $\velocity=\vitesse$, the electron thermalization naturally occurs in the ``appropriate'' referential in close connection to the physics of the problem and there is no 
need to use the abovementioned property in order to conduct the resolution at order $\varepsilon^{-1}$, therefore, we take $\decall=0$ in the following. We will also have to check the validity of such a strategy at higher orders; 
we will come back to this point in Section~\ref{sec:secordzeroelec}.

As mentioned earlier, the partial collision operators $\Jejo[i],~i\in\lourd$, are not orthogonal to the space spanned by the vector $\Ce$.
However, an electron momentum relation is obtained  by projecting eq.~\eqref{eqepsilon-1} onto this space. First, the electron pressure, diffusion velocity, mean velocity,  conduction current density in the mean heavy-particle velocity frame, and conduction current density in the inertial frame are defined as
\begin{gather}
 \pre=\ne\tempe,\\
\Ve=\frac{1}{\ne}\int\Ce\feo\phie\;\d\Ce,\qquad \ve = \Mh\vitesse+{\Ve},\label{defVe}\\
\JJe=\ne\qe\Ve,\qquad \je=\ne\qe\ve.
\end{gather}

Then, we have the following proposition.

\begin{proposition}
\label{th:corelmom}
Considering $\feo$ given by eq.~\eqref{eqbolte} and $\fHio$,  $i\in\lourd$, by eq.~\eqref{eqbolti},  the zero-order momentum exchanged between electrons and heavy particles reads
  \begin{equation}\label{eqelmom}
\sumi[j]\ppscale{\Jejo(\feo\phie,\fHio[j]),\Ce}
=\frac{1}{\Mh}\dx\pre -\frac{\ne\qe}{\Mh}\E-\delta_{b1}\je\pvect\B.
\end{equation}
\end{proposition}

\begin{proof}
Equation~\eqref{eqepsilon-1} is projected onto the space spanned by the vector $\Ce$
$$-\ppscale{\feo\boltFe(\phie),\Ce}=\ppscale{\Demuhat,\Ce}+\delta_{b1}\qe\ppscale{\dCe(\feo\phie)\dscal(\Ce\pvect\B),\Ce}. $$ Then,  eq.~\eqref{eqelmom} is readily established by simplifying the left-hand-side by means of eq.~\eqref{eqcolinvee}, $\ppscale{\Ce,\Jee}=0$, at order $\varepsilon$   and by integrating by parts the right-hand-side.$\quad$
\end{proof}

The zero-order momentum exchanged between electrons and heavy particles is thus expressed in terms of the electron pressure and electric force. In addition, the following lemma allows for the momentum exchanged between heavy particles and electrons  to be calculated at order zero.

\begin{lemma}\label{th:lemelmom}
Considering $\feo$ given by eq.~\eqref{eqbolte} and $\fHio$,  $i\in\lourd$, by eq.~\eqref{eqbolti},  the net zero-order momentum exchanged between electrons and heavy particles vanishes, $i.e.$
  \begin{eqnarray}\label{eqelmomrel1}
\ppscalh{\Jieu[\heavy]\left(\fHo,\feo\phie\right),\invH[\nH+\nu]}+
\smash{\sumi[j]}\ppscale{\Jejo\left(\feo\phie,\fHio[j]\right),C_{\elec\nu}}=0,
\end{eqnarray}
for $\nu\in\{1,2,3\}$.
\end{lemma}

\begin{proof}
Equation~\eqref{eqelmomrel1}  is derived from eq.~\eqref{eqcrossedcol2}  at order $\epsilon^2$  based on corollaries \ref{th:Jieuiso}, \ref{th:Jiediso}, \ref{th:prodenteio}, and \ref{th:Jeju}.
\end{proof}

Moreover, the zero-order momentum exchanged between heavy particles and electrons can be directly calculated after introducing the average force of an electron acting on a heavy particle $i$ given by
\begin{equation}\label{eqdefFie}
\Aie= \int  \Qijl i\elec 1 (|\adimg|^2) \;|\adimg| \adimg \; \feo(\adimg)\phie(\adimg) \;\d\adimg, \quad i\in\lourd.
\end{equation}

\begin{lemma}\label{th:lemelmomFie}
Considering $\feo$ given by eq.~\eqref{eqbolte} and $\fHio$,  $i\in\lourd$, by eq.~\eqref{eqbolti},  the zero-order momentum exchanged between heavy particles and electrons reads
\begin{eqnarray}\label{eqelmomrel1Fie}
\ppscalh{\Jieu[\heavy]\left(\fHo,\feo\phie\right),\invH[\nH+\nu]}
= \sum_{i\in\lourd} \ni\Aienu,
\end{eqnarray}
for $\nu\in\{1,2,3\}$.
\end{lemma}
We will see that the average forces $\Aie$, $i\in\lourd$, contribute to the heavy-particle diffusion driving forces and, in particular, yielding to anisotropic diffusion velocities for the heavy particles in the case $b=1$.

\subsection{Order $\epsilon^{0}$: heavy-particle Euler equations}

We derive  Euler equations based on the heavy-particle Boltzmann equation~\eqref{eqbolbolh} at order $\epsilon^{0}$ corresponding to the macroscopic time scale $t^0$. First, a linearized collision operator is introduced for heavy-particles.

\begin{definition}
The linearized collision operator $\boltFi[\heavy]= (\boltFi)_{i\in\lourd}$ reads
\begin{equation*}
\boltFi(\phi_\heavy)=-\frac{1}{\fHio}\sumi[j]
\Bigl[ \Jij(\fHio\phii,\fHio[j])+\Jij(\fHio,\fHio[j]\phii[j])\Bigr],
\quad i\in\lourd,
\end{equation*}
where $\fHio$,  $i\in\lourd$, is given by eq.~\eqref{eqbolti}, for a family $\phi_\heavy=(\phi_i)_{i\in\lourd}$.
\end{definition}

 The  first non-vanishing term  of the partial collision operator $\Jie[\heavy]$ is not included in the linearized collision operator since it does not exhibit any property of orthogonality to $\invspace_\heavy$ for the scalar product $\ppscalh{\cdot,\cdot}$. The kernel of $\boltFi[\heavy]$   is given in the following property, the proof of which is omitted since it is a well-established result~\cite{ferziger}.
 
\begin{property}
\label{th:propboltFi}
The kernel of the linearized collision operator $\boltFi[\heavy]$ is the space of scalar collisional invariants $\invspace_\heavy$.
\end{property}

Furthermore, we define  the heavy-particle  pressure, $\pri[\heavy]=\nH\temph$, the mixture pressure, $\pression=\pre+\pri[\heavy]$, the heavy-particle charge, $\nH q_\heavy=\sum_{j\in\lourd}\ni[j]\qi[j]$, the mixture charge, $nq=\ne\qe+\nH q_\heavy$, and the total current density $\courantel_0=\nH q_\heavy\,\vitesse+ \ne\qe\ve/\Mh$.
The energy exchanged between heavy particles and electrons reads at order zero 
\begin{equation}
\deltaEho=\ppscalh{\Jied[\heavy] (\fHo,\feo),\invH[\nH+4]}. 
\end{equation}
This quantity is of the order of the thermal energy divided by the macroscopic time scale, $n^0\boltz T^0/t^0$.  A more accurate expression is calculated by means of corollary~\ref{th:Jiediso} 
\begin{equation}
\deltaEho= (\tempe-\temph)\sumi[j]{\ni[j]}\frac{\colfreqiez[j]}{\mi[j]}.\label{eqdeltaEho}
\end{equation}
Then, the heavy-particle Euler equations are derived in the following proposition. 

\begin{proposition}
If $\phi_{\heavy}$ is a solution to  eq.~\eqref{eqbolbolh} at order $\epsilon^{0}$, $i.e.$
\begin{equation}
\fHio\boltFi(\phi_\heavy)=-\Diz(\fHio)+\Jizhat,
\quad i\in\lourd\label{eqiepsilon-1},
\end{equation}
where $\feo$ is given by eq.~\eqref{eqbolte},  $\fHio$, $i\in\lourd$, by eq.~\eqref{eqbolti}, and $\phie$ by eqs.~\eqref{eqepsilon-1}-\eqref{eqconsepsilon-1}, and if $\fHo\phi_\heavy=(\fHio\phii)_{i\in\lourd}$ satisfies the constraints
\begin{equation}
\ppscalh{\fHo\phi_\heavy,\invH[l]}=0,\quad l\in\{1,\ldots,\nh+4\}\label{eqiconsepsilon-1},
\end{equation}
then, the zero-order  conservation equations of heavy-particle  mass, momentum, and energy   read
\begin{gather}
\dt\rhoi + \dx\dscal(\rhoi\vitesse) = 0,\quad i\in\lourd,\label{eqeuler1}\\
\dt(\rhoH\vitesse)+\dx\dscal(\rhoH\vitesse\ptens\vitesse+\tfrac{1}{\Mh^2} p\identite)
=\tfrac{1}{\Mh^2}nq\E+\delta_{b1}\courantel_0\pvect\B,\label{eqeuler2}\\
\dt(\rhoH\energiei[\heavy])+\dx\dscal\left(\rhoH\energiei[\heavy]\vitesse\right)=-\pri[\heavy]\dx\dscal\vitesse +\deltaEho.\label{eqeuler3}
\end{gather}
\end{proposition}

\begin{proof}
Fredholm's alternative~\cite{grad} represents the solvability condition of eq.~\eqref{eqiepsilon-1}
\begin{equation*}
\ppscalh{\Diz[\heavy],\invH[l]}=\ppscalh{\Jizhat[\heavy],\invH[l]},
\end{equation*}
$l\in\{1,\ldots,\nh+4\}$. Integrating by parts the left-hand-side and simplifying the right-hand-side based on theorem~\ref{th:thmJie} and corollary~\ref{th:Jiediso}, one obtains eqs.~\eqref{eqeuler1}, \eqref{eqeuler3}, and the following momentum conservation equation
\begin{equation}
-\Mh\rhoH\frac{\D\vitesse}{\D t}-\frac{1}{\Mh}\dx\pri[\heavy]+\frac{1}{\Mh}\nH q_{\heavy}\Ep+\ppscalh{\Jieu[\heavy]\left(\fHo,\feo\phie\right),(\invH[\nH+\nu])_{\nu\in\{1,2,3\}}}=0.\label{eqeuler2bis}
\end{equation}
Simplifying the latter equation by means of the heavy-particle mass conservation equation $\dt\rhoH + \dx\dscal(\rhoH\vitesse) = 0$
and lemma~\ref{th:lemelmom}, yields eq.~\eqref{eqeuler2}.
\end{proof}

\subsection{Order $\epsilon^{0}$: zero-order electron drift-diffusion equations}\label{sec:secordzeroelec}

We derive zero-order electron drift-diffusion  equations and a momentum relation based on the electron Boltzmann equation~\eqref{eqbolbole} at order $\epsilon^{0}$ corresponding to the macroscopic time scale $t^0$. 
We also prove, at this order of the resolution, that any nonzero shift introduced at the previous 
order leads to a series of difficulties at the present order.
It thus demonstrates that the initial choice of referential leads to a quite natural
resolution at sucessive orders.

With the previously obtained Maxwell-Boltzmann electron distribution function in eq. \eqref{eqbolte}
we introduce the electron heat flux
\begin{equation}
\heate=\int\frac{1}{2}\Ce\dscal\Ce\Ce\feo\phie \d\Ce.\label{defheate}
\end{equation}
The energy exchanged between electrons and heavy particles reads at order zero
\begin{equation}
\deltaEeo=\sumi[j]\ppscale{\Jejd[j] (\feo,\fHio[j]),\inve[2]}.
\end{equation}
The latter expression is calculated by means of eq.~\eqref{eqcrossedcol3} at order $\varepsilon^2$
\begin{equation}
\deltaEeo+\deltaEho=0,
\end{equation}
where $\deltaEho$ is given by eq.~\eqref{eqdeltaEho}.
Then, the zero-order electron drift-diffusion equations are derived in the following proposition.

\begin{proposition}
If  $\phide$ is a solution to eq.~\eqref{eqbolbole} at order $\epsilon^{0}$, $i.e.$
\begin{equation}
\feo\boltFe(\phide)+\delta_{b1}\qe\dCe(\feo\phidi[\elec])\dscal\Ce\pvect\B=-\Dezhat(\feo,\phie)
+\Jee(\feo\phie,\feo\phie)+\Jezhat,\label{eqepsilono}
\end{equation}
where $\feo$ is given by eq.~\eqref{eqbolte}, $\fHio$, $i\in\lourd$, by eq.~\eqref{eqbolti}, $\phie$ by eqs.~\eqref{eqepsilon-1}-\eqref{eqconsepsilon-1},  and $\phii$, $i\in\lourd$ by eqs.~\eqref{eqiepsilon-1}-\eqref{eqiconsepsilon-1}, and if $\feo\phide$ satisfies the constraints
\begin{equation}
\ppscale{\feo\phide,\inve[l]}=0,\quad l\in\{1,2\},\label{eqconsepsilono}
\end{equation}
then, the  zero-order conservation equations of electron mass and energy   read 
\begin{gather}
\dt\rhoe + \dx\dscal\bigl(\rhoe\vitesse+\tfrac{1}{\Mh}\rhoe\Ve\bigr) = 0,\label{eqdrift1}\\
\dt(\rhoe\energiee)+\dx\dscal\left(\rhoe\energiee\vitesse\right)=-\pre\dx\dscal\vitesse-\tfrac{1}{\Mh}\dx\dscal\heate+\tfrac{1}{\Mh}\JJe\dscal\Ep+\deltaEeo.\label{eqdrift2}
\end{gather}
\end{proposition}

\begin{proof}
Fredholm's alternative~\cite{grad} represents the solvability condition of eq.~\eqref{eqepsilono}
$$
\ppscale{\Dezhat,\inve[l]}=\ppscale{\Jezhat,\inve[l]},\quad l\in\{1,2\}.$$
Integrating by parts the left-hand-side and simplifying the right-hand-side based on theorem~\ref{th:theoJei} and corollary~\ref{th:Jeju} yields eqs.~\eqref{eqdrift1} and \eqref{eqdrift2}.
\end{proof}

\begin{lemma}
\label{th:uniquefe2}
In the chosen frame of reference, any velocity $\decall$ leads to a new definition of $\phie^{\decall 2}$ in eq.~\eqref{phi_decal},
for which 
property~\ref{th:kerIe} 
is preserved, 
and thus leads to an equivalent resolubility condition for $\phie^{\decall 2}$
as for $\phide$.
However, the resolution of $\phie^{\decall 2}$ is not equivalent to the 
resolution of $\phide$: in particular, the expansion corresponding to $\decall\neq0$ yields a non standard Chapman-Enskog expansion where the second-order distribution perturbation does not satisfy the scalar constraints~\eqref{eqconsepsilono}.
\end{lemma}

\begin{proof}
The difference between $\phie^{\decall 2}$ and  $\phide$ reads~\eqref{phi_decal}
\begin{equation*}
\delta\phie^{\decall 2} =
\phie^{\decall 2}-\phide = -\frac{\Mh}{\tempe}(\Ce\dscal\decall)\phie + \frac{\Mh^2}{2\tempe}\left[\decall\dscal\decall + \frac{(\Ce\dscal\decall)^2}{\tempe}\right].
\end{equation*}
The projection of $\delta\phie^{\decall 2}$ onto the collisional invariants is given by
\begin{gather*}
\ppscale{\feo\delta\phie^{\decall 2},\inve[1]} =
\tfrac{\Mh}{\tempe}\ne \decall\dscal(\Mh\decall-\Ve),\\
\ppscale{\feo\delta\phie^{\decall 2},\inve[2]} =
\Mh\decall\dscal(2\Mh\ne\decall-\tfrac{1}{\tempe}\heate).
\end{gather*}
The difference $\delta\phie^{\decall 2}$ is then orthogonal to the collisional invariants if  and only if $\decall=0$. 
To conclude, the resolution of $\phie^{\decall 2}$ yields a linearized Boltzmann equation where the right-member is orthogonal to the collisional invariants---a direct calculation shows that 
$\boltFe(\delta\phie^{\decall 2})+\delta_{b1}\qe\dCe(\delta\phie^{\decall 2})\dscal\Ce\pvect\B$ is orthogonal to the collisional invariants---whereas the scalar constraints on the unknown function $\phie^{\decall 2}$ are not zero.
\end{proof}

Consequently, for the reasons invoked so far, we will not try to shift the center of the Maxwell-Boltzmann distribution for electrons and stick with $\decall=0$ at any order.

Then, we define the electron viscous tensor, second-order electron diffusion velocity,   and second-order current density
\begin{align}
\visqueux[\elec]&=\int\Ce\ptens\Ce\feo\phie\;\d\Ce,\label{eqPie}\\
\Vde&=\frac{1}{\ne}\int\Ce\feo\phide\;\d\Ce, \label{eqdefVede}\\
\Jde&=\ne\qe\Vde.
\end{align}
A first-order electron momentum relation is given in the following proposition.

\begin{proposition}
\label{th:propelmom1}
 Considering $\feo$ given by eq.~\eqref{eqbolte}, $\fHio$,  $i\in\lourd$, by eq.~\eqref{eqbolti},  $\phie$ by eqs.~\eqref{eqepsilon-1}-\eqref{eqconsepsilon-1},  $\phii$,  $i\in\lourd$, by eqs.~\eqref{eqiepsilon-1}-\eqref{eqiconsepsilon-1}, and $\phide$ by eqs.~\eqref{eqepsilono}-\eqref{eqconsepsilono}, the first-order momentum exchanged between electrons and heavy particles reads
\begin{equation}
\sumi[j]\ppscale{\Jejo(\feo\phide,\fHio[j]),\Ce}+\ppscale{\Jezhat,\Ce}
=\tfrac{1}{\Mh}\dx\dscal\visqueux[\elec]-\left(\delta_{b0}\je+\delta_{b1}\Jde\right) \pvect\B.\label{eqepsilon11} 
\end{equation}
\end{proposition}

\begin{proof}
Equation~\eqref{eqepsilono} is projected onto the space spanned by the vector $\Ce$
\begin{equation*}
-\ppscale{\feo\boltFe(\phide),\Ce}+\ppscale{\Jezhat,\Ce}
=\Dezhat(\feo,\phie)+\delta_{b1}\qe\ppscale{\dCe(\feo\phidi[\elec])\dscal\Ce\pvect\B,\Ce}.
\end{equation*}
Then,  eq.~\eqref{eqepsilon11} is readily established by simplifying the left-hand-side by means of eq.~\eqref{eqcolinvee}, $\ppscale{\Ce,\Jee}=0$, at order $\varepsilon^2$ and by integrating by parts the right-hand-side.
\end{proof}

The first-order momentum exchanged between electrons and heavy particles is thus expressed in terms of the electron viscous tensor and electric force. Besides, the following lemma allows for the momentum exchanged between heavy particles and electrons to be calculated  at order $\varepsilon$.

\begin{lemma}\label{th:lemelmom1}
  Considering $\feo$ given by eq.~\eqref{eqbolte}, $\fHio$,  $i\in\lourd$, by eq.~\eqref{eqbolti},  $\phie$ by eqs.~\eqref{eqepsilon-1}-\eqref{eqconsepsilon-1},   $\phii$,  $i\in\lourd$, by eqs.~\eqref{eqiepsilon-1}-\eqref{eqiconsepsilon-1}, and $\phide$ by eqs.~\eqref{eqepsilono}-\eqref{eqconsepsilono}, the net first-order momentum exchanged between electrons and heavy particles vanishes, $i.e.$,
  \begin{equation}\label{eqelmomrel2}
\ppscalh{\Jipuhat[\heavy] ,\invH[\nH+\nu]}+
{\sumi[j]}\ppscale{\Jejo(\feo\phide,\fHio[j]),C_{\elec\nu}}+\ppscale{\Jezhat,C_{\elec\nu}}
=0,\quad\nu\in\{1,2,3\}.
\end{equation}
\end{lemma}

\begin{proof}
Equation~\eqref{eqelmomrel2}  is derived from eq.~\eqref{eqcrossedcol2} at order $\epsilon^2$ based on corollaries \ref{th:Jieuiso}-\ref{th:Jeidiso}.
\end{proof}

\subsection{Order $\epsilon$: heavy-particle Navier-Stokes equations}

We derive  Navier-Stokes equations based on the heavy-particle Boltzmann equation~\eqref{eqbolbolh} at order $\epsilon$. First, we introduce  the diffusion velocity and mean velocity of species $i\in\lourd$, 
\begin{equation}
\Vi=\frac{1}{\ni}\int\Ci\fHio\phii \d\Ci, \qquad 
\vi=\vitesse+\tfrac{\varepsilon}{\Mh}\Vi,
\qquad i\in\lourd,\label{eqVi}
\end{equation}
the heavy-particle viscous tensor, 
\begin{equation}
\visqueux[\heavy]=\sumi[j]\int\mi[j]\Ci[j]\ptens\Ci[j]\fHio[j]\phii[j] \d\Ci[j],
\label{eqPii}
\end{equation}
the second-order electron mean velocity,
\begin{equation}
\vde = \Mh\vitesse + \Ve+\varepsilon\Vde,
\end{equation}
the heavy-particle heat flux, 
\begin{equation}
\heati[\heavy]=\sumi[j]\int\frac{1}{2}\mi[j]\Ci[j]\dscal\Ci[j]\Ci[j]\fHio[j]\phii[j] \d\Ci[j],
\label{eqqi}
\end{equation}
the heavy-particle conduction current density in the mean heavy-particle velocity frame, the heavy-particle conduction current density in the inertial frame, the second-order electron conduction current density in the inertial frame, and the total current density,
\begin{equation}
\JJi[\heavy]=\sumi[j]\ni[j]\qi[j]\Vi[j], \quad \ji[\heavy]=\sumi[j]\ni[j]\qi[j]\vi[j],\quad
\jde = \ne\qe\vde, \quad \courantel= \ji[\heavy] + \frac{1}{\Mh}\jde.
\end{equation}
Furthermore, we define the energy exchanged between heavy particles and electrons reads at order $\varepsilon$
\begin{multline}
\deltaEhu=\ppscalh{\Jieu[\heavy](\fHo\phii[\heavy],\feo\phie),\invH[\nH+4]}+\ppscalh{\Jied[\heavy](\fHo,\feo\phie),\invH[\nH+4]}\\
+ \ppscalh{\Jied[\heavy](\fHo\phii[\heavy],\feo),\invH[\nH+4]}.\label{eqeqeq}
\end{multline}
The first term can be calculated by means of theorem~\ref{th:thmJie} 
\begin{equation}
\ppscalh{\Jieu[\heavy](\fHo\phii[\heavy],\feo\phie),\invH[\nH+4]}=
\sumi[j]\ni[j]\Vi[j]\dscal\Aie[j] ,
\end{equation}
and the two other terms will vanish. Then, we establish the following lemma used in the  derivation of the heavy-particle Navier-Stokes equations.

\begin{lemma}\label{th:lem4.3}
 Considering $\feo$ given by eq.~\eqref{eqbolte}, $\fHio$,  $i\in\lourd$, by eq.~\eqref{eqbolti},  $\phie$ by eqs.~\eqref{eqepsilon-1}-\eqref{eqconsepsilon-1},  $\phii$, $i\in\lourd$, by eqs.~\eqref{eqiepsilon-1}-\eqref{eqiconsepsilon-1},  and $\phide$ by eqs.~\eqref{eqepsilono}-\eqref{eqconsepsilono}, the mass exchanged at order $\varepsilon$ between heavy particles and electrons vanishes, $i.e.$,
  \begin{equation}\label{eqmass}
\ppscalh{\Jipuhat[\heavy],\invH[l]}=0,\quad
l\in\{1,\ldots,\nH\}.
\end{equation}
\end{lemma}

\begin{proof}
Equation~\eqref{eqmass}  is readily derived from eq.~\eqref{eqcrossedcol1} at order $\varepsilon^3$.\end{proof}

\begin{proposition}
If $\phidi[\heavy]$ is a solution to  eq.~\eqref{eqbol2} at order $\epsilon^1$, $i.e.$
\begin{equation}
\fHio\boltFi(\phidi[\heavy])
=-\Diu(\fHio\!,\phii)+\sumi[j]\Jij(\fHio\phii,\fHio[j]\phii[j])+\Jipuhat,\quad i\in\lourd,\label{eqiepsilon1}
\end{equation}
where $\feo$ is given by eq.~\eqref{eqbolte}, $\fHio$, $i\in\lourd$, by eq.~\eqref{eqbolti}, $\phie$ by eqs.~\eqref{eqepsilon-1}-\eqref{eqconsepsilon-1},  $\phii$, $i\in\lourd$, by eqs.~\eqref{eqiepsilon-1}-\eqref{eqiconsepsilon-1},  and $\phide$ by eqs.~\eqref{eqepsilono}-\eqref{eqconsepsilono}, and if $\fHo\phidi[\heavy]=(\fHio\phidi)_{i\in\lourd}$  satisfies the constraints 
\begin{equation}
\ppscalh{\fo_\heavy\phidi[\heavy],\invH[l]}=0,~l\in\{1,\ldots,\nh+4\}, \label{eqiconsepsilon1}
\end{equation}
then, the first-order conservation equations of heavy-particle mass, momentum, and energy  read
\begin{equation}
\dt\rhoi + \dx\dscal(\rhoi\vitesse+\tfrac{\varepsilon}{\Mh}\rhoi\Vi) = 0,
\quad i\in\lourd,\label{eqNS1}
\end{equation}
\begin{multline}
\dt(\rhoH\vitesse)+\dx\dscal(\rhoH\vitesse\ptens\vitesse+\tfrac{1}{\Mh^2} p\identite)=
-\tfrac{\varepsilon}{\Mh^2}\dx\dscal(\visqueux[\heavy]+\visqueux[\elec])+\tfrac{1}{\Mh^2}nq\E \\+[\delta_{b0}\courantel_0+\delta_{b1}\courantel]\pvect\B, \label{eqNS2}
\end{multline}
\begin{multline}
\dt(\rhoH\energiei[\heavy])+\dx\dscal(\rhoH\energiei[\heavy]\vitesse)=
-(\pri[\heavy]\identite+\varepsilon\visqueux[\heavy])\pmat\dx\vitesse-\tfrac{\varepsilon}{\Mh}\dx\dscal\heati[\heavy] +\tfrac{\varepsilon}{\Mh}\JJi[\heavy]\dscal\Ep\\
+\deltaEho+\varepsilon\deltaEhu.\label{eqNS3}
\end{multline}
\end{proposition}

\begin{proof}
The Chapman-Enskog method allows for the following conservation equations to be derived
\begin{equation*}
\ppscalh{\Diz[\heavy],\invH[l]}+\varepsilon\ppscalh{\Diu[\heavy],\invH[l]}=\ppscalh{\Jizhat[\heavy],\invH[l]}+\varepsilon\ppscalh{\Jipuhat[\heavy],\invH[l]},
\end{equation*}
$l\in\{1,\ldots,\nh+4\}$. Integrating by parts the left-hand-side and simplifying the right-hand-side based on the proof of heavy-particle Euler eqs.~\eqref{eqeuler1}-\eqref{eqeuler3}, proposition~\ref{th:propelmom1}, lemma~\ref{th:lemelmom1}, and lemma~\ref{th:lem4.3}, one obtains eqs.~\eqref{eqNS1}-\eqref{eqNS3}.\\
\end{proof}

\begin{remark}
 When one single type of heavy particles is considered, the first-order energy exchange term, heavy-particle diffusion velocities, and conduction currrent degenerate, 
 $\deltaEhu=0$, $\Vi=0$, $i\in\lourd$, $\JJi[\heavy]=0$, the total current  is simplified as well, $\courantel=n q \vitesse+\ne\qe\Ve/\Mh$.
 Therefore, we retrieve the formalism of Degond and Lucquin.
In such a case, the Navier-Stokes system can be coupled to the system of drift-diffusion equations for the electrons obtained at order $\varepsilon^0$ in the previous section. Since 
no energy exchange occurs at order $\varepsilon^1$, there is no need to resolve 
the electrons at order $\varepsilon^1$ to obtain a conservative model 
which insures the positivity of the entropy production.
However, this oversimplified case hides the details of the complex interaction between the electron and the heavy-particle mixture which is exhibited by the previous system of conservation eqs.~\eqref{eqNS1}-\eqref{eqNS3}. For a multicomponent mixture of heavy particles, thus, we have to 
extend the model obtained so far for the electron one order further, as done in the following section.
\end{remark}

\subsection{Order $\epsilon$: first-order electron drift-diffusion equations
\label{sec:secsub2}}

We derive first-order electron drift-diffusion  equations based on the electron Boltzmann equation~\eqref{eqbolbole} at order $\epsilon^{1}$.

Then, we  define the second-order electron heat flux
\begin{equation}\label{eqdefqde}
\heatde=\int\frac{1}{2}\Ce\dscal\Ce\Ce\feo\phide \d\Ce.
\end{equation}
The energy exchanged between electrons and heavy particles at order $\varepsilon$  is caculated by means of eq.~\eqref{eqcrossedcol3} at order $\varepsilon^3$
\begin{equation}
\deltaEeu+\deltaEhu=0,
\end{equation}
where $\deltaEhu$ is given by eq.~\eqref{eqeqeq}. Moreover, we establish the following lemma used in the  derivation of the first-order electron drift-diffusion equations.

\begin{lemma}\label{th:lem4.4}
 Considering $\feo$ given by eq.~\eqref{eqbolte}, $\fHio$,  $i\in\lourd$, by eq.~\eqref{eqbolti},  $\phie$ by eqs.~\eqref{eqepsilon-1}-\eqref{eqconsepsilon-1},  $\phii$, $i\in\lourd$, by eqs.~\eqref{eqiepsilon-1}-\eqref{eqiconsepsilon-1},  and $\phide$ by eqs.~\eqref{eqepsilono}-\eqref{eqconsepsilono}, the mass exchanged at order $\varepsilon$ between electrons and heavy particles vanishes, $i.e.$,
\begin{equation}\label{eqmass2}
\ppscale{\Jepuhat,\inve[1]}=0.
\end{equation}
\end{lemma}

\begin{proof}
Equation~\eqref{eqmass2}  is readily derived from eq.~\eqref{eqcrossedcol1} at order $\varepsilon^3$.
\end{proof}

\begin{proposition}
If  $\phite$ is a solution to eq.~\eqref{eqbolbole} at order $\epsilon^{1}$, $i.e.$
\begin{multline}
\feo\boltFe\bigl(\phite\bigr)
+\delta_{b1}\qe\dCe\bigl(\feo\phite\bigr)
\dscal\Ce\pvect\B= 
 -\Depuhat(\feo,\phie,\phide)
+\Jee(\feo\phide,\feo\phie) \\
+ \Jee(\feo\phie,\feo\phide)+\Jepuhat,
\end{multline}
where $\feo$ is given by eq.~\eqref{eqbolte}, $\fHio$, $i\in\lourd$, by eq.~\eqref{eqbolti}, $\phie$ by eqs.~\eqref{eqepsilon-1}-\eqref{eqconsepsilon-1},  $\phi_i$, $i\in\lourd$, by eqs.~\eqref{eqiepsilon-1}-\eqref{eqiconsepsilon-1},  $\phide$ by eqs.~\eqref{eqepsilono}-\eqref{eqconsepsilono}, and $\phi_i^2$, $i\in\lourd$, by eqs.~\eqref{eqiepsilon1}-\eqref{eqiconsepsilon1}, and if $\feo\phite$ satisfies the constraints
\begin{equation}
\ppscale{\feo\phite,\inve[l]}=0,\quad l\in\{1,2\},
\end{equation}
then, the first-order conservation equations of electron mass and energy read
\begin{equation}
\dt\rhoe + \dx\dscal\Bigl[\rhoe\bigl(\vitesse+\tfrac{1}{\Mh}(\Ve+\varepsilon\Vde)\bigr)\Bigr] = 0,\label{eqdrift21}
\end{equation}
\begin{multline}\label{eqdrift22}
\dt(\rhoe\energiee)+\dx\dscal\left(\rhoe\energiee\vitesse\right)=-\pre\dx\dscal\vitesse-\tfrac{1}{\Mh}\dx\dscal\left(\heate+\varepsilon\heatde\right)\\
+\tfrac{1}{\Mh}\left(\JJe+\varepsilon\Jde\right)\dscal\Ep
+\delta_{b0}{\varepsilon}{\Mh}\JJe\dscal\vitesse\pvect\B+\deltaEeo+\varepsilon\deltaEeu.
\end{multline}
\end{proposition}

\begin{proof}
The Chapman-Enskog method allows for the following conservation equations to be derived
\begin{eqnarray*}
&&\ppscale{\Dezhat,\inve[l]}+\varepsilon\ppscale{\Depuhat,\inve[l]}=\ppscale{\Jezhat,\inve[l]}+\varepsilon\ppscale{\Jepuhat,\inve[l]},\quad l\in\{1,2\}.
\end{eqnarray*}
Integrating by parts the left-hand-side and simplifying the right-hand-side based on lemma~\ref{th:lem4.4}, one obtains eqs.~\eqref{eqdrift21}-\eqref{eqdrift22}. The added terms associated with the perturbation of the mean electron velocity at order $\varepsilon^{1}$ do not bring in any contribution to the conservation equations.
\end{proof}

Before switching to Section~\ref{sec:sectransport} in order to evaluate the expression of the transport fluxes, we briefly come back to question of the influence of the choice of the initial frame of reference.

\subsection{About the necessity of conducting the expansion in the $\vitesse$ frame}\label{sec:secsubjustif}

As mentioned earlier, the mean heavy-particle velocity frame is not commonly adopted in the literature to conduct the Chapman-Enskog expansion. 
We have already underlined that the natural choice of the 
hydrodynamic velocity frame is not appropriate insofar as the 
global hydrodynamic velocity $\speed$ depends on the parameter $\epsilon$.
 Besides, the choice of the inertial frame
provides a vanishing mean velocity of the electrons. 
In fact Degond and Lucquin~\cite{degond2} and  Lucquin~\cite{lucquin1,lucquin2}
reach such a conclusion. However, since the expansion of the collision operators
in terms of $\varepsilon$ depends on the initial choice of referential
(see remark \ref{th:rem1})
and since the choice of the inertial frame prevents some terms from vanishing
(such as $\Jeju({\feo},{\fHio[j]})$), 
we will first show that these authors compensate the presence of non-zero terms
in the integro-differential equations
by the help of the $\decall$ velocity on the electron
introduced in Section~\ref{sec:secsub3}. This is affordable for 
the resolution of $\phie$, as proved in the following.
We then investigate if such an approach can be extended 
to the resolution of $\phide$.

Let us review the Chapman-Enskog expansion in a general frame. Considering a frame moving with the velocity $\velocity$, the peculiar velocities are given by
\begin{equation}\label{eqpaculiarben}
\Ce^\velocity = \ce - \epsilon\Mh\velocity,\qquad
\Ci^\velocity = \ci - \Mh\velocity,\quad i\in\lourd.
\end{equation}
The space of scalar electron collisional invariants $\invspace_\elec^\velocity$ is spanned by the following elements
\begin{equation*}
\left\{
\vcenter{\halign{
$#\hfil$&$\;\,=\; #\hfil$
\cr
\inve[\velocity,1]&1,
\cr\noalign{\vskip4pt}
\inve[\velocity,2]&\tfrac{1}{2}\Ce^\velocity\dscal\Ce^\velocity,
\cr}} \right.
\end{equation*}
the space of scalar heavy-particle collisional invariants $\invspace_\heavy^\velocity$ by
\begin{equation*}
\left\{
\vcenter{\halign{
$#\hfil$&$\;\,=\; \bigl(#\bigr)_{i\in\lourd},\hfil$&$\qquad#,\hfil$
\cr
\invH[\velocity,j ]&\mi\delta_{ij}&j\in\lourd
\cr\noalign{\vskip4pt}
\invH[\velocity,\nh+\nu]&\mi\Cinorme[i\nu]^\velocity&\nu\in\{1,2,3\}
\cr\noalign{\vskip4pt}
\invH[\velocity,\nh+4]&\tfrac{1}{2}\mi\Ci^\velocity\dscal\Ci^\velocity
\cr}} \right.
\end{equation*}
and the macroscopic properties are expressed  as partial scalar products of the distribution functions and the collisional invariants
\begin{equation*}
\left\{
\vcenter{\halign{
$\ppscale{\fe^\velocity,#}{}\hfil$&$\;\,=\; #,\hfil$
\cr
\inve[\velocity,1]&\rhoe
\cr\noalign{\vskip4pt}
\inve[\velocity,2]&\frac{3}{2} \ne\tempe
\cr}} \right.
\end{equation*}
and 
\begin{equation*}
\left\{
\vcenter{\halign{
$\ppscalh{\fH^\velocity,#}{}\hfil$&$\;\,=\; #\hfil$&$\qquad#,\hfil$
\cr
\invH[\velocity,i]&\rhoi,&i\in\lourd
\cr\noalign{\vskip4pt}
\invH[\velocity,\nh+\nu]&\rhoi[\heavy]\Mh(\vitessenu-\velocitynu),&\nu\in\{1,2,3\}
\cr\noalign{\vskip4pt}
\invH[\velocity,\nh+4]&\frac{3}{2}\nH\temph + \frac{1}{2}\Mh^2\rhoi[\heavy]
(\velocity-\vitesse)\dscal(\velocity-\vitesse).
\cr}} \right.
\end{equation*}
Similarly to the low Mach number approximation for neutral gases, we decouple for the electrons the thermal energy from the mixture kinetic energy.

Then, we reword the Chapman-Enskog expansion of Section~\ref{sec:sec4} in the 
$\velocity$ frame. First, let us now formulate two propositions from the begining of this section.

\begin{proposition}[Order $\epsilon^{-2}$: electron thermalization]
The zero-order electron distribution function $\fe^{\velocity 0}$, 
solution to eq.~\eqref{eqbolbole} at order $\epsilon^{-2}$, 
$i.e.$, $\Demd(\fe^{\velocity 0})=\Jemd$,
that satisfies the scalar constraints $\ppscale{\fe^{\velocity 0},\inve[l]} = \ppscale{\fe^\velocity,\inve[l]}$, $l\in\{1,2\}$,
is a Maxwell-Boltzmann distribution function at the electron temperature
\begin{equation}\label{eqbolteben}
\fe^{\velocity 0} = \ne \left(\frac{1}{2\pi\tempe}\right)^{3/2} 
\exp\left( - \frac{1}{2\tempe}\Ce^\velocity\dscal\Ce^\velocity\right).
\end{equation}
\end{proposition}
\begin{proof}
The proof is identical to the one of proposition~\ref{th:propoelec}.
\end{proof}
\begin{proposition}[Order $\epsilon^{-1}$: heavy particle thermalization]
Considering $\fe^{\velocity 0}$ given by eq.~\eqref{eqbolteben},
the zero-order family of heavy-particle distribution functions $\fHi[\heavy]^{\velocity 0}$
solution to eq.~\eqref{eqbolbolh} at order~$\epsilon^{-1}$, 
$i.e.$, $\Jimu=0$, $i\in\lourd$,
that satisfies the scalar constraints $\ppscalh{\fH^{\velocity 0},\invH[l]} = \ppscalh{\fH^\velocity,\invH[l]}$, $l\in\{1,\ldots,\nh{+}4\}$,
is a family of Maxwell-Boltzmann distribution functions at the heavy-particle temperature
\begin{equation}\label{eqboltiben}
\fHi^{\velocity 0}= \ni \left(\frac{\mi}{2\pi\temph}\right)^{3/2} 
\exp\Bigl( 
- \frac{\mi}{2\temph}
\bigl[\Ci^\velocity{-}\Mh(\vitesse{-}\velocity)\bigr]
\dscal
\bigl[\Ci^\velocity{-}\Mh(\vitesse{-}\velocity)\bigr]
\Bigr),\quad i\in\lourd.
\end{equation}
\end{proposition}
\begin{proof}
Multiplying the equation $\Jimu=0$ by $\ln\left[{(2\pi)^{3/2}n^0}\fHi^{\velocity 0}{(\mi^3\partitioni[\heavy])}\right]$, where the heavy-particle translational partition function reads $\partitioni[\heavy]=({2\pi m_h^0\boltz T^0}/{\planck^2})^{3/2}$, integrating over $\d\Ci$, and summing over the heavy particles  yields the entropy production 
\begin{equation*}
\sum_{i,j\in\lourd}\prodent_{ij}^0+\sumi\prodent_{i \elec}^0=0,
\end{equation*}
where the partial entropy production terms read 
\begin{eqnarray*}
\prodent_{ij}^0&=&\int \seceff |\Ci-\Ci[j]| ~\Omega(\fHi^{\velocity 0\prime}\fHi[j]^{\velocity 0\prime},\fHi^{\velocity 0}\fHi[j]^{\velocity 0})~\d\Ci\d\Ci[j]\d\ovec,\quad i,j\in\lourd,\\
\prodent_{i \elec}^0&=&  \tfrac{1}{\mi}\int\partial_{\Ci}\fHi^{\velocity 0} \ln\left[\tfrac{(2\pi)^{3/2}n^0\fHi^{\velocity 0}}{\mi^3\partitioni[\heavy]}\right]\d\Ci \dscal \int  \Qijl i\elec 1 (|\adimg|^2)
\;|\adimg| \adimg ~
\fe^{\velocity 0} ~\d\adimg,~ i\in\lourd.
\end{eqnarray*}
Integrating by parts, the terms $\prodent_{i \elec}^0$, $ i\in\lourd$, vanish and we conclude as usual.
\end{proof}

At this step, two properties appear: the electron thermalization takes place in any velocity frame, whereas the zero-order heavy particle distribution functions do not depend on the selected frame. Indeed, we clearly have $\fHi^{\velocity 0}=\fHio$, $i\in\lourd$, for all velocity $\velocity$.

Considering then the Boltzmann equation at order $\epsilon^{-1}$,
the electron first perturbation $\phie^\velocity$ satisfies the linearized Boltzmann equation
\begin{multline}\label{eqepsilon-1ben}
\boltFe(\phie^\velocity)
+\delta_{b1}\qe\dCe[\velocity](\phie^\velocity)\dscal\Ce\pvect\B=
-\frac{1}{{\feo}^\velocity}\Demuhat(\fe^{\velocity 0}) \\
+ \sumi \ni \frac{\Mh}{\tempe} \Qijl \elec i 1(|\Ce^\velocity|^2)
|\Ce^\velocity| (\vitesse-\velocity)\dscal\Ce^\velocity ,
\end{multline}
with the constraints 
\begin{equation}\label{eqconsepsilon-1ben}
\ppscale{\fe^{\velocity 0}\phie^\velocity,\inve[\velocity,l]} = 0, \qquad l\in\{1,2\}.
\end{equation}
The right-member of eq.~\eqref{eqepsilon-1ben} satisfies the constraint to be orthogonal to the collisional invariants, that is the solvability condition.
Moreover, in order to avoid treating the newly introduced term in the integro-differential equation, one can use the absence of momentum constraints
on the electron distribution function and introduce a velocity shift
$\decall = \vitesse-\velocity$ and 
notice that~\cite{lucquin1} 
\begin{equation} \label{eqeqepsilon-1Ceben}
\boltFe(\Ce^\velocity)= -
\sumi \ni  \Qijl \elec i 1 (|\Ce^\velocity|^2) |\Ce^\velocity|\Ce^\velocity ,
\end{equation}
we thus obtain that the conduction of the Chapman-Enskog expansion 
in the $\velocity$ frame is equivalent 
to that in the $\vitesse$ frame
with
\begin{equation*}
\phie^\velocity = \phie + \frac{\Mh}{\tempe} (\vitesse-\velocity)\dscal\Ce^\velocity.
\end{equation*}
As already mentioned in Section \ref{sec:secsub3}, 
the electron velocity $\ve$ can be splitted 
into two parts at the same order of the multiscale expansion
$\ve = \Mh\velocity + {\Ve}^\velocity +\ordre (\varepsilon)$, with 
${\Ve}^\velocity =\Ve + \Mh(\vitesse-\velocity)$. 
We have thus  provided a nice interpretation of the algebra proposed in Lucquin \cite{lucquin1} where the use
of $\decall=\vitesse$ allows to eliminate the presence of the term
$\sum_{j\in\lourd} \Jeju(\fe^{\velocity  0},\fHi[j]^{\velocity 0})$ in the integro-differential equation for $\phie^\velocity$ obtained when working in the inertial frame $\velocity=0$.

As a conclusion, it amounts to ``coming back'' into the mean heavy-particle velocity referential.
Let us underline at this point, that the set of equations obtained for the heavy-particle Euler equations coupled to the  zero-order electron drift-diffusion equations is identical to the set obtained in Lucquin~\cite{lucquin1}.
While still equivalent at this order of the expansion to our study
and yielding the same macroscopic equations, it leads to an artificial complexity. 
This is  a first step in the justification of the choices made in Section~\ref{sec:secdiman} 
in terms of the referential and associated simplified algebra. 
However, at order $\epsilon$, which yields 
heavy-particle Navier-Stokes equations coupled to first-order electron drift-diffusion equations, we realize that 
such a compensation used through the velocity shift $\decall$ has a nasty influence on the structure of the expansion 
at the next order (see Lemma \ref{th:uniquefe2})
and hence will make the resolution of $\phide$ difficult. Concerning the heavy-particle Boltzmann equation at order $\epsilon^0$, the first-order perturbation functions $\phii^\velocity$, $i\in\lourd$, satisfy too eq.~\eqref{eqiepsilon-1}, that implies that $\phii^\velocity=\phii$, $i\in\lourd$.

Finally, the second-order electron perturbation $\phie^{\velocity 2}$ satisfies
\begin{multline}
\label{eqepsilonoben}
\boltFe(\phie^{\velocity 2})+\delta_{b1}\qe\dCe(\phie^{\velocity 2})\dscal\Ce\pvect\B=
-\frac{1}{\fe^{\velocity 0}}\Dezhat(\fe^{\velocity 0},\phie^\velocity) \\
+\frac{1}{{\feo}^\velocity}
\Jee(\fe^{\velocity 0}\phie^\velocity,\fe^{\velocity 0}\phie^\velocity)
+\frac{1}{\fe^{\velocity 0}} \sumi[j] \Bigl[
\Jeju(\fe^{\velocity 0}\phie^\velocity,\fHi[j]^{\velocity 0})
+ \Jejd(\fe^{ \velocity 0},\fHi[j]^{\velocity 0}) \Bigr],
\end{multline}
with the constraints
\begin{equation}
\ppscale{\feo\phie^{\velocity 2},\inve[\velocity,l]}=0,\quad l\in\{1,2\}.\label{eqconsepsilonoben}
\end{equation}
In the general case, the right-member of eq.~\eqref{eqepsilonoben} is not orthogonal to the collisional invariants, and the solvability condition is not satisfied. More precisely, the two first terms $\Dezhat(\fe^{\velocity 0},\phie^\velocity)/\fe^{\velocity 0}$
and 
$\Jee(\fe^{\velocity 0}\phie^\velocity,\fe^{\velocity 0}\phie^\velocity)/\fe^{\velocity 0}$
satisfy the constraints for all velocity $\velocity$, whereas
a lengthly calculation yields that the projection of the last term over the collisional invariants $\inve[\velocity,l]$, $l\in\{1,2\}$, does not vanish, except for $\velocity=\vitesse$. 

As a conclusion on this matter, if the velocity $\velocity$ is not equal to $\vitesse$, that is to say if the selected frame is not the mean-heavy particle velocity frame, the structure of the expansion of the collision operators in terms of $\varepsilon$  prevents from solving for $\phie^{\velocity 2}$. Moreover, 
the ``trick'' used in Degond and Lucquin at previous order for $\phie$ 
also brings in additional difficulties at subsequent order through the residual 
terms associated with $\decall$. 
The choice of the mean heavy-particle velocity frame not only yields a simplified algebra---the parity properties of the distribution functions are in accordance to the expression of the multiscale collisional operators--- but is mandatory in order 
to obtain 
the well-posed second-order problem for electron distribution, second-order electron diffusion velocity, and heat flux being crucial to balance the global energy equation. This statement fully justifies the choices made in this contribution.

\section{Transport coefficients}
\label{sec:sectransport}

In this section, we investigate the electron and  heavy-particle perturbation functions in order to obtain the expressions of the transport fluxes. We only treat the case $b=1$ corresponding to a strong magnetization (inducing anisotropic transport coefficients), the case $b=0$ corresponding to a weak magnetization will be investigated in a forthcoming publication.

\subsection{Extra notations for anisotropy}

We introduce some extra notations in order to conveniently express the solution to the Boltzmann equation in the presence of a strong magnetic field. First, we define a unitary vector for the magnetic field $\Bnorme=\B/|\B|$ and also three direction matrices
\begin{equation*}
\Mpa = \Bnorme\ptens\Bnorme,
\quad
\Mpe = \identite - \Bnorme\ptens\Bnorme,
\quad
\Mt = \begin{pmatrix}
0 & -\normBnorme[3] & \normBnorme[2] \\
\normBnorme[3] & 0 & -\normBnorme[1] \\
-\normBnorme[2] & \normBnorme[1] & 0
\end{pmatrix},
\end{equation*}
so that we have for any vector $\X$ in three dimensions
\begin{equation*}
\Xpa=\Mpa\X = \X\dscal\Bnorme\;\Bnorme,
\qquad
\Xpe=\Mpe\X = \X - \X\dscal\Bnorme\;\Bnorme,
\qquad
\Xt=\Mt\X = \Bnorme\pvect\X.
\end{equation*}
In the $(\X,\Bnorme)$ plane, the vector $\Xpa$ is the component of $\X$ parallel to the magnetic field and $\Xpe$ its component perpendicular to the magnetic field. Thus, we have $\X=\Xpa+\Xpe$. The vector $\Xt$ lies in the direction transverse to the $(\X,\Bnorme)$ plane.
The three vectors $\Xpa$, $\Xpe$, and $\Xt$ are then mutually orthogonal. We will show that the transport properties are anisotropic.  In the weak magnetization limit, the transport properties are identical in the parallel and perpendicular directions and vanish  in the transverse direction.

\subsection{First-order electron perturbation function}
The first-order perturbation function $\phie$ is a solution to  eq.~\eqref{eqepsilon-1} 
\begin{equation}\label{eqphie}
\boltFe(\phie)+\qe\dCe(\phie)\dscal\Ce\pvect\B= \psie,
\end{equation}
and satisfies the constraints~\eqref{eqconsepsilon-1}, where $\psie$ is given by the expression $\psie =-\Demuhat(\feo)/\feo$ and $\feo$  by eq.~\eqref{eqbolte}. After some algebra based on the expression of $\feo$, the quantity $\psie$ is transformed into
\begin{equation}\label{psie}
\psie=-\pre\;\psieDe\dscal\de-\psielambda\dscal\dx\Bigl(\frac{1}{\tempe}\Bigr),
\end{equation}
where the electron diffusion driving force $\de$ is defined by the relation
\begin{equation}
\de=\frac{1}{\pre}\dx \pre -\frac{\ne\qe}{\pre}\Ep,
\end{equation}
and with
\begin{equation}\label{defpsiemu}
\psieDe=\frac{1}{\Mh\pre}\Ce,
\qquad
\psielambda=\frac{1}{\Mh}\Bigl(\frac{5}{2}\tempe-\frac{1}{2}\Ce\dscal\Ce\Bigr)\Ce.
\end{equation}
The right-hand-side of eq.~\eqref{eqphie} does not depend on the heavy-particle driving forces. Therefore, the first-order electron perturbation function is decoupled from the heavy-particles.

The existence and uniqueness of a solution to eq.~\eqref{psie} is given in the following proposition.

\begin{proposition}
\label{th:existencephie}
The scalar function $\phie$ given by
\begin{multline}\label{expphie}
\phie = -\pre \Re\Bigl[
\Mpa \VarphieDeun + (\Mpe+\i\Mt) \VarphieDede
\Bigr] \dscal \de\\
-  \Re\Bigl[
\Mpa \Varphielambdaun + (\Mpe+\i\Mt) \Varphielambdade
\Bigr] \dscal \dx\Bigl(\frac{1}{\tempe}\Bigr),
\end{multline}
is the solution to eq.~\eqref{eqphie} under the constraints~\eqref{eqconsepsilon-1}, where
the vectorial functions $\VarphieDeun$, $\VarphieDede$, $\Varphielambdaun$, and $\Varphielambdade$ are the solutions of the problem
\begin{gather}
\boltFe(\Varphiemuun)=\psiemu\label{eqvarphie1},\\
\bigl(\boltFe+\i|\B|\boltFev\bigr)\bigl(\Varphiemude\bigr)=\psiemu,\label{eqvarphie2}
\end{gather}
where $\boltFev(\boldsymbol{u})=\qe\boldsymbol{u}$, under the constraints
\begin{align}
\ppscale{\feo\Varphiemuun,\inve[l]}&=0,&& l\in\{1,2\},\label{eqconstvarphie1}\\
\ppscale{\feo\Varphiemude,\inve[l]}&=0,&& l\in\{1,2\},\label{eqconstvarphie2}
\end{align}
with $\mu\in\{\symbolDkl_{\elec},\lambdahat_\elec\}$.
\end{proposition}

\begin{proof}
By linearity and isotropy of the linearized Boltzmann operator $\boltFe$,
the development~\eqref{psie} of $\psie$ can be followed through for $\phie$ as well
\begin{equation*}
\phie=-\pre\PhieDe\dscal\de-\Phielambda\dscal\dx\Bigl(\frac{1}{\tempe}\Bigr).
\end{equation*}
The functions $\Phiemu$, $\mu\in\{\symbolDkl_{\elec},\lambdahat_\elec\}$, are now vectorial and satisfy the equations
\begin{equation}\label{eqPhiemu}
\boltFe(\Phiemu)+\qe\Ce\pvect\B\dscal\dCe\Phiemu=\psiemu,
\end{equation}
and the scalar constraints
\begin{equation}\label{constPhiemu}
\ppscale{\feo\Phiemu,\inve[l]}=0,\quad l\in\{1,2\}.
\end{equation}
We seek a solution $\Phiemu$ in the form
\begin{equation*}
\Phiemu=\phiemuun\Ce+\phiemude\Ce\pvect\B+\phiemutr\Ce\dscal\B\;\B,
\end{equation*}
where $\phiemuun$, $\phiemude$ and $\phiemutr$ are scalar functions of $\Ce\dscal\Ce$, $(\Ce\dscal\B)^2$ and $\B\dscal\B$, since $\Phiemu$ must be invariant under a change of coordinates.
Substituting this expansion in \eqref{eqPhiemu}, and using isotropy, eq.~\eqref{eqPhiemu} splits into three separate coupled equations
\begin{gather}
\boltFe(\phiemuun\Ce)-\qe\B\dscal\B\phiemude\Ce=\psiemu,\label{eqphie1}\\
\boltFe(\phiemude\Ce\pvect\B)+\qe\phiemuun\Ce\pvect\B=0,\label{eqphie2}\\
\boltFe(\phiemutr\Ce\dscal\B~\B)+\qe\Ce\dscal\B\phiemude\B=0.\label{eqphie3}
\end{gather}
Further simplification is now obtained if, instead of three real quantities
$\phiemuun$, $\phiemude$ and $\phiemutr$,
we introduce one real and one complex unknown defined by 
\begin{equation*}
\varphiemuun=\phiemuun+\B\dscal\B\phiemutr,
\qquad
\varphiemude=\phiemuun+\i|\B|\phiemude.
\end{equation*}
Upon introducing $\Varphiemuun=\varphiemuun\Ce$ and $\Varphiemude=\varphiemude\Ce$, eqs.~\eqref{eqphie1}, \eqref{eqphie2}, and \eqref{eqphie3} can be conveniently rewritten in terms of these new functions
\begin{gather}
\boltFe(\Varphiemuun)=\psiemu\tag{\ref{eqvarphie1}},\\
\bigl(\boltFe+\i|\B|\boltFev\bigr)\bigl(\Varphiemude\bigr)=\psiemu,\tag{\ref{eqvarphie2}}
\end{gather}
Furthermore,  the constraints~\eqref{constPhiemu} are easily rewritten in the form 
\begin{gather}
\ppscale{\feo\Varphiemuun,\inve[l]}=0,\quad l\in\{1,2\},\tag{\ref{eqconstvarphie1}}\\
\ppscale{\feo\Varphiemude,\inve[l]}=0,\quad l\in\{1,2\}.\tag{\ref{eqconstvarphie2}}
\end{gather}
Moreover, expression~\eqref{expphie} for $\phie$ is immediately obtained using the recombination formula
\begin{equation*}
\Phiemu=\Mpa\Varphiemuun+\Mpe \real\bigl(\Varphiemude\bigr)
-\Mt \imaginary\bigl(\Varphiemude\bigr).
\end{equation*}
\end{proof}

The structure of the integral equation~\eqref{eqvarphie1} under the constraints~\eqref{eqconstvarphie1} is classical and the structure of equation~\eqref{eqvarphie2} under the constraints~\eqref{eqconstvarphie2} is similar in a complex framework. More specifically, the operator 
$\boltFe+\i|\B|\boltFev$ and the associated bilinear form
$\textbf{a}(\textbf{u},\textbf{v}) = \ppscale{\textbf{u},(\boltFe+\i|\B|\boltFev)\textbf{v}}$, defined on the proper Hilbert space of complex isotropic squared integrable functions associated with the scalar product $[\cdot,\cdot]$, are such that 
$|\textbf{a}(\textbf{u},\textbf{u})|\geq [\textbf{u},\textbf{u}]$, which yields existence and uniqueness thanks to the constraints. Moreover, from the isotropy of the operator $\boltFe$, the expressions $\varphiemuun$ and $\varphiemude$ cannot be functions of $(\Ce\dscal\B)^2$ as shown in~\cite{ferziger}.

We further introduce the electron bracket operators $\crochete{\cdot,\cdot}$ and $\dblparenthesee{\cdot,\cdot}$ associated with the two operators $\boltFe$ and $\boltFev$. For any $\xi_\elec$ and $\zeta_\elec$, we define
\begin{equation*}
\crochete{\xi_\elec,\zeta_\elec}=\ppscale{\feo\xi_\elec,\boltFe(\zeta_\elec)},
\qquad
\dblparenthesee{\xi_\elec,\zeta_\elec} = |\B|\;\ppscale{\feo\xi_\elec,\boltFev(\zeta_\elec)}.
\end{equation*}
These bracket operators develop into
\begin{align*}
\crochete{\xi_\elec,\zeta_\elec}=\;& \tfrac{1}{2}\sumi[j]
\ni[j]\int\seceff[\elec j] (|\Ce|^2,\ovec\dscal\evec) |\Ce| \feo(|\Ce|\evec)\\
&\null\qquad\qquad
[\xi_\elec(|\Ce|\evec)-\xi_{\elec}(|\Ce|\ovec) ]
\pcont[\overline{{\zeta}_\elec(|\Ce|\evec)-{\zeta}_{\elec }(|\Ce|\ovec) }]
\; \d\ovec\,\d\evec\,\d|\Ce|\\
&\null+\tfrac{1}{4} \int \seceff[\elec\elec 1] |\Ce-\Ceu| \feo\feuo \\
&\null\qquad\qquad
(\xi_\elec+\xi_{\elec 1}-\xi_{\elec}^{\prime}-\xi_{\elec 1}^{\prime})
\pcont(\overline{{\zeta}_\elec+{\zeta}_{\elec 1}-{\zeta}_{\elec }^{\prime}-{\zeta}_{\elec 1}^{\prime}})
 \;\d\ovec\d\Ce\d\Ceu,
\intertext{and}
\dblparenthesee{\xi_\elec,\zeta_\elec} =\;& 
|\B|\qe\int\;\feo\xi_\elec\pcont\bar{\zeta}_\elec\;\d\Ce.
\end{align*}
The bracket operator $\crochete{\cdot,\cdot}$ is hermitian
$\crochete{\xi_\elec,\zeta_\elec} = \crochetebar{\zeta_\elec,\xi_\elec}$, positive semi-definite $\crochete{\xi_\elec,\xi_\elec}\geq0$, and its kernel is spanned by the collisional invariants, $i.e.$, $\crochete{\xi_\elec,\xi_\elec}=0$ implies that $\xi_\elec$ is a (tensorial) collisional invariant, or in other words,  all its tensorial components are in the space 
$\invspace_\elec$.
The bracket operator $\dblparenthesee{\cdot,\cdot}$ is hermitian
$\dblparenthesee{\xi_\elec,\zeta_\elec} = \dblparentheseebar{\zeta_\elec,\xi_\elec}$
and negative definite $\dblparenthesee{\xi_\elec,\xi_\elec}<0$ if $\xi_\elec\neq0$.
\begin{remark}In the limit case $\B$ tends to zero, expression~\eqref{expphie} for the first-order electron perturbation function reduces to an isotropic form. We prove indeed that, for $\mu\in\{\symbolDkl_{\elec},\lambdahat_\elec\}$, $\Varphiemuun$ does not depend on the magnetic field and that $\Varphiemude$ converges to $\Varphiemuun$ for a vanishing magnetic field.\end{remark}

The expression of the electron diffusion velocity is given in the following proposition.

\begin{proposition}
\label{th:expressionVe}
The electron diffusion velocity $\Ve$ reads
\begin{multline}\label{expVe}
\Ve=-\bigl(\Deepa\depa+\Deepe\depe+\Deet\det\bigr)\\
-\bigl(\thetaepa\glogTepa+\thetaepe\glogTepe+\thetaet\glogTet\bigr),
\end{multline}
where the diffusion coefficients and thermal diffusion coefficients are given by
\begin{equation}\label{expDeethetae}
\vcenter{\halign{
$#\hfil$&$\;\,= \hfil #$&$\hfil#,\hfil$&$\qquad #\hfil$&$\;\,= \hfil #$&$\hfil#\hfil$
\cr
\Deepa&\frac{1}{3}{\pre\tempe  \Mh} &\crochete{\VarphieDeun,\VarphieDeun}
&
\thetaepa&-\frac{1}{3}{ \Mh}&\crochete{\VarphieDeun,\Varphielambdaun} ,
\cr\noalign{\vskip4pt}
\Deepe&\frac{1}{3}{\pre\tempe \Mh}  &\crochete{\VarphieDede,\VarphieDede}
&
\thetaepe&-\frac{1}{3}{ \Mh}&\crochete{\VarphieDede,\Varphielambdade},
\cr\noalign{\vskip4pt}
\Deet&-\frac{1}{3}{\pre\tempe \Mh} &\dblparenthesee{\VarphieDede,\VarphieDede}
&
\thetaet&\frac{1}{3}{ \Mh}&\dblparenthesee{\VarphieDede,\Varphielambdade}.
\cr}}
\end{equation}
Note that the previous expressions are real, in particular for $\thetaepe$ and $\thetaet$, although functions $\VarphieDede$ and $\Varphielambdade$ are complex.
\end{proposition}

\begin{proof}
Using definition~\eqref{defVe} of the diffusion velocity $\Ve$ and expression~\eqref{defpsiemu} of $\psieDe$ yields
\begin{equation*}
\Ve = \tempe\Mh\ppscale{\psieDe,\feo\phie}.
\end{equation*}
Further substituting expansion~\eqref{expphie} into the latter equation, and using isotropy, we obtain
expression~\eqref{expVe} for the diffusion velocity $\Ve$, where the transport coefficients are defined by
$\Deepa = \frac{1}{3} \pre\tempe\Mh \ppscale{\feo\varphieDeun,\psieDe}$,
$\thetaepa = -\frac{1}{3} \Mh \ppscale{\feo\varphielambdaun,\psieDe}$,
$\Deepe+\i\Deet = \frac{1}{3} \pre\tempe\Mh \ppscale{\feo\varphieDede,\psieDe}$,
$\thetaepe+\i\thetaet = -\frac{1}{3} \Mh \ppscale{\feo\varphielambdade,\psieDe}$.
Eqs.~\eqref{eqvarphie1} and \eqref{eqvarphie2} for $\mu=\symbolDkl_{\elec}$ classically yields~\cite{ferziger,graille1}
\begin{gather*}
\Deepa = \tfrac{1}{3}{\pre\tempe  \Mh} \; \crochete{\VarphieDeun,\VarphieDeun},\\
\Deepe+\i\Deet = \tfrac{1}{3}{\pre\tempe  \Mh} 
\Bigl(\crochete{\VarphieDede,\VarphieDede}-\i\dblparenthesee{\VarphieDede,\VarphieDede}\Bigr),\\
\thetaepa = -\tfrac{1}{3}{\Mh} \; \crochete{\Varphielambdaun,\VarphieDeun},\\
\thetaepe+\i\thetaet = -\tfrac{1}{3}{\Mh} 
\Bigl(\crochete{\Varphielambdade,\VarphieDede}-\i\dblparenthesee{\Varphielambdade,\VarphieDede}\Bigr).
\end{gather*}
As the bracket operators $\crochete{\cdot,\cdot}$ and $\dblparenthesee{\cdot,\cdot}$ are hermitian, we immediately conclude for expressions of $\Deepa$, $\Deepe$, $\Deet$ and $\thetaepa$. Concerning $\thetaepe$ and $\thetaet$, we use the imaginary part of eq.~\eqref{eqvarphie2} for $\mu\in\{\symbolDkl_{\elec},\lambdahat_\elec\}$, so that
\begin{equation*}
\Im\left[
\bigl(\boltFe+\i|\B|\boltFev\bigr)\bigl(\Varphiemude
\bigr)\right] = 0,\qquad 
\mu\in\{\symbolDkl_{\elec},\lambdahat_\elec\}.
\end{equation*}
Taking the scalar product of the previous equation with $\Varphiemude$, $\mu\in\{\symbolDkl_{\elec},\lambdahat_\elec\}$ yields the four following relations
\begin{equation*}
\vcenter{\halign{
$\crochete{\hfil \,#\,\hfil}$&$\;+\;\dblparenthesee{\hfil\,#\,\hfil}$&$ \;=\; 0#\hfil$
\cr
\Re\VarphieDede,\Im\Varphielambdade&\Re\VarphieDede,\Re\Varphielambdade&,
\cr\noalign{\vskip4pt}
\Im\VarphieDede,\Im\Varphielambdade&\Im\VarphieDede,\Re\Varphielambdade&,
\cr\noalign{\vskip4pt}
\Im\VarphieDede,\Im\Varphielambdade&\Re\VarphieDede,\Im\Varphielambdade&,
\cr\noalign{\vskip4pt}
\Im\VarphieDede,\Re\Varphielambdade&\Re\VarphieDede,\Re\Varphielambdade&.
\cr}}
\end{equation*}
Then, a direct calculation implies that
\begin{equation*}
\vcenter{\halign{
$#\crochete{\VarphieDede,\Varphielambdade}\hfil$&$\;=\;#,\hfil$&
$\qquad#\dblparenthesee{\VarphieDede,\Varphielambdade}\hfil$&$\;=\;#,\hfil$
\cr
\Re & -\tfrac{3}{\Mh}\thetaepe & \Re & \tfrac{3}{\Mh}\thetaet
\cr\noalign{\vskip4pt}
\Im&0&\Im&0
\cr}}
\end{equation*}
so that 
$\thetaepe = -\tfrac{1}{3}\Mh \crochete{\VarphieDede,\Varphielambdade}$ and
$\thetaet = \tfrac{1}{3}\Mh\dblparenthesee{\VarphieDede,\Varphielambdade}$.
\end{proof}

An alternative form of the diffusion velocity is given by
\begin{multline}
\Ve=-\Deepa\left(\depa+\chiepa\glogTepa\right)
-\Deepe \left( \depe + \chiepe\glogTepe + \chiet\glogTet \right) \\
-\Deet     \left( \det + \chiepe\glogTet - \chiet\glogTepe \right),
\label{eqVealt}
\end{multline}
where the real thermal diffusion ratios $\chiepa$, $\chiepe$, $\chiet$ are defined by the relations
\begin{equation}
\thetaepa=\Deepa\chiepa, \quad 
\thetaepe + \i \thetaet = (\Deepe+\i\Deet)(\chiepe+\i\chiet).\label{eqthetae}
\end{equation}
Then, the viscous tensor is calculated in the following proposition.

\begin{proposition}
The electron viscous tensor vanishes, $i.e.$,
\begin{equation}
\visqueux[\elec]=0.
\end{equation}
\end{proposition}

\begin{proof}
Using definition~\eqref{eqPie} of the stress tensor and expression~\eqref{expphie} of $\phie$, one readily obtains that $\visqueux[\elec]=0$.
\end{proof}

The electron heat flux is given in the following proposition.
\begin{proposition}
\label{th:expressionqe}
The electron heat flux $\heate$ reads
\begin{multline} \label{expheate}
\heate=-\left[\lambdaepaprim\gTepa+\lambdaepeprim\gTepe+\lambdaetprim\gTet\right]\\
-\pre\left(\thetaepa\depa+\thetaepe\depe+\thetaet\det\right)+\rhoe\enthalpiee\Ve
\end{multline}
where the partial thermal conductivities are given by
\begin{equation}\label{explambdaeprim}
\vcenter{\halign{
$#\hfil$&$\;\,= \hfil #$&$\hfil#\hfil$
\cr
\lambdaepaprim&\frac{1}{3\tempe^2}{\Mh} &\crochete{\Varphielambdaun,\Varphielambdaun},
\cr\noalign{\vskip4pt}
\lambdaepeprim&\frac{1}{3\tempe^2}{\Mh} &\crochete{\Varphielambdade,\Varphielambdade},
\cr\noalign{\vskip4pt}
\lambdaetprim&-\frac{1}{3\tempe^2}{\Mh} &\dblparenthesee{\Varphielambdade,\Varphielambdade}.
\cr}}
\end{equation}
\end{proposition}

\begin{proof}
Using definition~\eqref{defheate} of the heat flux $\heate$ and expression~\eqref{defpsiemu} of $\psielambda$ yields
\begin{equation*}
\heate = \rhoe\enthalpiee\Ve - \Mh\ppscale{\psielambda,\feo\phie}.
\end{equation*}
Further substituting expansion~\eqref{expphie} into the latter equation, and using isotropy, we obtain
expression~\eqref{expheate} for the heat flux $\heate$ where the transport coefficients 
$\thetaepa$, $\thetaepe$, $\thetaet$ are given in eq.~\eqref{expDeethetae} and 
the partial thermal conductivities
$\lambdaepaprim$, $\lambdaepeprim$, $\lambdaetprim$
are defined by
$\lambdaepaprim = \frac{1}{3\tempe^2}\Mh \ppscale{\feo\varphielambdaun,\psielambda}$,
$\lambdaepeprim+\i\lambdaetprim = 
\frac{1}{3\tempe^2}\Mh \ppscale{\feo\varphielambdade,\psielambda}$.
Eqs.~\eqref{eqvarphie1} and \eqref{eqvarphie2} for $\mu=\lambda_{\elec}$ classically yields~\cite{ferziger,graille1}
\begin{gather*}
\lambdaepaprim = \tfrac{1}{3\tempe^2}{\Mh} \; 
\crochete{\Varphielambdaun,\Varphielambdaun},\\
\lambdaepeprim+\i\lambdaetprim = \tfrac{1}{3\tempe^2}{\Mh} 
\Bigl(
\crochete{\Varphielambdade,\Varphielambdade}-
\i\dblparenthesee{\Varphielambdade,\Varphielambdade}
\Bigr).
\end{gather*}
As the bracket operators $\crochete{\cdot,\cdot}$ and $\dblparenthesee{\cdot,\cdot}$ are hermitian, we immediately conclude for the expressions of $\lambdaepaprim$, $\lambdaepeprim$, and $\lambdaetprim$.
\end{proof}

Using the thermal diffusion ratios defined in eq.~\eqref{eqthetae}, the electron heat flux is rewritten
\begin{multline} 
\heate=-\left[\lambdaepa\gTepa+\lambdaepe\gTepe+\lambdaet\gTet\right]\\
+\pre\left(\chiepa\Vepa+\chiepe\Vepe+\chiet\Vet\right)+\rhoe\enthalpiee\Ve,\label{eqqealt}
\end{multline}
where the thermal conductivities $\lambdaepa$, $\lambdaepe$, $\lambdaet$ are real quantities given by
\begin{gather*}
\lambdaepa=\lambdaepaprim-\frac{\pre}{\tempe}\Deepa\chiepa\chiepa,\\
\lambdaepe+\i\lambdaet
=\lambdaepeprim+\i\lambdaetprim
-\frac{\pre}{\tempe}(\Deepe+\i\Deet)(\chiepe+\i\chiet)(\chiepe+\i\chiet).
\end{gather*}

The positivity properties associated with the heat flux and the diffusion velocities can 
be written with the help of the mass-energy transport matrices
\begin{equation}\label{eqmatmat}
\Aepa = \begin{pmatrix}
\frac{\tempe}{\pre} \lambdaepaprim & \thetaepa \\
\thetaepa & \Deepa
\end{pmatrix},\quad
\Aepe = \begin{pmatrix}
\frac{\tempe}{\pre} \lambdaepeprim & \thetaepe \\
\thetaepe & \Deepe
\end{pmatrix},\quad
\Aet = \begin{pmatrix}
\frac{\tempe}{\pre} \lambdaetprim & \thetaet \\
\thetaet & \Deet
\end{pmatrix}.
\end{equation}

\begin{proposition}\label{th:positivitefluxepapet}
Considering any two-dimensional  real vectors $\Xpa$, $\Xpe$, and $\Xt$, the two following inequalities are satisfied
\begin{gather}
\pscal{\Aepa\Xpa,\Xpa}\geq 0, \label{posApa}\\
\pscal{\Aepe\Xpe,\Xpe} + \pscal{\Aepe\Xt,\Xt} + \pscal{\Aet\Xpe,\Xt} - \pscal{\Aet\Xt,\Xpe} \geq 0.
\label{posApeAt}
\end{gather}
\end{proposition}

\begin{proof} Introducing
$\Xpa = (\Xpaun,\Xpade)$, $\Xpe=(\Xpeun,\Xpede)$,  and $\Xt=(\Xtun,\Xtde)$, expressions \eqref{expDeethetae} and \eqref{explambdaeprim} for transport coefficients yield 
\begin{gather*}
\pscal{\Aepa\Xpa,\Xpa} = \tfrac{\pre\tempe\Mh}{3}
\crochete{\Yun,\Yun},\\
\pscal{\Aepe\Xpe,\Xpe} + \pscal{\Aepe\Xt,\Xt} + \pscal{\Aet\Xpe,\Xt} - \pscal{\Aet\Xt,\Xpe} 
= \tfrac{\pre\tempe\Mh}{3}
\crochete{\Yde,\Yde},
\end{gather*}
with
\begin{align*}
\Yun &= \Xpade\VarphieDeun - \tfrac{1}{\pre\tempe}\Xpaun\Varphielambdaun,\\
\Yde &= (\Xpede+\i\Xtde) \VarphieDede
- \tfrac{1}{\pre\tempe} (\Xpeun+\i\Xtun) \Varphielambdade.
\end{align*}
Inequalities \eqref{posApa} and \eqref{posApeAt} are then obtained thanks to the positivity of the bracket operator $\crochete{\cdot,\cdot}$.
\end{proof}

\begin{remark}
In the limit case $\B$ tends to zero, the behavior of the transport coefficients can be investigated. We formally prove that the matrix $\Aepa$ does not depend on the magnetic field, that $\Aepe$ converges to $\Aepa$, and that $\Aet$ vanishes. So that we obtain in the limit case the same contribution as with zero magnetic field ($b=0$) for the electron diffusion velocities and heat flux~\cite{magin1}.
\end{remark}

\subsection{First-order heavy-particle perturbation function}
\label{sec:coefftransportH}

The first-order perturbation function $\phiH = {(\phii)}_{i\in\lourd}$  is solution to eq.~\eqref{eqiepsilon-1}, $i.e.$,
\begin{equation*}
\boltFi(\phi_\heavy)=\psii+\frac{1}{\fHio}\Jizhat,\quad i\in\lourd,
\end{equation*}
and satisfies the constraints~\eqref{eqiconsepsilon-1}, where $\psii=-\Diz(\fHio)/\fHio$, $i\in\lourd$. After some lengthy calculation based on the expression~\eqref{eqbolti} of $\fHo$,  the Euler eqs.~\eqref{eqeuler1}, \eqref{eqeuler3}, and \eqref{eqeuler2bis}, theorem~\ref{th:thmJie}, and corollary~\ref{th:Jiediso}, one obtains 
\begin{equation}
\boltFi(\phi_\heavy)=-\psiieta\pmat\dx\vitesse-\prh\sumi[j]\psiiDj ij \dscal\dihat[j]-\psiilambda\dscal\dx\left(\frac{1}{\temph}\right)-\psiipres(\tempe-\temph),\label{eqphiphi}
\end{equation}
where
\begin{equation}\label{eqPsii}
\left\{
\vcenter{\halign{
$#\hfil$&$\;\,=\; #,\hfil$&$\quad#\hfil$
\cr
\psiieta&\frac{\mi}{\temph}\bigl(\Ci\ptens\Ci-\frac{1}{3}\Ci\dscal\Ci\identite\bigr)
&i\in\lourd,
\cr\noalign{\vskip4pt}
\psiiDj ij&\frac{1}{\Mh\pri}\bigl(\delta_{ij}-\frac{\rhoi}{\rhoH}\bigr)\Ci
&i,j\in\lourd,
\cr\noalign{\vskip4pt}
\psiilambda&\frac{1}{\Mh}\bigl(\frac{5}{2}\temph-\frac{1}{2}\mi\Ci\dscal\Ci\bigr)\Ci
&i\in\lourd,
\cr\noalign{\vskip4pt}
\psiipres&
\frac{2}{\temph^2}
\bigl(\frac{\colfreqiez}{3\mi}-\sumi[j]\frac{\ni[j]\colfreqiez[j]}{\nH\mi[j]}\bigr)
\bigl(\frac{3}{2}\temph-\frac{1}{2}\mi\Ci\dscal\Ci\bigr)
&i\in\lourd.
\cr}} \right.
\end{equation}
Quantity $\pri=\ni\temph$ stands for the partial pressure of species $i\in\lourd$. A linearly independent family of diffusion driving forces is also introduced
\begin{equation}
\dihat=\frac{1}{\prh}\dx \pri -\frac{\ni\qi}{\prh}\Ep-\frac{\ni\Mh}{\prh}\Aie,\quad i\in\lourd.
\end{equation}
The average electron forces acting on the heavy particles  belong to the category of the diffusion driving forces and allows for a coupling between the heavy particles and electrons. Expression of $\phie$ given in eq.~\eqref{expphie}  and definition~\eqref{eqdefFie} implies that  $\Aie$, $i\in\lourd$,  is proportional to the electron diffusion driving force and the electron temperature gradient. Thus, the heavy-particle transport fluxes  to be derived are also expected to be proportional to the electron forces.

The existence and uniqueness of a solution to eq.~\eqref{eqphiphi} is then established in the following proposition.

\begin{proposition}
 The scalar functions family $\phiH = {(\phii)}_{i\in\lourd}$, given by
 \begin{equation}
\phii=-\Phiieta\pmat\dx\vitesse-\prh\sumi[j]\PhiiDj ij \dscal\dihat[j]-\Phiilambda\dscal\dx\left(\frac{1}{\temph}\right)-\Phiipres(\tempe-\temph),\quad i\in\lourd,
\label{eqphiiexp}
\end{equation}
is the solution to eq.~\eqref{eqphiphi} under the constraints~\eqref{eqiconsepsilon-1},
where the tensorial functions family $\PhiHeta = {(\Phiieta)}_{i\in\lourd}$, the vectorial functions families 
$\PhiHDj = {(\PhiiDj{i}{j})}_{i\in\lourd}$, $j\in\lourd$, and
$\PhiHlambda = {(\Phiilambda)}_{i\in\lourd}$, and the scalar functions family 
$\PhiHpres = {(\Phiipres)}_{i\in\lourd}$ are the solutions of the problems
\begin{equation}\label{eqPhiHmu}
\boltFi(\PhiHmu)=\psihmu,\quad i\in\lourd,
\end{equation}
under the scalar constraints
\begin{equation}\label{constPhiHmu}
\ppscalh{\fHo\PhiHmu,\invH[l]}=0,\quad l\in\{1,\ldots,\nh+4\},
\end{equation}
with $\mu\in\{\eta,(\symbolDkl_j)_{j\in\lourd},\lambdahat_\heavy,\Thetah\}$.
\end{proposition}

\begin{proof}
By linearity and isotropy of the linearized Boltzmann operator $\boltFi$, the development of $\psii$ can be followed through for $\phii$ as well
\begin{equation*}
\phii=-\Phiieta\pmat\dx\vitesse-\prh\sumi[j]\PhiiDj ij \dscal\dihat[j]-\Phiilambda\dscal\dx\left(\frac{1}{\temph}\right)-\Phiipres(\tempe-\temph),\quad i\in\lourd,
\end{equation*}
where the functions families $\PhiHmu$, for $\mu\in\{\eta,(\symbolDkl_j)_{j\in\lourd},\lambdahat_\heavy,\Thetah\}$ satisfy eq.~\eqref{eqPhiHmu} under the scalar constraints \eqref{constPhiHmu}.
We seek a solution in the form
\begin{align*}
\Phiimu&=\Phiimuun\Ci, & i\in\lourd,
&\quad \mu\in\{(\symbolDkl_j)_{j\in\lourd},\lambdahat_\heavy\},\\
\Phiieta&=\PhiHetaun\left(\Ci\ptens\Ci-\tfrac{1}{3}\Ci\dscal\Ci\identite\right),&i\in\lourd.
\end{align*}
Quantities $\Phiimuun$,  $\mu\in\{\eta,(\symbolDkl_j)_{j\in\lourd},\lambdahat_\heavy\}$, and $\Phiipres$ are scalar functions of $\Ci\dscal\Ci$, for $i\in\lourd$,
since $\PhiHmu$,  $\mu\in\{\eta,(\symbolDkl_j)_{j\in\lourd},\lambdahat_\heavy,\Thetah\}$ must be invariant under a change of coordinates. 
Uniqueness of the solution is readily proved based on the linearity property of the operator $\boltFi[\heavy]$, its kernel given in property~\ref{th:propboltFi}, and the constraints~\eqref{eqiconsepsilon-1} satisfied by $\phi_\heavy$.
 \end{proof}
 
We further introduce the heavy-particle bracket operator $\crocheth{\cdot,\cdot}$ associated with the operator $\boltFi[\heavy]$. For any $\xi_\heavy,\zeta_\heavy$, we define
\begin{equation*}
\crocheth{\xi_\heavy,\zeta_\heavy}=\ppscalh{\fHo\xi_\heavy,\boltFi[\heavy](\zeta_\heavy)}.
\end{equation*}
The bracket operator develops into
\begin{equation*}
\crocheth{\xi_\heavy,\zeta_\heavy}=
\tfrac{1}{4}\sum_{i,j\in\lourd}\int\fHio\fHio[j] 
(\xi_i+\xi_j-\xi_i^{\prime}-\xi_j^{\prime})\pcont(\overline{{\zeta}_i+{\zeta}_j-{\zeta}_i^{\prime}-{\zeta}_j^{\prime}})
|\Ci-\Ci[j]|\sigma_{ij} \d\ovec \d\Ci \d\Ci[j].
\end{equation*}
The bracket operator $\crocheth{\cdot,\cdot}$ is hermitian
$\crocheth{\xi_\heavy,\zeta_\heavy} = \crochethbar{\zeta_\heavy,\xi_\heavy}$, positive semi-definite $\crocheth{\xi_\heavy,\xi_\heavy}\geq0$, and its kernel is spanned by the collisional invariants, $i.e.$, $\crocheth{\xi_\heavy,\xi_\heavy}=0$ implies that $\xi_\heavy$ is a (tensorial) collisional invariant, or in other words, that all its tensorial components are in the space 
$\invspace_\heavy$.
The expression of the heavy-particles diffusion velocities is given in the following proposition.

\begin{proposition}
The diffusion velocity of species $i\in\lourd$ reads
\begin{equation}
\Vi=-\sumi[j]\Dij\dihat[j]-\thetahi\glogTh,\label{eqtranspoVi}
\end{equation}
where the diffusion coefficients and thermal diffusion coefficients are given by
\begin{equation} \label{eqdefDijthetahi}
\vcenter{\halign{
$#\hfil$&$\;\,=\; #,\hfil$&$\quad#\hfil$
\cr
\Dij
&\frac{1}{3}{\prh \temph  \Mh}\crocheth{\PhiiDj \heavy i,\PhiiDj \heavy j }
&i,j\in\lourd,
\cr\noalign{\vskip4pt}
\thetahi
&-\frac{1}{3}{\Mh}\crocheth{\PhiiDj \heavy i,\Phiilambda[\heavy]}
&i\in\lourd.
\cr}}
\end{equation}
\end{proposition}

\begin{proof}
Using definition~\eqref{eqVi} of the diffusion velocity and expression~\eqref{eqPsii} of $\psiiDj \heavy i$, $i\in\lourd$, yields
\begin{equation*}
\Vi=\temph\Mh\ppscalh{\psiiDj \heavy i,\fHo\phii[\heavy]},\quad i\in\lourd.
\end{equation*}
Further substituting expansion~\eqref{eqphiiexp} into the latter equation, we obtain expression~\eqref{eqtranspoVi} of the diffusion velocities.
\end{proof}
In particular, the heavy-particle diffusion velocities are thus proportional to the electron driving force and electron temperature gradient through the $\Aie$ contribution to $\di$, $i\in\lourd$.  Kolesnikov~\cite{anatoliy2} has already introduced electron heavy-particle diffusion coefficients and thermal diffusion coefficients and ratios  to couple the heavy-particle diffusion velocities to the electron forces. Therefore, we propose to refer to this phenomenom as the Kolesnikov effect for the heavy particles.

>From the properties of the bracket operator, we infer that the diffusion matrix $\Dij[]$ is symmetric. Moreover,  an alternative form of the diffusion velocities is given by
\begin{equation}
\Vi=-\sumi[j]\Dij\left(\dihat[j]+\chihi[j]\glogTh\right),\quad i\in\lourd,\label{eqtranspoVialt}
\end{equation}
where the thermal diffusion ratios are defined from the relations
\begin{equation}\label{eqchih}
\left\{ \begin{aligned}
&\sumi[j]\Dij\chihi[j]=\thetahi,\quad i\in\lourd,\\ 
&\sumi[j]\chihi[j]=0.
\end{aligned}
\right.
\end{equation}
Then, we introduce the tensor $$\tensS=\left[\dx \vitesse +(\dx \vitesse)^T\right]-\frac{2}{3} \dx\dscal\vitesse\;\identite,$$
in order to express the viscous tensor in the following proposition.

\begin{proposition}
The heavy-particle viscous tensor reads
\begin{equation}
\visqueux[\heavy]=-\etah\tensS,\label{eqtranspoPi}
\end{equation}
where the shear viscosity is given by
\begin{equation}
\etah=\frac{\temph}{10}\crocheth{\Phiieta[\heavy],\Phiieta[\heavy]}.
\end{equation}
\end{proposition}

\begin{proof}
Using definition~\eqref{eqPii} of the viscous tensor and expression~\eqref{eqPsii} of $\psiieta[\heavy]$  yields
\begin{equation*}
\visqueux[\heavy]=\temph\ppscalh{\psiieta[\heavy],\fHo\phi_\heavy}.
\end{equation*}
Further substituting expansion~\eqref{eqphiiexp} into the latter equation, we obtain expression~\eqref{eqtranspoPi} of the viscous tensor.
\end{proof}

The expression of the heavy-particle heat flux is given in the following proposition.

\begin{proposition}
The heavy-particle heat flux reads
\begin{equation}
\label{eqtranspoqh} 
\heati[\heavy]=-\lambdaHprim\dx\temph-\prh\sumi[j]\thetahi[j]\dihat[j]+\sumi[j]\rhoi[j]\enthalpiei[j]\Vi[j],
\end{equation}
where the partial thermal conductivity is given by
\begin{equation} \label{eqdeflambdaHprim}
\lambdaHprim=\frac{1}{3\temph^2}{\Mh}\crocheth{\Phiilambda[\heavy],\Phiilambda[\heavy]}.
\end{equation}
\end{proposition}

\begin{proof}
Using definition~\eqref{eqqi} of the heavy-particle heat flux and expression~\eqref{eqPsii} of 
$\psiilambda$ yields
\begin{equation*}
\heati[\heavy]=-\Mh\ppscalh{\psiilambda[\heavy],\fHo\phi_\heavy}+\frac{5}{2}\temph\sumi[j]\ni[j]\Vi[j].
\end{equation*}
Further substituting expansion~\eqref{eqphiiexp} into the latter equation, we obtain expression~\eqref{eqtranspoqh} of the heat flux.
\end{proof}

Using the thermal diffusion ratios defined in eq.~\eqref{eqchih}, the heavy-particle heat flux is rewritten
\begin{equation}
\label{eqtranspoqh2} 
\heati[\heavy]=
-\lambdaH\dx\temph+\prh\sumi[j]\chihi[j]\Vi[j]+\sumi[j]\rhoi[j]\enthalpiei[j]\Vi[j],
\end{equation}
where the thermal conductivity is given by
\begin{equation}
\lambdaH=\lambdaHprim-\nH\sumi[j]\thetahi[j]\chihi[j].
\end{equation}

The positivity properties associated with the heat flux and the diffusion velocities can 
be written with the help of the mass-energy transport matrix
\begin{equation*}
\Ah = \begin{pmatrix}
\frac{\temph}{\prh} \lambdaHprim & [(\thetahi)_{i\in\lourd}]^T\\
(\thetahi)_{i\in\lourd}& (\Dij)_{i,j\in\lourd}
\end{pmatrix}.
\end{equation*}

\begin{proposition}\label{th:positivitefluxh}
The heavy particles mass-energy transport matrix $\Ah$ is symmetric, positive semi-definite,  and its kernel is onedimensional and spanned by the vector
$[0,(\rhoi)_{i\in\lourd}]^T$.
\end{proposition}

\begin{proof} 
We consider a vector $\X$ written in the form 
$\X = [\XhT,(\Xhi)_{i\in\lourd}]^T$ and
introduce the family $\Yh = {(\Yhi)}_{i\in\lourd}$ given by
\begin{equation*}
\Yhi = \sumi[j] \Xhi[j] \PhiiDj{i}{j} - \frac{1}{\prh\temph}\XhT \Phiilambda.
\end{equation*}
Expressions \eqref{eqdefDijthetahi} and \eqref{eqdeflambdaHprim} for transport coefficients yield 
\begin{equation*}
\pscal{\Ah\X,\X} = \tfrac{1}{3}\prh\temph\Mh \; \crocheth{\Yh,\Yh}.
\end{equation*}
The positivity is then obtained thanks to the positivity of the heavy-particle bracket operator $\crocheth{\cdot,\cdot}$. 
Moreover, using the scalar constraints \eqref{constPhiHmu}, 
that imply that $\Yh$ is orthogonal to the collisional invariants,
the quantity $\pscal{\Ah\X,\X}$ vanishes if $\Yh$ is a collisional invariant,
consequently if $\Yh=0$. Finally, the linear rank of the family 
$(\Phiilambda,\PhiiDj{i}{1},\ldots,\PhiiDj{i}{\nh})$ is exactly $\nh$ because it is the rank of the corresponding right member $(\psiilambda,\psiiDj{i}{1},\ldots,\psiiDj{i}{\nh})$. We then conclude that $\Yh=0$ if and only if $\X$ lies in the space spanned by the vector
$[0,(\rhoi)_{i\in\lourd}]^T$.
\end{proof}

\subsection{Second-order electron perturbation function}

The second-order perturbation function $\phide$ is a solution to  eq.~\eqref{eqepsilono}, $i.e.$,
\begin{equation}\label{eqphiedeux}
\boltFe(\phide)+\qe\dCe(\phidi[\elec])\dscal\Ce\pvect\B=
\psied,
\end{equation}
and satisfies the constraints~\eqref{eqconsepsilono}, where 
\begin{equation*}
\psied=\frac{1}{\feo}\Bigl(-\Dezhat(\feo,\phie)+\Jee(\feo\phie,\feo\phie)+\Jezhat\Bigr).
\end{equation*}
Introducing second-order heavy-particle diffusion driving forces $\did = - \ni\Vi$, $i\in\lourd$, 
one obtains after some lengthy calculation 
\begin{equation*}\label{eqsimplifiedphide} 
\psied=-\psieeta\pmat\dx\vitesse
-\pre \sumi \psieDj\dscal\did - \psiedtilde,
\end{equation*}
where $\psiedtilde$ is a scalar function of $\Ce\dscal\Ce$, and
\begin{equation}
\left\{
\vcenter{\halign{
$#\hfil$&$\;\,=\; #\hfil$
\cr
\psieeta&
\frac{1}{\tempe}(\Ce\ptens\Ce-\frac{1}{3}\Ce\dscal\Ce\identite),
\cr\noalign{\vskip4pt}
\psieDj & \frac{1}{\pre\tempe}\Qijl \elec i 1 (|\Ce|^2) |\Ce|\Ce,\quad i \in\lourd.
\cr}} \right.
\end{equation}
The coupling  of the electrons with the heavy particles occurs in the integral equation for the second-order perturbation function through the $\did$ forces, $i\in\lourd$. Thus, the second-order electron transport fluxes  to be derived are also expected to be proportional to the heavy-particle forces.

The complete resolution of eq.~\eqref{eqphiedeux} is not necessary since we only need to express the second-order transport fluxes $\Vde$ and $\heatde$ in terms of bracket operators. Consequently, we only have to examine the contribution of the two vectorial terms $\psieDe$ and $\psieDj$, $i\in\lourd$.

\begin{proposition}
The scalar function $\phie$ given by
\begin{equation}\label{eqphide}
\phide = - \Phieeta\pmat\dx\vitesse -  \pre \sumi  
\Re\Bigl[
\Mpa \VarphieDjun[i] + (\Mpe{+}\i\Mt) \VarphieDjde[i]
\Bigr] \dscal \did - \phidetilde,
\end{equation}
is the solution to eq.~\eqref{eqphiedeux} under the constraints~\eqref{eqconsepsilono}.
The vectorial functions $\VarphieDjun[i]$, $\VarphieDjde[i]$, $i\in\lourd$, are the solutions of the problems
\begin{gather}
\boltFe(\VarphieDjun[i])=\psieDj[i]\label{eqvarphied1},\\
\bigl(\boltFe+\i|\B|\boltFev\bigr)\bigl(\VarphieDjde[i]\bigr)=\psieDj[i],\label{eqvarphied2}
\end{gather}
under the constraints
\begin{align}
\ppscale{\feo\VarphieDjun[i],\inve[l]}&=0,&& l\in\{1,2\},\label{eqconstvarphied1}\\
\ppscale{\feo\VarphieDjde[i],\inve[l]}&=0,&& l\in\{1,2\}.\label{eqconstvarphied2}
\end{align}
The tensorial function $\Phieeta$ verifies
\begin{equation*}
\boltFe(\Phieeta)+\qe\dCe(\Phieeta)\dscal\Ce\pvect\B = \psieeta,
\end{equation*}
and the function $\phidetilde$ is a scalar function of $\Ce\dscal\Ce$ and  $(\Ce\dscal\B)^2$.
\end{proposition}

\begin{proof}
The proof of this proposition is identical to the one of proposition~\ref{th:existencephie} since eqs.~\eqref{eqphie} and \eqref{eqphiedeux} for $\phie$ and $\phide$ only differ with their second members.
\end{proof}

The expressions of the second-order electron diffusion velocity and heat flux and of the average electron force are given in the following proposition.
\begin{proposition}
\label{th:expressionVede}
The second-order electron diffusion velocity $\Vde$ is given by
\begin{equation}\label{expVde}
\Vde=-\sumi \bigl(\Dejpa[i]\didpa+\Dejpe[i]\didpe+\Dejt[i]\didt\bigr).
\end{equation}
The second-order electron heat flux $\heatde$ reads
\begin{equation} \label{expheatde}
\heatde=- \pre \sumi
\left(\thetaeipa\didpa+\thetaeipe\didpe+\thetaeit\didt\right)
+\rhoe\enthalpiee\Vde.
\end{equation}
The average electron force $\Aie$ acting on heavy particles $i\in\lourd$ is given by
\begin{multline}
\Aie = -\frac{\pre}{\Mh} \bigl(
\Dejpa[i]\depa+\Dejpe[i]\depe+\Dejt[i]\det \bigr) \\
-\frac{\pre}{\Mh} \bigl(
\thetaeipa\glogTepa+\thetaeipe\glogTepe+\thetaeit\glogTet
\bigr).
\end{multline}
The diffusion coefficients and thermal diffusion coefficients read
\begin{equation}\label{expDei}
\vcenter{\halign{
$#\hfil$&$\;=\hfil#$&$\hfil#,\hfil$&$\quad#\hfil$&$\;=\hfil#$&$\hfil#,\hfil$&$\quad i\in\lourd#$
\cr
\Dejpa[i]&\frac{1}{3}{\pre\tempe  \Mh} &\crochete{\VarphieDeun,\VarphieDjun[i]}&
\thetaeipa&-\frac{1}{3}{\Mh} &\crochete{\VarphieDjun[i],\Varphielambdaun}
&,
\cr\noalign{\vskip4pt}
\Dejpe[i]&\frac{1}{3}{\pre\tempe \Mh}  &\crochete{\VarphieDede,\VarphieDjde[i]}&
\thetaeipe&-\frac{1}{3}{\Mh} &\crochete{\VarphieDjde[i],\Varphielambdade}
&,
\cr\noalign{\vskip4pt}
\Dejt[i]&-\frac{1}{3}{\pre\tempe \Mh} &\dblparenthesee{\VarphieDede,\VarphieDjde[i]}&
\thetaeit&\frac{1}{3}{\Mh} &\dblparenthesee{\VarphieDjde[i],\Varphielambdade}
&.
\cr}}
\end{equation}
Note that the previous expressions are real, in particular for $\Dejpe[i]$, $\Dejt[i]$, $\thetaeipe$, and $\thetaeit$, $i\in\lourd$, although functions $\Varphielambdade$, $\VarphieDede$ and $\VarphieDjde[i]$ are complex.
\end{proposition}

\begin{proof}
Using definition \eqref{eqdefVede} (respectively \eqref{eqdefqde} and \eqref{eqdefFie}) of the second-order diffusion velocity $\Vde$ (respectively the second-order electron heat flux $\heatde$ and average electron force $\Aie$, $i\in\lourd$), the same proof as that of proposition~\ref{th:expressionVe} yields to conclude.
\end{proof}
\begin{remark}
The term $\Phieeta\pmat\dx\vitesse$
of eq.~\eqref{eqphide} contributes to a second-order electron momentum relation not  investigated here.
 \end{remark}
\begin{remark}
To the authors's knowledge, it is the first time that such second-order transport coefficients are rigorously derived from a multiscale analysis.  The second-order electron diffusion velocity and heat flux are thus proportional to the heavy-particle diffusion velocities. That is the Kolesnikov effect for the electrons. However, it is important to mention that the second-order electron transport fluxes shall not be confused with the Burnett transport fluxes appearing in second-order macroscopic equations~\cite{ferziger}.
 \end{remark}

\section{Conservation equations}\label{sec:secentropy}

 We review the  heavy-particle Navier-Stokes eqs.~\eqref{eqNS1}-\eqref{eqNS3} and electron drift-diffusion eqs.~\eqref{eqdrift21} and \eqref{eqdrift22}. We also derive a total energy equation and an entropy equation. Then, we introduce a conservative  formulation of the system of equations.
  
 \subsection{Mass}
 
 The species mass conservation equations read
\begin{gather}
\dt\rhoe + \dx\dscal\Bigl[
\rhoe\bigl(\vitesse+\tfrac{1}{\Mh}(\Ve+\varepsilon\Vde)\bigr)
\Bigr] =  0,\label{eqsum1}\\
\dt\rhoi + \dx\dscal\Bigl[
\rhoi\bigl(\vitesse+\tfrac{\varepsilon}{\Mh}\Vi\bigr)
\Bigr] = 0,\quad i\in\lourd.\label{eqsum2}
\end{gather}
Summing eq.~\eqref{eqsum2} over $i\in\lourd$ and using the constraint 
$\sum_{j\in\lourd}\rhoi[j]\Vi[j]=0$ given in eq.~\eqref{eqiconsepsilon-1}, a heavy-particle mass  conservation equation is obtained
\begin{equation}\label{coucou}
\dt\rhoH + \dx\dscal(\rhoH\vitesse) = 0.
\end{equation}
The heavy-particle mass is conserved in the mean heavy-particle velocity referential. Then, adding the electron drift eq.~\eqref{eqsum1} to eq.~\eqref{coucou} and using
eq.~\eqref{eqrelation}, $i.e.$,
\begin{equation*}
\rho\speed  = \rhoi[\heavy]\vitesse + \varepsilon^2\rhoe \bigl( \vitesse + \tfrac{1}{\Mh}(\Ve + \varepsilon\Vde)\bigr),
\end{equation*}
 a  conservation equation of global mass $\rho=\rhoi[\heavy]+\varepsilon^2\rhoe$ is also established
 \begin{equation}
\dt\rho + \dx\dscal(\rho\speed) = 0.
\end{equation}
The global mass is conserved in the hydrodynamic referential, 
although the transport fluxes are calculated in the mean heavy-particle velocity referential. 
It is the only place where the difference between the global hydrodynamic velocity and 
the mean heavy-particle velocity, of the order of $\varepsilon^2$, plays an essential role.
It is another evidence of the coherence of our formalism compared to other approaches found in the
literature.

 \subsection{Momentum}
  The momentum conservation is expressed by 
\begin{equation}
\dt(\rhoH\vitesse)+\dx\dscal(\rhoH\vitesse\ptens\vitesse+\tfrac{1}{\Mh^2} p\identite)=
-\tfrac{\varepsilon}{\Mh^2}\dx\dscal\visqueux[\heavy]+\tfrac{1}{\Mh^2}nq\E 
+[\delta_{b0}\courantel_0+\delta_{b1}\courantel]\pvect\B. \label{eqsum3}
\end{equation}
A flow kinetic energy is obtained by projecting the previous equation onto the mean heavy-particle velocity
\begin{multline}
\dt(\tfrac{1}{2}\rhoH|\vitesse|^2)+\dx\dscal\Bigl[
\vitesse\bigl(\tfrac{1}{2}\rhoH|\vitesse|^2+ \tfrac{1}{\Mh^2} {p}\bigr)\Bigr]
=\tfrac{1}{\Mh^2}p\;\dx\dscal\vitesse-\tfrac{\varepsilon}{\Mh^2}\vitesse\dscal\dx\dscal\visqueux[\heavy]\\
+\tfrac{1}{\Mh^2}nq\E\dscal\vitesse+\vitesse\dscal(\delta_{b0}\courantel_0+\delta_{b1}\courantel)\pvect\B
 \label{eqkinenergy}.
\end{multline}

\subsection{Energy} 

The electron energy equation reads
\begin{multline}\label{eqsum4}
\dt(\rhoe\energiee)+\dx\dscal\left(\rhoe\energiee\vitesse\right)=-\pre\dx\dscal\vitesse-\tfrac{1}{\Mh}\dx\dscal\left(\heate+\varepsilon\heatde\right)\\
+\tfrac{1}{\Mh}\left(\JJe+\varepsilon\Jde\right)\dscal\Ep
+\delta_{b0}{\varepsilon}{\Mh}\JJe\dscal\vitesse\pvect\B+\deltaEeo+\varepsilon\deltaEeu,
\end{multline}
and the heavy-particle energy equation reads
\begin{multline}\label{eqsum5}
\dt(\rhoH\energiei[\heavy])+\dx\dscal(\rhoH\energiei[\heavy]\vitesse)=
-(\pri[\heavy]\identite+\varepsilon\visqueux[\heavy])\pmat\dx\vitesse-\tfrac{\varepsilon}{\Mh}\dx\dscal\heati[\heavy] +\tfrac{\varepsilon}{\Mh}\JJi[\heavy]\dscal\Ep\\
+\deltaEho+\varepsilon\deltaEhu.
\end{multline}
So that a global energy equation is derived by summing eqs.~\eqref{eqsum4} and \eqref{eqsum5}
\begin{multline}
\dt(\rho\energiei[])+\dx\dscal\left(\rho\energiei[]\vitesse\right)=
-(p\identite+\varepsilon\visqueux[\heavy])\pmat\dx\vitesse
-\tfrac{1}{\Mh}\dx\dscal\fluxchaleur\\
+\tfrac{1}{\Mh}\left(\JJe+\varepsilon\Jde+\varepsilon\JJi[\heavy]\right)\dscal\Ep
+\delta_{b0}{\varepsilon}{\Mh}\JJe\dscal\vitesse\pvect\B\label{eqsum5bis},
\end{multline}
where quantity  $\fluxchaleur=\heate+\varepsilon\heatde+\varepsilon\heati[\heavy]$ is  the total heat flux.
Finally, a total energy equation is derived by adding eq.~\eqref{eqkinenergy}
\begin{equation}\label{eqsum5ter}
\dt(\energie)+\dx\dscal\left(\enthalpie\vitesse\right)=-\varepsilon\dx\dscal(\visqueux[\heavy]\dscal\vitesse)-\tfrac{1}{\Mh}\dx\dscal\fluxchaleur+\courantel\dscal\E,
\end{equation}
where quantity $\energie=\rho e+\Mh^2\rhoH\tfrac{1}{2}|\vitesse|^2$ stands for  the total energy and $\enthalpie=\rho\energie+{p} $, the total enthalpy. The term $\courantel\dscal\E$ of eq.~\eqref{eqsum5ter} represents the power developed by the electromagnetic field. It has the form prescribed by Poynting's theorem. Hence, the first principle of thermodynamics is satisfied.

\subsection{Electron and heavy-particle entropy equations}

In addition to the thermal energy, we introduce other relevant thermodynamic functions. First, the species Gibbs free energy is defined by the relations
\begin{equation}
\rhoe\gibbse=\ne\tempe\ln\left(\frac{\ne n^0}{\tempe^{3/2}\partitione}\right),\quad\rhoi\gibbsi=\ni\temph\ln\left[\frac{\ni n^0}{\left(\mi\temph\right)^{3/2}\partitioni[\heavy]}\right],\quad i\in\lourd,
\end{equation}
where the translational partition functions  read
\begin{equation}
\partitione=\left(\frac{2\pi \me^0\boltz T^0}{\planck^2}\right)^{3/2},\quad
\partitioni[\heavy]=\left(\frac{2\pi m_h^0\boltz T^0}{\planck^2}\right)^{3/2}.
\end{equation}
Then, the species enthalpy is given by 
\begin{equation}
\rhoe\enthalpiee=\frac{5}{2}\ne\tempe,\quad
\rhoi\enthalpiei=\frac{5}{2}\ni\temph,\quad i\in\lourd.
\end{equation}
Finally, the species entropy is introduced as
\begin{equation}
\entropiee=\frac{\enthalpiee-\gibbse}{\tempe},\quad
\entropiei=\frac{\enthalpiei-\gibbsi}{\temph},\quad i\in\lourd.\label{eqthermos}
\end{equation}
Therefore, the mixture entropy reads $\rho\entropie=\sum_{j\in\espece}\rhoi[j]\entropiei[j]$. The thermodynamic functions exhibit a wider range of validity than in classical thermodynamics, introduced for stationary homogeneous equilibrium states~\cite{giovangigli1}.  Indeed,  they are interpreted in the framework of kinetic theory by establishing a relation between the thermodynamic entropy  and the kinetic entropy. The latter quantity is  based upon the distribution functions
\begin{multline}\label{eqkinentropy}
\kinentropy=\sumi[j]\int \fHi[j]\left\{1-\ln\left[\frac{(2\pi)^{3/2}n^0}{\mi[j]^3\partitioni[\heavy]}\fHi[j]\right]\right\}\d\Ci[j]\\
+\int \fe\left\{1-\ln\left[\frac{(2\pi)^{3/2}n^0}{\partitione}\fe\right]\right\}\d\Ce.
\end{multline}

\begin{proposition}
The kinetic entropy and the thermodynamic entropy are asymptotically equal at order $\varepsilon^2$, $i.e.$,
\begin{equation}
\kinentropy=\rho\entropie+\ordre(\epsilon^2),\label{eqkinkin}
\end{equation}
provided that the distribution functions follow the Enskog expansion given in 
eqs. \eqref{conste} and \eqref{constH}.
\end{proposition}

\begin{proof}
Using definition~\eqref{eqkinentropy} and expansions~\eqref{chapenske} and \eqref{chapenski}, the kinetic entropy is found to be
\begin{multline*}
\sumi[j]\int \fHio[j]\left\{1-\ln\left[\frac{(2\pi)^{3/2}n^0}{\mi[j]^3\partitioni[\heavy]}\fHio[j]\right]\right\}\d\Ci[j]+\int \feo\left\{1-\ln\left[\frac{(2\pi)^{3/2}n^0}{\partitione}\feo\right]\right\}\d\Ce\\
+\varepsilon\sumi[j]\int \fHio[j]\phii[j]\ln\left[\frac{(2\pi)^{3/2}n^0}{\mi[j]^3\partitioni[\heavy]}\fHio[j]\right]\d\Ci[j]
+\varepsilon \int \feo\phie\ln\left[\frac{(2\pi)^{3/2}n^0}{\partitione}\feo\right]\d\Ce\\
+\ordre(\epsilon^2).
\end{multline*}
The first-order term vanishes seeing the constraints~\eqref{eqconsepsilon-1}
and \eqref{eqiconsepsilon-1}.  Then, using expressions~\eqref{eqbolte} and \eqref{eqbolti} and definition~\eqref{eqthermos}, eq.~\eqref{eqkinkin} is readily obtained.
\end{proof}

Consequently, a first-order conservation equation of thermodynamic entropy can be used  instead of a conservation equation of kinetic entropy to ensure that the second principle of thermodynamics is satisfied. First, we introduce the heavy-particle entropy $\rhoi[\heavy] \entropiei[\heavy]=\sum_{j\in\lourd}\rhoi[j] \entropiei[j]$ and derive the entropy equations.

\begin{proposition}
The electron and heavy-particle entropy equations associated with the macroscopic conservation equations~\eqref{eqsum1}-\eqref{eqsum5} read
\begin{gather}
\label{eqentroe} 
\dt(\rhoe \entropiee)+\dx\dscal\left(\rhoe\entropiee\vitesse\right)
+\dx\dscal (\entropyfoe{+}\varepsilon\entropyfue)=
\entropysoe{+}\varepsilon\entropysue,\\
\label{eqentroh}
\dt(\rhoi[\heavy] \entropiei[\heavy])
+\dx\dscal\left(\rhoi[\heavy]\entropiei[\heavy]\vitesse\right)
+\varepsilon\dx\dscal \entropyfuh=
\entropysoh+\varepsilon\entropysuh,
\end{gather}
where the electron and heavy-particle entropy fluxes are given by
\begin{gather}
\label{eqentropyflux1}
\entropyfoe=\frac{1}{\Mh\tempe}(\heate-\rhoe\gibbse\Ve),
\qquad
\entropyfue=\frac{1}{\Mh\tempe}(\heatde-\rhoe\gibbse\Vde),\\
\label{eqentropyflux2}
\entropyfuh=\frac{1}{\Mh\temph}\Bigl(\heati[\heavy]-\sumi[j]\rhoi[j]\gibbsi[j]\Vi[j]\Bigr),
\end{gather}
and the electron and heavy-particle entropy production rates by
\begin{align}
\label{eqentropyprodeo}
\entropysoe&= \frac{1}{\tempe}\deltaEeo
-\frac{\pre}{\Mh\tempe} \de\dscal\Ve
-\frac{1}{\Mh\tempe}\glogTe\dscal(\heate-\rhoe\enthalpiee\Ve),\\
\label{eqentropyprodeu}
\entropysue&= \frac{1}{\tempe}\deltaEeu
-\frac{\pre}{\Mh\tempe} \de\dscal\Vde
-\frac{1}{\Mh\tempe}\glogTe\dscal(\heatde-\rhoe\enthalpiee\Vde),\\
\label{eqentropyprodho}
\entropysoh&= \frac{1}{\temph}\deltaEho,\\
\entropysuh&= \frac{1}{\temph}\deltaEhu 
-\frac{1}{\temph}\visqueux[\heavy]\pmat\dx\vitesse
-\frac{\prh}{\Mh\temph} \sumi 
\frac{1}{\prh} \bigl(\dx \pri -\ni\qi\Ep\bigr)
\dscal\Vi \nonumber\\
\label{eqentropyprodhu}
& \qquad\qquad\qquad\qquad\qquad\qquad
-\frac{1}{\Mh\temph}\glogTh\dscal\Bigl(\heati[\heavy]-\sumi\rhoi\enthalpiei\Vi\Bigr).
\end{align}
\end{proposition}

\begin{proof}
Based on the relations
\begin{eqnarray*}
\rhoe\d\left(\frac{\gibbse}{\tempe}\right)=\d\ne-\frac{3\ne}{2\tempe}\d\tempe,\quad\rhoi\d\left(\frac{\gibbsi}{\temph}\right)=\d\ni-\frac{3\ni}{2\temph}\d\temph,
\quad i\in\lourd,
\end{eqnarray*}
and definition~\eqref{eqthermos}, one obtains
\begin{multline*}
\dt(\rhoe \entropiee)+\dx\dscal\left(\rhoe\entropiee\vitesse\right)=\frac{1}{\tempe}\left[\dt(\rhoe \energiee)+\dx\dscal\left(\rhoe\energiee \vitesse\right)\right]+\ne\dx\dscal\vitesse\\
-\left[\dt\rhoe+\dx\dscal\left(\rhoe\vitesse\right)\right]\frac{\gibbse}{\tempe},
\end{multline*}
\begin{multline*}
\dt(\rhoi[\heavy] \entropiei[\heavy])+\dx\dscal\left(\rhoi[\heavy]\entropiei[\heavy]\vitesse\right)=\frac{1}{\temph}\left[\dt(\rhoi[\heavy] \energiei[\heavy] )+\dx\dscal\left(\rhoi[\heavy] \energiei[\heavy] \vitesse\right)\right]
+\nH\dx\dscal\vitesse\\
-\sumi[j]\left[\dt\rhoi[j]  +\dx\dscal\left(\rhoi[j] \vitesse\right)\right]\frac{\gibbsi[j]}{\temph}.
\end{multline*}
Then, using eqs.~\eqref{eqsum1}, \eqref{eqsum2}, \eqref{eqsum4}, \eqref{eqsum5}, and the relations
\begin{eqnarray*}
\d\left(\frac{\gibbse}{\tempe}\right)=-\frac{\enthalpiee}{\tempe^2}\d\tempe+\frac{1}{\pre}\d\pre,\quad\d\left(\frac{\gibbsi}{\temph}\right)=-\frac{\enthalpiei}{\temph^2}\d\temph+\frac{1}{\mi\pri}\d\pri,
\quad i\in\lourd,
\end{eqnarray*}
we readily obtain eqs.~\eqref{eqentroe} and \eqref{eqentroh}, with  the entropy fluxes given in eqs.~\eqref{eqentropyflux1} and \eqref{eqentropyflux2} and the entropy production rates given in eqs.~\eqref{eqentropyprodeo}-\eqref{eqentropyprodhu}.
\end{proof}

Adding eqs.~\eqref{eqentroe} and \eqref{eqentroh}, a global entropy equation is found
\begin{equation}
\dt(\rho \entropie)+\dx\dscal\left(\rho\entropie\vitesse\right)+\dx\dscal \entropyf=\Upsilon,
\end{equation}
where the global entropy flux is given by
\begin{equation}
\entropyf=\entropyfo_\elec+\varepsilon\entropyfu_\elec+\varepsilon\entropyfu_\heavy,
\end{equation}
and the global entropy production rate by
\begin{equation}\label{eqsum6}
\prodent=\entropysoe+\varepsilon\entropysue+\entropysoh+\varepsilon\entropysuh.
\end{equation}

\begin{proposition}\label{th:verycrazy}
Defining 
$\Xh = (\glogTh,\ditilde[1],\ldots,\ditilde[\nh])^T$,
$\Xepa = (\glogTepa,\depa)^T$, and $\Xepe = (\glogTepe,\depe)^T$,
where
\begin{equation*}
\ditilde = \dihat - \ni \frac{\pre}{\prh}\frac{\temph}{\tempe}
\bigl(\Dejt[i]\det + \thetaeit\glogTet\bigr),
\end{equation*}
the global entropy production rate $\prodent$ defined in eq.~\eqref{eqsum6} can be rewritten in the following form
\begin{multline}\label{eqprodentexpand}
\prodent = \frac{(\tempe-\temph)^2}{\tempe\temph}\sumi[j]
\frac{\ni[j]}{\mi[j]}\colfreqiez[j] + \etah \tensS\pmat\tensS
+ \epsilon \frac{\prh}{\Mh\temph} \pscal{\Ah\Xh,\Xh} \\
+ \frac{\pre}{\Mh\tempe} \pscal{\Aepa\Xepa,\Xepa} 
+ \frac{\pre}{\Mh\tempe} \pscal{\Aehpe\Xepe,\Xepe} ,
\end{multline}
where the matrix $\Aehpe$ 
\begin{equation*}
\Aehpe = \begin{pmatrix}
\frac{\tempe}{\pre} \lambdaepeprim & \thetaepe \\
\thetaepe & \Deepe
\end{pmatrix}
-\epsilon \frac{\pre}{\prh}\frac{\temph}{\tempe}
\sumi[i,j] \Dij \ni\ni[j]
\begin{pmatrix}
\thetaeit\thetaeit[j] & \Dejt[i]\thetaeit[j] \\
\Dejt\thetaeit & \Dejt[i]\Dejt
\end{pmatrix},
\end{equation*}
is a perturbation of the mass-energy transport matrix $\Aepe$ defined in eq.~\eqref{eqmatmat}.
In particular, the global entropy production rate is nonnegative provided that $\epsilon$ is small enough and the collision frequencies $\colfreqiez$, $i\in\lourd$, are nonnegative.
\end{proposition}

\begin{proof}
Expression \eqref{eqprodentexpand} is obtained after some lengthly calculation based on the expressions of the diffusion velocities $\Ve$, $\Vde$, $\Vi$, $i\in\lourd$, heat fluxes $\heate$, $\heatde$, $\heati[\heavy]$, viscous stress tensor $\visqueux[\heavy]$, energy exchange terms $\deltaEeo$, $\deltaEeu$, $\deltaEho$, $\deltaEhu$, and average forces $\Aie$, $i\in\lourd$, given in Section~\ref{sec:sectransport}.

The positivity of the collision frequencies $\colfreqiez$, $i\in\lourd$, (respectively the viscosity $\etah$) immediately yields the positivity of the first term
$(\tempe-\temph)^2/(\tempe\temph)\sum_{j\in\lourd}{\ni[j]}{\colfreqiez[j]}/{\mi[j]}$ (respectively the second term $\etah \tensS\pmat\tensS$).
Moreover, propositions~\ref{th:positivitefluxepapet} and \ref{th:positivitefluxh} ensure that both the following terms 
$\frac{\prh}{\Mh\temph} \pscal{\Ah\Xh,\Xh}$ 
and
$\frac{\pre}{\Mh\tempe} \pscal{\Aepa\Xepa,\Xepa}$
are nonnegative.
Finally, introducing 
$\Y = \depe\ptens\VarphieDede - \glogTepe\ptens\Varphielambdade$,
the last term is expanded as
\begin{equation} \label{eqprodentropycomplicated}
\pscal{\Aehpe\Xepe,\Xepe} = \crochete{\Y,\Y}
- \epsilon \crocheth{\YZ,\YZ},
\end{equation}
with
\begin{equation*}
\YZ_i = \tfrac{1}{3} \pre\temph\Mh \sumi[j] \ni[j]
\dblparenthesee{\Y,\VarphieDjde} \ptens \PhiiDj{i}{j},
\qquad i\in\lourd.
\end{equation*}
We conclude after noticing that
the classical term $\crochete{\Y,\Y}$ is nonnegative and vanishes if and only if $\Y=0$ thanks to the scalar constraints \eqref{eqconstvarphie1} and \eqref{eqconstvarphie2}.
\end{proof}

In the whole general case, we are not able to write the entropy production rate as a sum of nonnegative contributions independently of the value of $\epsilon$. However, it is the case for vanishing magnetic field, the last term $\pscal{\Aehpe\Xepe,\Xepe}$ given in \eqref{eqprodentropycomplicated} being nonnegative as soon as the magnetic field disappears---the matrix $\Aehpe$ reduces then to $\Aepe$---and that for any value of $\epsilon$.

The nonnegativity of the global entropy production rate implies that the second principle of thermodynamics is satisfied. This statement could be equivalently formulated  by means of a H-Theorem. Besides, the electron and heavy-particle temperatures must be equal when an equilibrium state is reached. Provided that the collision frequency $\colfreqiez$, $i\in\lourd$, is positive, the quasi-equilibrium states described by the Maxwell-Boltzmann distribution functions given in eqs.~\eqref{eqbolte} and \eqref{eqbolti} create some nonnegative entropy expressed by the term $(\tempe-\temph)^2/(\tempe\temph)\sum_{j\in\lourd}{\ni[j]}{\colfreqiez[j]}/{\mi[j]}$. The latter term vanishes when the electron and heavy-particle temperatures are identical.

\subsection{Plasma magnetization}

We recall that the intensity of the magnetic field is expressed by means of the $b$ parameter used to define the scaling of the nondimensional electron Hall parameter $q^0 B^0 t_\elec^0/\me^0=\varepsilon^{1-b}$. Three categories of plasmas are reviewed in Table~\ref{tab3}. A value of $b<0$ corresponds to unmagnetized plasmas, $b=0$, weakly magnetized plasmas, and $b=1$, strongly magnetized plasmas.

\begin{table}[ht]
\caption{Magnetic field influence.\label{tab3}}
{\small\begin{tabular}{@{}cll@{}} \toprule
$b$&{\bf Conservation equations}&{\bf Transport properties}\\
&&\\
$<0$&$-$&$-$\\
&&\\
0&Bulk magnetic force&  Electron bulk magnetic driving force\\
&Electron magnetic force&\\
&&\\
1&Bulk magnetic force&Electron bulk magnetic driving force\\
&Electron magnetic force& Heavy-particle bulk magnetic driving forces\\
&Heavy-particle magnetic force&Anisotropic electron transport properties\\
 \botrule
\end{tabular}}
\end{table}

\subsection{Mathematical structure}
The system of mass, momentum, total energy, and entropy eqs.~\eqref{eqsum1}, \eqref{eqsum2}, \eqref{eqsum3}, \eqref{eqsum5ter}, and \eqref{eqsum6} is conservative from a fluid standpoint in the variables
$$\U=[\rhoe,~(\rhoi)_{i\in \lourd}, ~\rhoi[\heavy]\vitesse, ~\energie,~\rho s]^T,$$
that reads
\begin{equation}\label{eqsyscons}
\dt\U+\dx\dscal\Fc+\dx\dscal\Fd=\Fs,
\end{equation}
with the convective fluxes
$$\Fc=[\rhoe\vitesse,~(\rhoi)_{i\in\lourd}\vitesse,~\rhoH\vitesse\ptens\vitesse+\frac{1}{\Mh^2} p\identite,~\enthalpie\vitesse,~\rho\entropie\vitesse]^T,$$
the diffusive fluxes 
$$\Fd=[\frac{\rhoe}{\Mh}(\Ve+\varepsilon\Vde),\frac{\varepsilon}{\Mh}(\rhoi\Vi)_{i\in\lourd},~\frac{\varepsilon}{\Mh^2}\visqueux[\heavy],~\frac{\varepsilon}{\Mh^2}\visqueux[\heavy]\dscal\vitesse+\frac{1}{\Mh}\fluxchaleur,~\entropyf]^T,$$
and the source terms
$$\Fs=[0,0,\frac{nq}{\Mh^2}\E \\
+(\delta_{b0}\courantel_0+\delta_{b1}\courantel)\pvect\B,~\courantel\dscal\E,~\Upsilon]^T.$$

Then, we extract a purely convective system from eq.~\eqref{eqsyscons}
\begin{equation}\label{eqsysconcon}
\dt\U+\dx\dscal\Fc=\Fs',
\end{equation}
where the convective source terms are given by
$$
\Fs'=[0,~0,~\frac{nq}{\Mh^2}\E 
+(\delta_{b0}+\delta_{b1})\courantel'\pvect\B,~\courantel'\dscal\E,~\Upsilon']^T,
$$
with the current $\courantel'=n q\vitesse$,
and the entropy production rate
$$\Upsilon'=\frac{(\tempe-\temph)^2}{\tempe\temph}\sum_{j\in\lourd}\frac{\ni[j]}{\mi[j]}{\colfreqiez[j]}.$$
The purely convective system  given in eq.~\eqref{eqsysconcon} is rewritten in a quasi-linear form
\begin{eqnarray}\label{eqsysql}
\dt\W+\Aw\dscal \dx\W&=&\Fs_\W',
\end{eqnarray}
by means of the variables
$$\W=[\rhoe,~(\rhoi)_{i\in \lourd}, ~\vitesse, ~\pre,~\prh]^T,$$
the source terms
$$\Fs_\W'=[0,~0,~\frac{nq}{\Mh^2\rho_\heavy}\E 
+\frac{1}{\rho_\heavy}(\delta_{b0}+\delta_{b1})\courantel'\pvect\B,~\tfrac{2}{3}\deltaEeo,~\tfrac{2}{3}\deltaEho]^T,$$
and the Jacobian matrices 
\begin{equation}\label{eqjaco}
\Anu= \begin{pmatrix}
\vhnu & 0& \rhoe\enu^T&0&0\\ 
  0 & \vhnu(\delta_{ij})_{i,j\in\lourd}& (\rhoi)_{i\in\lourd}\enu^T&0&0\\
0& 0& \vhnu\identite&\frac{1}{\Mh^2\rhoi[\heavy]}\enu&\frac{1}{\Mh^2\rhoi[\heavy]}\enu\\
0& 0& \gam\pre\enu^T&\vhnu&0\\  
0& 0& \gam\prh\enu^T&0&\vhnu\\  
\end{pmatrix},\quad\nu\in\{1,2,3\},
\end{equation}
where the specific heat ratio reads $\gam=5/3$ and symbol $\enu$ stands for the unit vector in the $\nu$ direction. 

For any direction defined by the unit vector $\no$, the matrix $\no\dscal \Aw$ is shown to be diagonalizable with real eigenvalues and a complete set of eigenvectors. There are two nonlinear acoustic fields with the eigenvalues $\vitesse\dscal \no\pm c$, where the sound speed is given by $c^2=p/(\rhoi[\heavy]\Mh^2)$, and linearly degenerate fields with the eigenvalue $\vitesse\dscal \no$ of multiplicity $\ns+3$.
Thus,  the macroscopic system of conservation equations 
derived from kinetic theory in the proposed mixed hyperbolic-parabolic scaling  
has a hyperbolic structure, as far as the convective part of the system is concerned. 
Such a property is far from being obvious since the obtained sound speed involves the electron pressure
and seeing that the rigourous derivation of the momentum equation of the heavy particles
involves many ingredients throughout the paper.

\section{Conclusions}
In the present contribution, we have derived from kinetic theory a unified fluid model for multicomponent plasmas by accounting for the electromagnetic field influence,  neglecting the particle internal energy and the reactive collisions. Given the strong disparity of mass between the electrons and heavy particles,  such as molecules, atoms, and ions, we have conducted a dimensional analysis of the Boltzmann equation  following Petit and Darrozes~\cite{petit} and introduced a scaling based on the $\varepsilon$ paramter, or square root of the ratio of the electron mass to a characteristic heavy-particle mass. The multiscale analysis occurs at three levels: in the kinetic equations, the collisional invariants, and the collision operators. The Boltzmann equation has been expressed in the mean heavy-particle velocity referential to allow for the resolubility of the first- and second-order electron perturbation  function equations, as opposed to the inertial referential chosen by Degond and Lucquin~\cite{degond1,degond2}. Then, the resolubility of the electron and heavy-particle perturbation functions has been classically based on the identification of the kernel of the linearized collision operators and space of  scalar collisional invariants of both types of species. The system has been examined at successive orders of approximation by means of a  generalized Chapman-Enskog method. The micro- and macroscopic equations derived at each order are reviewed in Table~\ref{tab12}. 
Depending on the type of  species, the quasi-equilibrium solutions are Maxwell-Boltzmann velocity distribution functions at the electron temperature or the heavy-particule temperature, therefore, allowing for thermal nonequilibrium to occur.
At order $\varepsilon^1$, the set of macroscopic conservation equations of mass, momentum, and energy comprises multicomponent Navier-Stokes equations for the heavy particles, which 
follow a hyperbolic scaling, and  first-order drift-diffusion equations 
for the electrons, which follow a parabolic scaling. The expressions of the transport fluxes have also been derived: first- and second-order diffusion velocity and heat flux for the electrons, and first-order diffusion velocities, heat flux, and viscous tensor for the heavy particles. The transport coefficients have been written in terms of bracket operators; both electron and heavy-particle transport coefficients  exhibit anisotropy, provided that the magnetic field is strong. 
  We have also proposed a complete description of the Kolesnikov effect, $i.e.$, the crossed contributions to the mass and energy transport fluxes coupling the electrons and heavy particles. 
This effect, appearing in multicomponent plasmas, is essential to obtain a positive entropy production. Besides, it contains, as  degenerate case, the single heavy-species plasmas considered by Degond and Lucquin for which the Kolesnikov effect is not present.
The properties of electron and heavy-particle mass-energy transport matrices have been established by using the mathematical structure of the bracket operators. In particular, the properties of symmetry and positivity implies that the second principle of thermodynamics is satisfied, as shown by deriving an entropy equation. Moreover, the first principle of thermodynamic was also verified by deriving a total energy equation.
Finally, the system of equations was found to be conservative and the purely 
convective system hyperbolic, thus leading to a well defined structure.

The proposed formalism remains valid for  collision operators of Fokker-Planck-Landau type. These operators can be used to model the charged particle interaction, instead of  Bolztmann operators associated with a Coulomb potential screened at the Debye distance. Besides, the explicit expression of the diffusion coefficients, thermal diffusion coefficients,  viscosity,  and partial thermal conductivities can be obtained by means of a variational procedure to solve the integral equations (Galerkin spectral method~\cite{chapman}). The expressions of the thermal conductivity, thermal diffusion ratios, and Stefan-Maxwell equations for the diffusion velocities can be derived by means of a Goldstein expansion of the perturbation function, as proposed by Kolesnikov and Tirskiy~\cite{anatoliy}.  Finally, the mathematical structure of the transport matrices obtained by the variational procedure can readily be used to build efficient transport algorithms, as already shown by Ern and Giovangigli~\cite{ern} for neutral gases, or Magin and Degrez~\cite{magin2} for unmagnetized plasmas.
 
\section*{Acknowledgment}   
The authors would like to thank Dr. Anne Bourdon for helpful discussion. The authors acknowledge support from two scientific departments of CNRS (French Center for Scientific Research): MPPU (Math\'ematiques, Physique, Plan\`ete et Univers) and ST2I (Sciences et Technologies de l'Information et de l'Ing\'enierie), through a 2007-2008 PEPS project entitled: ``\textit{Analyse et simulation de probl\`emes multi-\'echelles : applications aux plasmas froids, \`a la combustion et aux  \'ecoulements diphasiques}".

\end{document}